\documentclass[a4paper, 10pt]{article}

\usepackage{graphics,graphicx}
\usepackage{amssymb,amsmath}
\usepackage{enumerate}

\usepackage{color}

\newtheorem{theorem}{Theorem}[section]

\newtheorem{cor}{Corollary}[section]
\newtheorem{lemma}{Lemma}[section]

 \def \sm {\setminus}
 \def \es {\emptyset}

\newenvironment{proof}[1][]%
{\noindent {\setcounter{equation}{0}\it Proof.
}{#1}{}}{\hfill$\Box$\vspace{2ex}}

\def\longbox#1{\parbox{0.85\textwidth}{#1}}

\begin{document}

\title{Square-free graphs with no six-vertex induced~path}

\author{T. Karthick\thanks{Computer Science Unit, Indian Statistical
Institute, Chennai Centre, Chennai 600029, India.}
\and%
Fr\'ed\'eric Maffray\thanks{CNRS, Laboratoire G-SCOP,
Univ.~Grenoble-Alpes, Grenoble, France. Deceased on August 22, 2018.}}

\date{\today}

\maketitle

\begin{abstract}
We elucidate the structure of $(P_6,C_4)$-free graphs by showing that
every such graph either has a clique cutset, or a universal vertex, or
belongs to several special classes of graphs.  Using this result, we show that for any
$(P_6,C_4)$-free graph $G$,  $\lceil\frac{5\omega(G)}{4}\rceil$ and $\lceil\frac{\Delta(G) +
\omega(G) +1}{2}\rceil$ are tight upper bounds for the chromatic number of
$G$.  Moreover, our structural results imply that every ($P_6$,$C_4$)-free graph with no clique cutset has bounded clique-width,  and thus the existence of a polynomial-time algorithm that computes the chromatic number (or stability number) of  any $(P_6,C_4)$-free graph.

\end{abstract}

\section{Introduction}\label{sec:intro}

All our graphs are finite and have no loops or multiple edges.  For
any integer $k$, a \emph{$k$-coloring} of a graph $G$ is a mapping
$c:V(G)\rightarrow\{1,\ldots,k\}$ such that any two adjacent vertices
$u,v$ in $G$ satisfy $c(u)\neq c(v)$.  A graph is \emph{$k$-colorable}
if it admits a $k$-coloring.  The \emph{chromatic number} $\chi(G)$ of
a graph $G$ is the smallest integer $k$ such that $G$ is
$k$-colorable.  In general, determining whether a graph is
$k$-colorable or not is well-known to be $NP$-complete for every fixed
$k\ge 3$.  Thus designing algorithms for computing the chromatic
number by putting restrictions on the input graph and obtaining bounds
for the chromatic number are of interest.

A \emph{clique} in a graph $G$ is a set of pairwise adjacent vertices.
Let $\omega(G)$ denote the maximum clique size in a graph $G$.
Clearly $\chi(H)\ge \omega(H)$ for every induced subgraph $H$ of $G$.
A graph $G$ is \emph{perfect} if every induced subgraph $H$ of $G$
satisfies $\chi(H) = \omega(H)$.  The existence of triangle-free
graphs with aribtrarily large chromatic number shows that for general
graphs the chromatic number cannot be upper bounded by a function of
the clique number.  However, for restricted classes of graphs such a
function may exist.  Gy\'arf\'as~\cite{Gyarfas} called such classes of
graphs \emph{$\chi$-bounded} classes.  A family of graphs $\cal{G}$ is
$\chi$-bounded with $\chi$-bounding function $f$ if, for every induced
subgraph $H$ of $G\in \cal{G}$, $\chi(H)\le f(\omega(H))$.  For
instance, the class of perfect graphs is $\chi$-bounded with
$f(\omega)= \omega$.

Given a family of graphs ${\cal F}$, a graph $G$ is \emph{${\cal
F}$-free} if no induced subgraph of $G$ is isomorphic to a member of
${\cal F}$; when ${\cal F}$ has only one element $F$ we say that $G$
is $F$-free.  Several classes of graphs defined by forbidding certain
families of graphs were shown to be $\chi$-bounded: even-hole-free
graphs \cite{ACHRS-evenhole}; odd-hole-free graphs \cite{SS};
quasi-line graphs \cite{CAO-quasiline}; claw-free graphs with
stability number at least 3 \cite{CP-claw}; see also \cite {CCH, CKS,
CPST, KP, KZ} for more instances.

For any integer $\ell$ we let $P_\ell$ denote the path on $\ell$
vertices and $C_\ell$ denote the cycle on $\ell$ vertices.  A cycle on
$4$ vertices is referred to as a \emph{square}.  It is well known that
every $P_4$-free graph is perfect.  Gy\'arf\'as~\cite{Gyarfas} showed
that the class of $P_k$-free graphs is $\chi$-bounded.  Gravier et
al.~\cite{GHM} improved Gy\'arf\'as's bound slightly by showing that
every $P_k$-free graph $G$ satisfies $\chi(G) \le
(k-2)^{\omega(G)-1}$.  In particular every $P_6$-free graph $G$
satisfies $\chi(G) \le 4^{\omega(G)-1}$.  Improving this exponential
bound seems to be a difficult open problem.  In fact the problem of
determining whether the class of $P_5$-free graphs admits a polynomial
$\chi$-bounding function remains open, and the known $\chi$-bounding
function $f$ for such class of graphs satisfies $c(\omega^2/\log w)\le
f(\omega)\le 2^{\omega}$ \cite{KPT-P5}.  So the recent focus is on
obtaining (linear) $\chi$-bounding functions for some classes of
$P_t$-free graphs, where $t\ge 5$.  It is shown in \cite{CKS} that every $(P_5,C_4)$-free graph $G$ satisfies
$\chi(G)\le \lceil\frac{5\omega(G)}{4}\rceil$, and in \cite{CK} that every $(P_2\cup P_3,C_4)$-free graph $G$ satisfies
$\chi(G)\le \lceil\frac{5\omega(G)}{4}\rceil$.  Gaspers and Huang \cite{GH}
studied the class of $(P_6, C_4)$-free graphs (which  generalizes the class of $(P_5,C_4)$-free graphs and the class of $(P_2\cup P_3,C_4)$-free graphs)    and showed that every
such graph $G$ satisfies $\chi(G)\le \frac{3\omega(G)}{2}$.  We
improve their result and establish the best possible bound, as
follows.
\begin{theorem}\label{thm:54bound}
Let $G$ be any $(P_6,C_4)$-free graph.  Then $\chi(G)\le
\lceil\frac{5\omega(G)}{4}\rceil$.  Moreover, this bound is tight.
\end{theorem}

\medskip

The degree of a vertex in $G$ is the number of vertices adjacent to
it.  The maximum degree over all vertices in $G$ is denoted by
$\Delta(G)$.  For any graph $G$, we have $\chi(G)\le \Delta(G)+1$.
Brooks~\cite{Brooks} showed that if $G$ is a graph with $\Delta(G)\ge
3$ and $\omega(G)\le \Delta(G)$, then $\chi(G)\le \Delta(G)$.
Reed~\cite{Reed} conjectured that every graph $G$ satisfies $\chi(G)
\leq \lceil\frac{\Delta(G) + \omega(G) +1}{2}\rceil$.  Despite several
partial results \cite{King,Rabern,Reed}, Reed's conjecture is still
open in general, even for triangle-free graphs.  Using
Theorem~\ref{thm:54bound}, we will show that Reed's conjecture holds
for the class of ($P_6$,$C_4$)-free graphs:
\begin{theorem}\label{thm:reeds}
If $G$ is a $(P_6,C_4)$-free graph, then $\chi(G) \leq
\lceil\frac{\Delta(G) + \omega(G) +1}{2}\rceil$.
\end{theorem}
One can readily see that the bounds in Theorem~\ref{thm:54bound} and
in Theorem~\ref{thm:reeds} are tight on the following example.  Let
$G$ be a graph whose vertex-set is partitioned into five cliques $Q_1,
\ldots, Q_5$ such that for each $i\bmod 5$, every vertex in $Q_i$ is
adjacent to every vertex in $Q_{i+1}\cup Q_{i-1}$ and to no vertex in
$Q_{i+2}\cup Q_{i-2}$, and $|Q_i|=q$ for all $i$ ($q>0$).  Clearly $\omega(G)=2q$ and $\Delta(G)=3q-1$.
Since $G$ has no stable set of size~$3$, $G$ is $P_6$-free and  $\chi(G)\ge
\lceil\frac{5q}{2}\rceil$. Moreover, since no two non-adjacent vertices in $G$ has a common neighbor in $G$, we also see that $G$ is $C_4$-free.

\medskip

Finally, we   also have the following result.
\begin{theorem}\label{thm:algo}
There is a polynomial-time algorithm which computes the chromatic
number of any $(P_6,C_4)$-free graph.
\end{theorem}
The proof of Theorem~\ref{thm:algo} is based on the concept of
clique-width of a graph $G$, which was defined in \cite{CER} as the
minimum number of labels which are necessary to generate $G$ using a
certain type of operations.  (We omit the details.)  It is known from
\cite{KobRot,Rao} that if a class of graphs has bounded clique-width,
then there is a polynomial-time algorithm that computes the chromatic
number of every graph in this class.  We are able to prove that every
$(P_6, C_4)$-free graph that has no clique cutset has clique-width at
most~$36$, which implies the validity of Theorem~\ref{thm:algo}.
However a similar result, using similar techniques, was proved by
Gaspers, Huang and Paulusma \cite{GHP}.  Hence we refer to \cite{GHP},
or to the extended version of our manuscript \cite{KM-arxiv} for the
detailed proof of Theorem~\ref{thm:algo}.

We finish on this theme by noting that the class of $(P_6, C_4)$-free
graph itself does not have bounded clique-width, since the class of
split graphs (which are all $(P_6, C_4)$-free) does not have bounded
clique-width \cite{BDHP,MR}.  The clique-width argument might also be
used for solving other optimization problems in ($P_6, C_4)$-free
graphs, in particular the stability number.  However this problem was
solved earlier by Mosca \cite{Mos}, and the weighted version was
solved in \cite{BH}, and both algorithms have reasonably low
complexity.

\bigskip

Theorems~\ref{thm:54bound} and \ref{thm:reeds} will be derived from
the structural theorem below (Theorem~\ref{thm:struc0}).  Before
stating it we recall some definitions.

\bigskip

In a graph $G$, the \emph{neighborhood} of a vertex $x$ is the set
$N_G(x)=\{y\in V(G)\setminus x\mid xy\in E(G)\}$; we drop the
subscript $G$ when there is no ambiguity.  The \emph{closed}
neighborhood is the set $N[x]=N(x)\cup\{x\}$.  Two vertices $x,y$ are
\emph{clones} if $N[x]=N[y]$.  For any $x\in V(G)$ and $A\subseteq
V(G)\setminus x$, we let $N_A(x) = N(x)\cap A$.  For any two subsets
$X$ and $Y$ of $V(G)$, we denote by $[X,Y]$, the set of edges that has
one end in $X$ and other end in $Y$.  We say that $X$ is
\emph{complete} to $Y$ or $[X,Y]$ is complete if every vertex in $X$
is adjacent to every vertex in $Y$; and $X$ is \emph{anticomplete} to
$Y$ if $[X,Y]=\emptyset$.  If $X$ is singleton, say $\{v\}$, we simply
write $v$ is complete (anticomplete) to $Y$ instead of writing $\{v\}$
is complete (anticomplete) to $Y$.  If $S \subseteq V(G)$, then $G[S]$
denote the subgraph induced by $S$ in $G$.  A vertex is
\emph{universal} if it is adjacent to all other vertices.  A
\emph{stable set} is a set of pairwise non-adjacent vertices.  A
\emph{clique-cutset} of a graph $G$ is a clique $K$ in $G$ such that
$G \setminus K$ has more connected components than $G$.  A
\emph{matching} is a set of pairwise non-adjacent edges.  The
\emph{union} of two vertex-disjoint graphs $G$ and $H$ is the graph
with vertex-set $V(G)\cup V(H)$ and edge-set $E(G)\cup E(H)$.  The
union of $k$ copies of the same graph $G$ will be denoted by $kG$; for
example $2P_3$ denotes the graph that consists in two disjoint copies
of $P_3$.

A vertex is \emph{simplicial} if its neighborhood is a clique.  It is
easy to see that in any graph $G$ that has a simplicial vertex,
letting $S$ denote the set of simplicial vertices, every component of
$G[S]$ is a clique, and any two adjacent simplicial vertices are
clones.

A \emph{hole} is an induced cycle of length at least $4$.  A graph is
\emph{chordal} if it contains no hole as an induced subgraph.  Chordal
graphs have many interesting properties (see e.g.~\cite{Gol80}), in
particular: every chordal graph has a simplicial vertex; every chordal
graph that is not a clique has a clique-cutset; and every chordal
graph that is not a clique has two non-adjacent simplicial vertices.

In a graph $G$, let $A,B$ be disjoint subsets of $V(G)$.  It is easy
to see that the following two conditions (i) and (ii) are equivalent:
(i) any two vertices $a,a'\in A$ satisfy either $N_B(a)\subseteq
N_B(a')$ or $N_B(a')\subseteq N_B(a)$; (ii) any two vertices $b,b'\in
B$ satisfy either $N_A(b)\subseteq N_A(b')$ or $N_A(b')\subseteq
N_A(b)$.  If this condition holds we say that the pair $\{A,B\}$ is
\emph{graded}. 
 Clearly in a $C_4$-free graph any two
disjoint cliques form a graded pair. See also Lemma~\ref{lem:ab}
below.

\paragraph{Some special graphs}

Let $F_1,F_2,F_3$ be three graphs (as in \cite{GH}), as shown in
Figure~\ref{fig:f123}.

\begin{figure}[h]
\centering
 \includegraphics{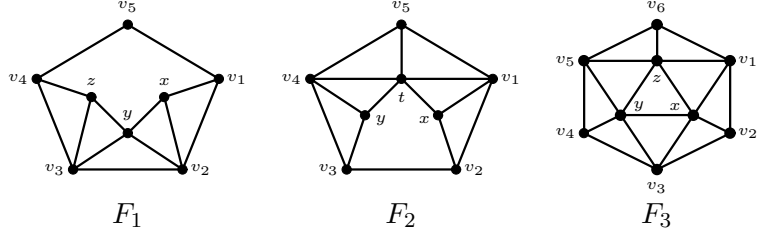}
\caption{$F_1$, $F_2$, $F_3$}\label{fig:f123}
\end{figure}

Let $H_1, H_2, H_3, H_4, H_5$ be five graphs, as shown in
Figure~\ref{fig:h12345}, where $H_1$ is the Petersen graph.
\begin{figure}[h]
\centering
 \includegraphics[width=11.75cm]{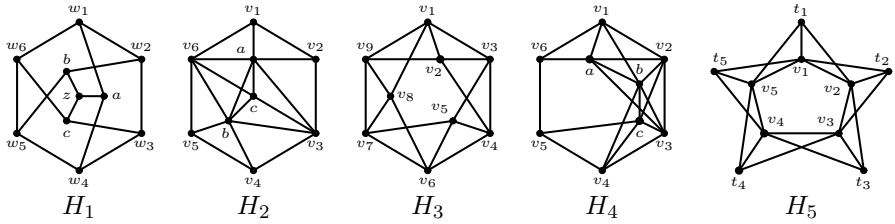}
\caption{$H_1$, $H_2$, $H_3$, $H_4$, $H_5$}\label{fig:h12345}
\label{fig:2}
\end{figure}

\begin{figure}[h]
\centering
 \includegraphics{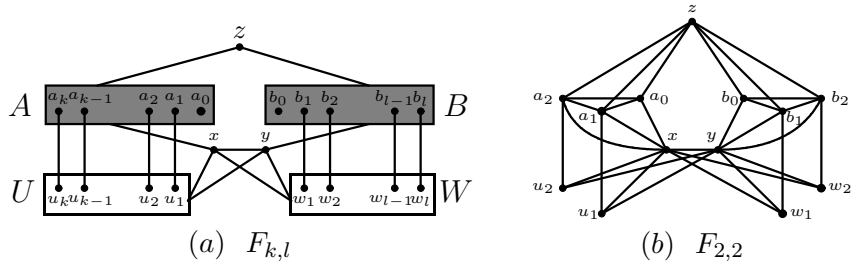}
\caption{(a) Schematic representation of the graph $F_{k,l}$. Here,   the vertices in a shaded box form a clique, and an edge between a vertex and a box indicates that the vertex is adjacent to all the vertices in the box. For example, the vertex $x$ is adjacent to all the  vertices in the boxes $A$, $U$,  and  $W$. (b) $F_{2,2}$.}\label{fig:fkl}
\end{figure}

\paragraph{Graphs $F_{k,\ell}$}

For integers $k,\ell\ge 0$ let $F_{k,\ell}$ be the graph whose
vertex-set can be partitioned into sets $A, B, U, W$ and $\{x,y,z\}$
such that:
\begin{itemize}\setlength\itemsep{0pt}
\item
$A =\{a_0,a_1,\ldots,a_k\}$ is a clique of size $k+1$, and $U =
\{u_1,u_2,\ldots,u_k\}$ is a stable set of size $k$, and the edges
between $A$ and $U$ form a matching of size $k$, namely,
$[A,U]=\{a_iu_i\mid i\in \{1,\ldots,k\}\}$;
\item
$B=\{b_0, b_1, \ldots, b_\ell\}$ is a clique of size $\ell+1$, and
$W=\{w_1, \ldots, w_\ell\}$ is a stable set of size $\ell$, and the
edges between $B$ and $W$ form a matching of size $\ell$, namely,
$[B,W]=\{b_jw_j\mid j\in\{1,\ldots,\ell\}\}$;
\item
The neighborhood of $x$ is $A\cup U\cup W\cup\{y\}$;
\item
The neighborhood of $y$ is $B\cup U\cup W\cup\{x\}$;
\item
The neighborhood of $z$ is $A\cup B$.
\end{itemize}

See Figure~\ref{fig:fkl} for the schematic representation of the graph $F_{k,l}$ and for the graph $F_{2,2}$.

\paragraph{Blowups}

A \emph{blowup} of a graph $H$ is any graph $G$ such that $V(G)$ can
be partitioned into $|V(H)|$ (not necessarily non-empty) cliques
$Q_v$, $v\in V(H)$, such that $[Q_u,Q_v]$ is complete if $uv\in E(H)$,
and $[Q_u,Q_v]=\es$ if $uv\notin E(H)$. See Figure~\ref{fig:bbb}:(a) for a blowup of a $C_5$.

\begin{figure}[h]
\centering
 \includegraphics[width=12cm]{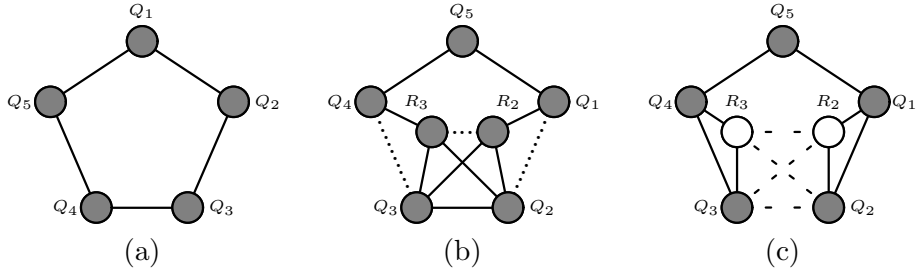}
\caption{Schematic representations of: (a) a blowup of a $C_5$, (b) a band, and (c) a belt.
  In (a), (b) and (c), the circles represent a collection of sets into which the vertex set of the graph is partitioned. Each shaded circle represents a nonempty clique,  a solid line between two circles indicates
that the two sets are complete to each other, and the absence of a line between two circles indicates that the
two sets are anticomplete to each other. In (b), a dotted line between two circles means that the respective pair of sets is graded. For example, the pair $\{Q_3, Q_4\}$ is graded. In (c),  the dashed lines between the sets $R_2, R_3, Q_2$ and $Q_3$ mean that the adjacency between these sets are subject to the fourth item of the definition of a belt.}\label{fig:bbb}
\end{figure}

\paragraph{Bands}

A \emph{band} is any graph $G$ (see Figure~\ref{fig:bbb}:(b)) whose vertex-set can be partitioned
into seven sets $Q_1,\ldots,Q_5,R_2,R_3$ such that:
\begin{itemize}\setlength\itemsep{0pt}
\item
Each of $Q_1,\ldots,Q_5,R_2,R_3$ is a clique.
\item
The sets $[Q_5, Q_1\cup Q_4]$, $[R_2, Q_1\cup Q_2\cup Q_3]$, $[R_3,
Q_2\cup Q_3\cup Q_4]$ and $[Q_2, Q_3]$ are complete.
\item
The sets $[Q_1, Q_3\cup R_3\cup Q_4]$, $[Q_4, Q_1\cup Q_2\cup R_2]$
and $[Q_5, Q_2\cup R_2\cup Q_3\cup R_3]$ are empty.
\item
The pairs $\{Q_1, Q_2\}$, $\{Q_3, Q_4\}$ and $\{R_2, R_3\}$ are
graded.
\end{itemize}

\paragraph{Belts}

A \emph{belt} is any $(P_6,C_4, C_6)$-free graph $G$ (see Figure~\ref{fig:bbb}:(c)) whose vertex-set
can be partitioned into seven sets $Q_1, \ldots, Q_5, R_2, R_3$ such
that:
\begin{itemize}\setlength\itemsep{0pt}
\item
Each of $Q_1, \ldots, Q_5$ is a clique.
\item
The sets $[Q_1, Q_2\cup R_2\cup Q_5]$ and $[Q_4, Q_3\cup R_3\cup Q_5]$
are complete.
\item
The sets $[Q_1, Q_3\cup R_3\cup Q_4]$, $[Q_4, Q_2\cup R_2\cup Q_1]$,
$[Q_5, Q_2\cup R_2\cup Q_3\cup R_3]$ are empty.
\item
For each $j\in\{2,3\}$, $[Q_j, R_j]$ is complete, every vertex in
$Q_j\cup R_j$ has a neighbor in $Q_{5-j}\cup R_{5-j}$, and no vertex
of $R_j$ is universal in $G[R_j]$.
\end{itemize}

\begin{figure}[t]
\centering
 \includegraphics[height=5cm]{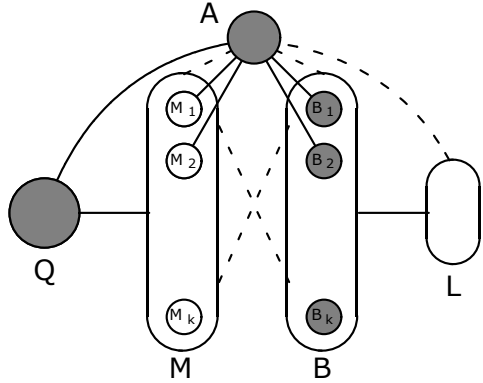}
\caption{Partial structure of a boiler. Here, each shaded circle represents a nonempty clique, and ovals labelled $M$ and $B$ represents the union of the sets represented by the circles inside that oval. The sets in oval $B$ forms a clique, and the ovals $M$ and $L$ induces a ($P_4,2P_3$)-free graph.  A solid line between two shapes indicates
that the respective sets are complete to each other. The absence of a line between any two shapes indicates that the
respective sets are anticomplete to each other.  A dashed line between any two shapes means that the adjacency between these sets are subject to the  definition of a boiler.}\label{fig:boiler}
\end{figure}

\paragraph{Boilers}
A \emph{boiler} is a $(P_6, C_4, C_6)$-free graph $G$  whose vertex-set
can be partitioned into five sets $Q, A, B, L, M$ such that:
\begin{itemize}\setlength\itemsep{0pt}
\item
The sets $Q$, $A$, $B$ and $M$ are non-empty, and $Q$, $A$ and $B$ are
cliques.
\item
The sets $[Q,A]$, $[Q,M]$, and $[B,L]$ are complete.
\item
The sets $[Q,B]$, $[Q, L]$ and $[L,M]$ are empty.
\item
$G[L]$ and $G[M]$ are $(P_4, 2P_3)$-free.
\item
Every vertex in $L$ has a neighbor in $A$.
\item
For some integer $k\ge 3$, $M$ is partitioned into $k$ non-empty sets
$M_1, \ldots,$ $M_k$, pairwise anticomplete, and $B$ is partitioned into
$k$ non-empty sets $B_1, \ldots, B_k$, such that for each
$i\in\{1,\ldots,k\}$ every vertex in $M_i$ has a neighbor in $B_i$ and no
neighbor in $B\setminus B_i$; and every vertex in $B$ has a neighbor
in $M$.
\item
$[A,M_1\cup B_1\cup M_2\cup B_2]$ is complete, and for each
$i\in\{3,\ldots,k\}$ every vertex in $A$ is either complete or
anticomplete to $M_i\cup B_i$, and no vertex in $A$ is complete to
$B$.
\end{itemize}

See Figure~\ref{fig:boiler} for the partial structure of a boiler.

\medskip
We consider that the definition of blowups (of certain fixed graphs)
and of bands (using Lemma~\ref{lem:ab}) is also a complete description
of the structure of such graphs.  However this is not so for belts and
boilers.  Such graphs have additional properties, and a
description of their structure is given in Section~\ref{sec:bb}.

\medskip

Now we can state our main structural result.  The existence of
such a decomposition theorem was inspired to us by the results from
\cite{GH} which go a long way in that direction.
\begin{theorem}\label{thm:struc0}
If $G$ is any $(P_6, C_4)$-free graph, then one of the following
holds:
\begin{itemize}\setlength\itemsep{0pt}
\item
$G$ has a clique cutset.
\item
$G$ has a universal vertex.
\item
$G$ is a blowup of either $H_1, \ldots,$ $H_5$, $F_3$ or $F_{k,\ell}$
(for some $k,\ell\ge 1$).
\item
$G$ is either a band, a belt, or a boiler.
\end{itemize}
\end{theorem}

Theorem~\ref{thm:struc0} is derived from Theorem~\ref{thm:structure}.
\begin{theorem}\label{thm:structure}
Let $G$ be a  $(P_6, C_4)$-free graph that has no clique-cutset and no
universal vertex.  Then the following hold:
\begin{enumerate}\setlength\itemsep{0pt}
\item
If $G$ contains an $F_3$, then $G$ is a blowup of $F_3$.
\item
If $G$ contains an $F_1$ and no $F_3$, then $G$ is a band.
\item\label{c6}
If $G$ is $F_1$-free, and $G$ contains an induced $C_6$, then $G$ is a
blowup of one of the graphs $H_1, H_2, H_3, H_4$.
\item
If $G$ is $C_6$-free, and $G$ contains an $F_2$, then $G$ is a blowup
of either $H_5$ or $F_{k,\ell}$ for some integers $k,\ell\ge 1$.
\item
If $G$ contains no $C_6$ and no $F_2$, and $G$ contains a $C_5$, then
$G$ is  either a belt or a boiler.
\end{enumerate}
\end{theorem}
\begin{proof}
The proof of each of these items is given below in
Theorems~\ref{lem:f3}, \ref{thm:f1nof3}, \ref{thm:c6nof1}, \ref{thm:f2}
and~\ref{thm:final} respectively.
\end{proof}

\noindent{\it Proof of Theorem~\ref{thm:struc0}, assuming
Theorem~\ref{thm:structure}.} \\
Let $G$ be any $(P_6, C_4)$-free graph.  If $G$ is chordal, then
either $G$ is a complete graph (so it has a universal vertex) or $G$
has a clique cutset.  Now suppose that $G$ is not chordal.  Then it
contains an induced cycle of length either $5$ or $6$.  So it
satisfies the hypothesis of one of the items of
Theorem~\ref{thm:structure} and consequently it satisfies the
conclusion of this item.  This established Theorem~\ref{thm:struc0}.
$\Box$

\section{Classes of square-free graphs}\label{sec:classes}

In this section, we study some classes of square-free graphs and prove
some useful lemmas and theorems that are needed for the later
sections.  We first note that any blowup of a $P_6$-free chordal graph is  $P_6$-free chordal.

\begin{lemma}\label{lem:simpath}
In a chordal graph $G$, every non-simplicial vertex lies on a chordless
path between two simplicial vertices.
\end{lemma}
\begin{proof}
Let $x$ be a non-simplicial vertex in $G$, so it has two non-adjacent
neighbors $y,z$.  If both $y,z$ are simplicial, then $y$-$x$-$z$ is
the desired path.  Hence assume that $y$ is non-simplicial.  Since $G$
is not a clique, it has two simplicial vertices, so it has a
simplicial vertex $s$ different from $z$.  So $s\notin\{y,z\}$.  In
$G\setminus s$, the vertex $x$ is non-simplicial, so, by induction,
there is a chordless path $P=p_0$-$p_1$-$\cdots$-$p_k$ in $G\setminus
s$, with $k\ge 2$, such that $p_0$ and $p_k$ are simplicial in
$G\setminus s$ and $x=p_i$ for some $i\in\{1,\ldots,k-1\}$.  If $p_0$
and $p_k$ are simplicial in $G$, then $P$ is the desired path.  So
suppose that $p_0$ is not simplicial in $G$, so $sp_0\in E(G)$.  Since
$s$ is simplicial in $G$ we have $N_P(s)\subseteq\{p_0, p_1\}$.  Then we see that either $s$-$p_0$-$p_1$-$\cdots$-$p_k$ or
$s$-$p_1$-$\cdots$-$p_k$ is the desired path.
\end{proof}

\begin{lemma}\label{lem:chxa}
In a chordal graph $G$, let $X$ and $A$ be disjoint subsets of $V(G)$
such that $A$ is a clique and every simplicial vertex of $G[X]$ has a
neighbor in $A$.  Then every vertex in $X$ has a neighbor in $A$.
\end{lemma}
\begin{proof}
Consider any non-simplicial vertex $x$ of $G[X]$.  By
Lemma~\ref{lem:simpath} there is a chordless path
$P=p_0$-$p_1$-$\cdots$-$p_k$ in $G[X]$, with $k\ge 2$, such that $p_0$
and $p_k$ are simplicial in $G[X]$ and $x=p_i$ for some
$i\in\{1,\ldots,k-1\}$.  By the hypothesis $p_0$ has neighbor $a\in A$
and $p_k$ has a neighbor $a'$ in $A$.  Suppose that $x$ has no
neighbor in $\{a,a'\}$.  Let $h$ be the largest integer in
$\{0,\ldots,i-1\}$ such that $p_h$ has a neighbor in $\{a,a'\}$, and
let $g$ be the smallest integer in $\{i+1,\ldots,k\}$ such that $p_g$
has a neighbor in $\{a,a'\}$.  Then $\{p_h, p_{h+1}, \ldots, p_g, a,
a'\}$ contains a hole, a contradiction.  So $x$ has a neighbor in $A$.
\end{proof}

%

\begin{lemma}\label{lem:ab}
In a $C_4$-free graph $G$, let $A,B$ be two disjoint cliques.  Then:
\begin{itemize}\setlength\itemsep{0pt}
\item
There is a labeling $a_1,\ldots,a_{|A|}$ of the vertices of $A$ such
that $N_B(a_1)\supseteq N_B(a_2)\supseteq\cdots\supseteq
N_B(a_{|A|})$.  Similarly, there is a labeling $b_1,\ldots,b_{|B|}$ of
the vertices of $B$ such that $N_A(b_1)\supseteq
N_A(b_2)\supseteq\cdots\supseteq N_A(b_{|B|})$.
\item
If every vertex in $A$ has a neighbor in $B$, then some vertex in $B$
is complete to $A$.
\item
If every vertex in $A$ has a non-neighbor in $B$, then some vertex in
$B$ is anticomplete to $A$.
\item
If $[A,B]$ is not complete, there are indices $i\le |A|$ and $j\le
|B|$ such $a_i b_j\notin E(G)$, and $a_ib_h\in E(G)$ for all $h<j$,
and $a_gb_j\in E(G)$ for all $g<i$.  Moreover, every maximal clique of
$G$ contains one of $a_i,b_j$.
\end{itemize}
\end{lemma}
\begin{proof}
Consider any two vertices $a,a'\in A$.  If there are vertices $b\in
N_B(a)\setminus N_B(a')$ and $b'\in N_B(a')\setminus N_B(a)$, then
$\{a,a',b,b'\}$ induces a $C_4$.  Hence we have either
$N_B(a)\subseteq N_B(a')$ or $N_B(a')\subseteq N_B(a)$.  This
inclusion relation for all $a,a'$ implies the existence of a total
ordering on $A$, which corresponds to a labeling as desired, and the
same holds for $B$.  This proves the first item of the lemma.  The
second and third item are immediate consequences of the first.

Now suppose that $A$ is not complete to $B$.  Consider any vertex
$a_{i'}\in A$ that has a non-neighbor in $B$, and let $j$ be the
smallest index such that $a_{i'}b_j\notin E(G)$.  Let $i$ be the
smallest index such that $a_ib_j\notin E(G)$.  So $i\le i'$.  We have
$a_gb_j\in E(G)$ for all $g<i$ by the choice of $i$.  We also have
$a_ib_h\in E(G)$ for all $h<j$, for otherwise, since $i\le i'$ we also
have $a_{i'}b_h\notin E(G)$, contradicting the definition of $j$.
This proves the first part of the fourth item.

Finally, consider any maximal clique $K$ of $G$.  Let $g$ be the
largest index such that $a_g\in K$ and let $h$ be the largest index
such that $b_h\in K$.  By the properties of the labelings and the
maximality of $K$ we have $K=\{a_1, \ldots, a_g\}\cup \{b_1, \ldots,
b_h\}$.  If both $g<i$ and $h<j$, then the properties of $a_i,b_j$
imply that $K\cup\{a_i\}$ (and also $K\cup\{b_j\}$) is a clique of
$G$, contradicting the maximality of $K$.  Hence we have either $g\ge
i$ or $h\ge j$, and so $K$ contains one of $a_i, b_j$.
\end{proof}

\begin{lemma}\label{lem:chxab}
In a $(P_6,C_4)$-free graph $G$, let $X$, $Y$ and $\{c\}$ be disjoint
subsets of $V(G)$ such that:
\begin{itemize}
\item
$Y$ is a clique, and every vertex in $X$ has a neighbor in $Y$,
\item
$c$ is complete to $X$ and anticomplete to $Y$;
\item
Either $G[X]$ is not connected, or there are vertices $c',c''\in
V(G)\setminus (X\cup Y)$ such that $c'$ is complete to $Y$ and
anticomplete to $X$, and $c''$ is anticomplete to $X\cup Y$, and
$c'c''\in E(G)$.
\end{itemize}
Then $G[X]$ is $(P_4, 2P_3)$-free.
\end{lemma}
\begin{proof}
First suppose that there is a $P_4$ $p_1$-$p_2$-$p_3$-$p_4$ in $G[X]$.
By the hypothesis $p_1$ has a neighbor $a\in Y$.  Then $ap_3\notin
E(G)$, for otherwise $\{p_1,a,p_3,c\}$ induces a $C_4$; and similarly
$ap_4\notin E(G)$.  If $G[X]$ is connected, then either
$p_3$-$p_2$-$p_1$-$a$-$c'$-$c''$ or $p_4$-$p_3$-$p_2$-$a$-$c'$-$c''$
is a $P_6$.  Now suppose that $G[X]$ is not connected.  So $X$
contains a vertex $p$ that is anticomplete to $\{p_1, p_2, p_3,
p_4\}$.  By the hypothesis $p$ has a neighbor $a'\in Y$.  As above we
have $ap\notin E(G)$ and $a'p_i\notin E(G)$ for all
$i\in\{1,\ldots,4\}$ for otherwise there is a $C_4$.  But then either
$p$-$a'$-$a$-$p_1$-$p_2$-$p_3$ or $p$-$a'$-$a$-$p_2$-$p_3$-$p_4$ is a
$P_6$.

Now suppose that there is a $2P_3$ in $G[X]$, with vertices $p_1,
\ldots, p_6$ and edges $p_1p_2, p_2p_3$, $p_4p_5, p_5p_6$.  We know
that $p_1$ has a neighbor $a\in Y$, and as above we have $ap_i\notin
E(G)$ for each $i\in\{3,4,5,6\}$, for otherwise there is a $C_4$.
Likewise, $p_6$ has a neighbor $a'\in Y$, and $a'p_j\notin E(G)$ for
each $j\in\{1,2,3,4\}$.  Then $p_{h+1}$-$p_h$-$a$-$a'$-$p_g$-$p_{g-1}$
is an induced $P_6$ for some $h\in\{1,2\}$ and $g\in\{5,6\}$.
\end{proof}


\paragraph{$(P_4,C_4)$-free graphs}

We want to understand the structure of $(P_4,C_4, 2P_3)$-free graphs
as they play a major role in the structure of belts and boilers.
Recall that $(P_4,C_4)$-free graphs were studied by Golumbic
\cite{Gol78}, who called them \emph{trivially perfect} graphs.
Clearly any such graph is chordal.  It was proved in \cite{Gol78} that
every connected $(P_4,C_4$)-free graph has a universal vertex.  It
follows that trivially perfect graphs are exactly the class $\cal T$
of graphs that can be built recursively as follows, starting from
complete graphs: \\
-- The disjoint union of any number of trivially perfect graphs is
trivially perfect; \\
-- If $G$ is any trivially perfect graph, then the graph obtained from
$G$ by adding a universal vertex is trivially perfect.

As a consequence, any connected member $G$ of $\cal T$ can be
represented by a rooted  directed  tree $T(G)$   defined as follows.  If
$G$ is a clique, let $T(G)$ have one node, which is the set $V(G)$.
If $G$ is not a clique, then by Golumbic's result the set $U(G)$ of
universal vertices of $G$ is not empty, and $G\setminus U(G)$ has a
number $k\ge 2$ of components $G_1, \ldots, G_k$.  Let then $T(G)$ be
the tree whose root is $U(G)$ and the children (out-neighbors) of
$U(G)$ are the roots of $T(G_1),\ldots,T(G_k)$.

The following properties of $T(G)$ appear immediately.  Every node of
$T(G)$ is a non-empty clique of $G$, and every vertex $v$ of $G$ is in
exactly one such clique, which we call $A_v$; moreover, $A_v$ is a
homogeneous set (all member of $A_v$ are pairwise clones).  For every
vertex $v$ of $G$, the closed neighborhood of $v$ consists of $A_v$
and all the vertices in the cliques that are descendants and ancestors of $A_v$ in
$T(G)$.  Every maximal clique of $G$ is the union of the nodes of a
directed path in $T(G)$.  All vertices in any leaf of $T(G)$ are
simplicial vertices of $G$, and every simplicial vertex of $G$ is in
some leaf of $T(G)$.

We say that a member $G$ of $\cal T$ is \emph{basic} if every node of
$T(G)$ is a clique of size~$1$.  (We can view $T(G)$ as a directed
tree, where every edge is directed away from the root; and then $G$ is
the underlying undirected graph of the transitive closure of $T(G)$.).
It follows that every member of $\cal T$ is a blowup of a basic member
of $\cal T$.  In a basic member $G$ of $\cal T$, two vertices are
adjacent if and only if one of them is an ancestor of the other in
$T(G)$, and every clique of $G$ consists of the set of vertices of any
directed path in $T(G)$.

\medskip

A \emph{dart} is the graph   with
vertex-set $\{a,b,c,d,e\}$ and edge-set $\{ab,bc,cd,da,$ $ac,ce\}$.   Let $K_{1,3}^+$
be the tree obtained from $K_{1,3}$ by subdividing one edge.   Next we give the following useful lemma.
\begin{lemma}\label{lem:3simp}
Let $G$ be a $(P_4,C_4)$-free graph. \\
(a) If $G$ does not have three pairwise non-adjacent simplicial vertices,
then $G$ is a blowup of $P_3$.  \\
(b) If $G$ does not have four pairwise non-adjacent simplicial
vertices, then $G$ is a blowup of a dart.
\end{lemma}
\begin{proof}
The hypothesis of (a) or (b) means that, if $H$ is a connected component of $G$,  then $T(H)$ is a tree with at most
three leaves.  Since each internal vertex of $T(H)$ has at least two leaves, $T(H)$ is either $K_1$, $K_2$, $P_3$ (rooted at its
vertex of degree~$2$), $K_{1,3}$ (rooted at its vertex of degree~$3$),
or $K_{1,3}^+$ (rooted at its vertex of degree~$2$).  Then the
conclusion follows directly from our assumption on $G$ and the preceding arguments.
\end{proof}

\paragraph{$(P_4,C_4, 2P_3)$-free graphs}

Let $\cal C$ be the class of $(P_4,C_4,2P_3)$-free graphs.  So $\cal
C\subset \cal T$.  If $G$ is any member of $\cal C$, and $G$ is
connected and not a clique, then since $G$ is $2P_3$-free all
components of $G\setminus U(G)$, except possibly one, are cliques.  So
all children of $U(G)$ in $T(G)$, except possibly one, are leaves.
Applying this argument recursively we see that the tree $T(G)$
consists of a rooted directed path plus a positive number of leaves
adjacent to every node of this path, with at least two leaves adjacent
to the last node of this path.  We call such a tree a \emph{bamboo}.
By the same argument as above, every member of $\cal C$ is a blowup of
a basic member of~$\cal C$.

\paragraph{$\cal C$-pairs}

A graph $G$ is a \emph{$\cal C$-pair} if $G$ is $P_6$-free, chordal,
and $V(G)$ can be partitioned into two sets $X$ and $A$ such that $A$
is a clique, $G[X]\in{\cal C}$, every vertex in $X$ has a neighbor in
$A$, and any two non-adjacent vertices in $X$ have no common neighbor
in $A$.  Depending on the context we may also write that $(X,A)$ is a
$\cal C$-pair.

We say that $G$ is a \emph{basic $\cal C$-pair} if the subgraph $G[X]$
is a basic member of $\cal C$, with vertices $x_1, \ldots, x_k$ for some
integer $k$, and a clique $A=\{a_0, a_1, \ldots, a_k\}$; and for each
$i\in\{1,\ldots,k\}$, if $x_i$ is simplicial in $G[X]$ then
$N_{A}(x_i)=\{a_i\}$, else $N_A(x_i)$ consists of $\{a_i\}$ plus the
union of $N_A(y)$ over all descendants $y$ of $x_i$ in $T(G[X])$.

Before describing how all $\cal C$-pairs can be obtained from basic
$\cal C$-pairs we need to introduce another definition.  Let $H$ be
any graph and $M$ be a matching in $H$.  An \emph{augmentation} of $H$
along $M$ is any graph $G$ whose vertex-set can be partitioned into
$|V(H)|$ cliques $Q_v$, $v\in V(H)$, such that $[Q_u,Q_v]$ is complete
if $uv\in E(H)\setminus M$, and $[Q_u,Q_v] =\es$ if $uv\notin E(H)$,
and $\{Q_u,Q_v\}$ is a graded pair if $uv\in M$.  (See \cite{MReed}
for a similar definition.)

In a basic $\cal C$-pair $G$, with the same notation as above, we say
that a matching $M$ is \emph{acceptable} if there is a clique
$\{x_{i_1}, \ldots, x_{i_h}\}$ in $G[X]$ such that
$M=\{x_{i_1}a_{i_1}, \ldots, x_{i_h}a_{i_h}\}$.

\begin{figure}[h]
\centering
 \includegraphics{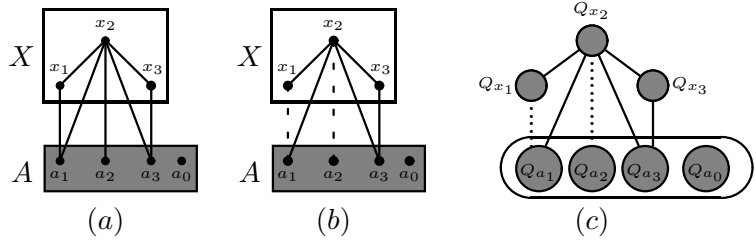}
\caption{Schematic representations of: (a) a basic $\cal C$-pair, (b) an acceptable matching in (a), and (c) an augmentation of the graph in (a) along an acceptable matching in (b).
In (a) and (b), the vertices in a shaded box represents a clique. In (b),  the dashed lines  represent the matching edges.  In (c), the circles represent a collection of sets into which the vertex set of the graph is partitioned, each shaded circle represents a clique, and the  circles inside the oval form a clique,  a solid line between two circles indicates
that the two sets are complete to each other,   the dotted line between two circles means that the respective pair of sets is graded,  and the absence of a line between two circles indicates that the
two sets are anticomplete to each other.}\label{fig:bcpair}
\end{figure}


\begin{theorem}\label{thm:cpair}
A graph is a $\cal C$-pair then it is an augmentation of a
basic $\cal C$-pair along an acceptable matching.
\end{theorem}
\begin{proof}
Let $G$ be any $\cal C$-pair,
with the same notation as above.  Since $G[X]$ is $(P_4, C_4,
2P_3)$-free it admits a representative tree $T(G[X])$ which is a
bamboo.  We claim that:
\begin{equation}\label{cp1}
\longbox{If $Y,Z$ are two nodes of $T(G[X])$ such that $Z$ is a
descendant of $Y$, then $Y$ is complete to $N_{A}(Z)$.}
\end{equation}
Proof: Consider any $y\in Y$ and $a\in N_{A}(Z)$; so there is a vertex
$z\in Z$ with $za\in E(G)$.  Since $Y$ is not a leaf of $T(G[X])$,
there is a child $Z'$ of $Y$ in $T(G[X])$ such that $Z'$ is not on the
directed path from $Z$ to $Y$, and so $Z$ and $Z'$ are not adjacent
(they are anticomplete to each other).  Pick any $z'\in Z'$.  Then
$yz, yz'\in E(G)$ and $zz'\notin E(G)$.  We know that $z'$ has a
neighbor $a'\in A$.  We have $az', a'z\notin E(G)$ by the definition
of a $\cal C$-pair ($z$ and $z'$ have no common neighbor in $A$).
Then $ya, ya'\in E(G)$, for otherwise $G[y,z,z',a,a']$ contains an
induced hole of length $4$ or $5$, contradicting the fact that $G$ is
chordal.  So (\ref{cp1}) holds.

\smallskip

Let $X_1, \ldots, X_k$ be the nodes of $T(G[X])$.  For each
$i\in\{1,\ldots,k\}$, let $U_i$ be the union of $N_{A}(Z)$ over all
descendants $Z$ of $X_i$ in $T(G[X])$, and let $A_i = N_{A}(X_i)
\setminus U_i$.  Let $A_0=A\setminus (A_1\cup\cdots\cup A_k)$ (so
$[X,A_0]=\emptyset$).

Let $X_{i_1}, \ldots, X_{i_h}$ be the nodes of $T(G[X])$ that are not
homogeneous in $G$ (if any).  Note that for each $i\in\{i_1, \ldots,
i_h\}$ the pair $\{X_i, A_i\}$ is graded since $G$ is $C_4$-free. 
     We
claim that:
\begin{equation}\label{cp2}
\mbox{$X_{i_1}\cup\cdots\cup X_{i_h}$ is a clique.}
\end{equation}
Proof: Suppose, on the contrary, and up to symmetry, that $[X_{i_1},
X_{i_2}]$ is not complete, and so $[X_{i_1}, X_{i_2}]=\emptyset$.  For
each $t\in\{1,2\}$, since $X_{i_t}$ is not homogeneous in $G$, there
are vertices $y_t,z_t\in X_{i_t}$ and a vertex $a_t\in A$ that is
adjacent to $y_t$ and not to $z_t$.  Since non-adjacent vertices in
$X$ have no common neighbor in $A$, we have $a_1\neq a_2$ and
$a_1y_2,a_1z_2, a_2y_1,a_2z_1\notin E(G)$.  Then
$z_1$-$y_1$-$a_1$-$a_2$-$y_2$-$z_2$ is a $P_6$.  So (\ref{cp2}) holds.

\smallskip

Let $H$ be the basic member of $\cal C$ of which $G[X]$ is a blowup.
Let $H$ have vertices $x_1, \ldots, x_k$, where $x_i$ corresponds to the
node $X_i$ of $T(G[X])$ for all~$i$.  Let $G_0$ be the graph obtained
from $H$ by adding a set $A=\{a_0, a_1, \ldots, a_k\}$, disjoint from
$V(H)$, and edges so that $A$ is a clique in $G_0$ and, for all
$i\in\{1,\ldots,k\}$ and $j\in\{0,1,\ldots,k\}$, vertices $x_i$ and $a_j$
are adjacent in $G_0$ if and only if $[X_i, A_j]\neq\emptyset$ in $G$.
By this construction and by (\ref{cp1}) $G_0$ is a basic $\cal
C$-pair.  In $G_0$ let $M=\{x_{i_1}a_{i_1}, \ldots, x_{i_h}a_{i_h}\}$.
It follows from (\ref{cp2}) that $M$ is an acceptable matching of
$G_0$ and from all the points above that $G$ is an augmentation of
$G_0$ along $M$.
\end{proof}

\section{Structure of ($P_6$, $C_4$)-free graphs}

In this section, we give the proof of Theorem~\ref{thm:structure}.  We
say that a subgraph $H$ of $G$ is \emph{dominating} if every vertex in
$V(G)\setminus V(H)$ is a adjacent to a vertex in $H$.  We will use
the following theorem of Brandst\"adt and Ho\`ang \cite{BH}.
\begin{theorem}[\cite{BH}]\label{thm:BH}
Let $G$ be a ($P_6, C_4$)-free graph that has no clique cutset.  Then
the following statements hold.  \\
(i) Every induced $C_5$ is dominating.  \\
(ii) If $G$ contains an induced $C_6$ which is not dominating, then
$G$ is the join of a complete graph and a blowup of the Petersen
graph.  \hfill{$\Box$}
\end{theorem}

In the next two theorems we make some general observations about the
situation when a $(P_6, C_4)$-free graph contains a hole (which must
have length either $5$ or $6$). Observe that in a $C_4$-free graph $G$, if $u$-$v$-$w$ is a $P_3$, then   any $x\in V(G)\sm \{u,v,w\}$ which is adjacent to $u$ and $w$ is also adjacent to $v$.
 
\begin{theorem}\label{thm:c5}
Let $G$ be any $(P_6, C_4)$-free graph that contains a $C_5$ with
vertex-set $C=\{v_1,\ldots,v_5\}$ and $\{v_iv_{i+1}\mid i\in
\{1,\ldots,5\}, i\bmod 5\}$.  Let:
\begin{eqnarray*}
A &=& \{x\in V(G)\sm C\mid N_C(x)=C\}. \\
T_i &=& \{x\in V(G)\sm C\mid N_C(x)=\{v_{i-1}, v_i, v_{i+1}\}. \\
W_i &=& \{x\in V(G)\sm C\mid N_C(x)=\{v_i\}. \\
X_{i,i+1} &=& \{x\in V(G)\sm C\mid N_C(x)=\{v_{i}, v_{i+1}\}.
\end{eqnarray*}
Moreover, let $T = T_1\cup\cdots\cup T_5$, $W =
W_1\cup\cdots\cup W_5$, and $X = X_{12}\cup X_{23}\cup X_{34}\cup X_{45}\cup X_{51}$. Then the following properties hold for all $i$:
\begin{enumerate}[(a)]
\item\label{c5-a}
$A\cup T_i$ is a clique.
\item\label{c5-b}
$[T_i, T_{i+2}]$, $[X_{i,i+1}, X_{i+2,i+3}]$, $[W_i, W_{i+1}]$, $[T_i,
W_{i-2}\cup W_{i+2}]$, $[T_i, X_{i+2,i+3}]$ and $[X_{i,i+1}, W_{i}\cup
W_{i+1}]$ are empty.
\item\label{c5-c}
$[X_{i,i+1}, X_{i+1,i+2}]$, $[W_i, W_{i+2}]$, and
$[X_{i,i+1}, W_{i-1}\cup W_{i+2}]$~are~complete.
\item\label{c5-d}
If $G$ is $C_6$-free, then for each $i$ one of $X_{i,i+1}$ and
$X_{i+1,i+2}$ is empty, and one of $W_i$ and $W_{i+2}$ is empty, and
one of $X_{i,i+1}$ and $W_{i-1}\cup W_{i+2}$ is empty.
\item\label{c5-e}
If $G$ has no clique cutset, then the set $\{x\in V(G)\sm C\mid N_C(x)=\emptyset\}$ is empty and $[T_i, W_i]$ is
complete.
\item\label{c5-f}
If $G$ has no clique cutset, then  $V(G)= V(C)\cup A\cup T\cup W\cup X$.
\end{enumerate}
\end{theorem}
\begin{proof}
(\ref{c5-a}) If there are non-adjacent vertices $a,b\in A\cup T_i$,
then $\{a,v_{i-1},b,v_{i+1}\}$ induces a $C_4$.

(\ref{c5-b}) Let $i=1$ and suppose that there is an edge $xy$ in one
of the listed sets.  If $x\in T_1$ and $y\in T_3$, then $\{x,y,v_4,
v_5\}$ induces a $C_4$.  If $x\in X_{12}$ and $y\in X_{34}$, then
$\{x,v_2,v_3,y\}$ induces a $C_4$.  If $x\in T_1$ and $y\in W_4$ then
$\{x,y,v_4,v_5\}$ induces a $C_4$.  If $x\in W_1$ and $y\in W_2$, then
$x$-$y$-$v_2$-$v_3$-$v_4$-$v_5$ is an induced $P_6$.  If $x\in T_1$
and $y\in X_{34}$, then $\{x,v_2,v_3,y\}$ induces a $C_4$.  If $x\in
X_{12}$ and $y\in W_1$, then $y$-$x$-$v_2$-$v_3$-$v_4$-$v_5$ is a
$P_6$.  The other cases are symmetric.

(\ref{c5-c}) and (\ref{c5-d}) Let $i=1$ and suppose that there are
vertices $x\in X_{12}\cup W_1$ and $y\in X_{23}\cup W_3$.  If
$xy\notin E(G)$, then $x$-$v_1$-$v_5$-$v_4$-$v_3$-$y$ is a $P_6$.
This proves (\ref{c5-c}).  If $xy\in E(G)$ then the same vertices
induce a $C_6$, which proves (\ref{c5-d}).

(\ref{c5-e}) Follows from Theorem~\ref{thm:BH}.

(\ref{c5-f}) Follows by Theorem~\ref{thm:BH} and (\ref{c5-e}).
\end{proof}

\begin{theorem}\label{thm:c6}
Let $G$ be any $(P_6, C_4)$-free graph that contains a $C_6$ with
vertex-set $C=\{v_1,\ldots,v_6\}$ and $\{v_iv_{i+1}\mid i\in
\{1,\ldots,6\}, i\bmod 6\}$.  Let:
\begin{eqnarray*}
S &=& \{x\in V(G)\sm C\mid N_C(x)=C\}. \\
A_i &=& \{x\in V(G)\sm C\mid N_C(x)=\{v_{i-1}, v_i, v_{i+1}\}\}. \\
B_i &=& \{x\in V(G)\sm C\mid N_C(x)=\{v_{i-1}, v_i, v_{i+1}, v_{i+2}\}. \\
D_i &=& \{x\in V(G)\sm C\mid N_C(x)=\{v_{i}, v_{i+3}\}\}.\\
L &=& \{x\in V(G)\sm C\mid N_C(x)=\emptyset\}.
\end{eqnarray*}
Moreover, let $A = A_1\cup\cdots\cup A_6$, $B =
B_1\cup\cdots\cup B_6$, and $D = D_1\cup\cdots\cup D_6$.
Then the following properties hold for all $i$, $i \bmod 6$:
\begin{enumerate}[(a)]\setlength\itemsep{0pt}
\item \label{c6-V} $V(G) = V(C)\cup A\cup B\cup D \cup S \cup L$.
\item \label{c6-cliques}
Each of $A_i \cup B_i\cup B_{i+5}$, $D_i$ and $S$ is a clique.
\item \label{c6-edge-comp}
$[A_i, A_{i+1} \cup A_{i+5} \cup D_i]$, $[B_i, B_{i+1} \cup B_{i+3}
\cup B_{i+5} \cup D_{i+2}]$, and $[S, A_i\cup B_i\cup D_i]$ are
complete.
\item \label{c6-edge-emp}
$[A_i, A_{i+3} \cup B_{i+2} \cup B_{i+3} \cup D_{i+1} \cup D_{i+2}]$,
$[B_i, B_{i+2} \cup B_{i+4}]$, and $[D_i, D_{i+1}]$ are empty.
\item \label{c6-BiDi}
If $B_i \neq \emptyset$, then $D_i \cup D_{i+1}=\emptyset$.
\item\label{c6-subBi}
If $B_i \neq \emptyset$ and $B_{i+1}\neq \emptyset$, then $B_{i+3}
\cup B_{i+4} = \emptyset$.
\end{enumerate}
\end{theorem}

\begin{proof} We note that $D_i = D_{i+3}$, for all $i$.

(\ref{c6-V})  Suppose that there is a vertex $x$ in $G$. We may assume that $x\in V(G)\sm V(C)$. If $x$ has no neighbor in $C$, then $x\in L$. So, suppose that $x$ has a neighbor in $C$. If $N_C(x)=\{v_i\}$ (or $\{v_i,v_{i+1}\}$), for some $i$, then $(C\sm \{v_{i+1}\})\cup \{x\}$ induces a $P_6$.  In all the remaining cases,  we see that either $C\cup\{x\}$ contains an induced $C_4$ or    $x\in A\cup B\cup D\cup S$.  So (\ref{c6-V}) holds.

(\ref{c6-cliques}) If there are non-adjacent vertices $x$ and $y$ in
one of the listed sets, then either $\{x, v_{i-1}, v_{i+1}, y\}$ or
$\{x, v_{i}, v_{i+3}, y\}$ induces a $C_4$.

(\ref{c6-edge-comp}) Let $i=1$ and suppose that there are non-adjacent
vertices $x$ and $y$ in one of the listed sets.  If $x \in A_1$ and $y
\in A_2\cup D_1$, then $\{v_2, x, v_6, v_5, v_4, y\}$ induces a $P_6$.
If $x\in B_1$ and $y\in B_2$, then $\{x, v_1, y, v_3\}$ induces a
$C_4$.  If $x\in B_1$ and $y\in B_4\cup D_3$, then $\{x, v_3, y,
v_6\}$ induces a $C_4$.  If $x\in S$ and $y\in A_1\cup B_1\cup D_1$,
then either $\{x, v_6, y, v_2\}$ or $\{x, v_1, y, v_3\}$ induces a
$C_4$.  The other cases are symmetric.

(\ref{c6-edge-emp}) Let $i=1$ and suppose that there is an edge $xy$
in one of the listed sets.  If $x \in A_1$ and $y \in A_4\cup B_3\cup
D_2$, then $\{x, v_6, v_5, y\}$ induces a $C_4$.  If $x\in B_1$ and
$y\in B_3$, then $\{x, v_6, v_5, y\}$ induces a $C_4$.  If $x\in D_1$
and $y\in D_2$, then $\{x, v_1, v_2, y\}$ induces a $C_4$.  The other
cases are symmetric.

(\ref{c6-BiDi}) Let $i=1$ and let $x\in B_1$.  Up to symmetry, if
there exists a vertex $y\in D_1$, then by (\ref{c6-edge-comp}), $xy\in
E(G)$.  But then $\{x, y, v_4, v_3\}$ induces a $C_4$.  So $D_1 =
\emptyset$.

(\ref{c6-subBi}) Let $i=1$.  Let $x\in B_1$ and $y\in B_2$.  Up to
symmetry, if there exists a vertex $z\in B_4$, then by
(\ref{c6-edge-comp}), $xy, xz\in E(G)$, and by (\ref{c6-edge-emp}),
$yz\notin E(G)$.  But then $\{x, y, v_3, z\}$ induces a $C_4$.  So
$B_4 = \emptyset$.

This shows Theorem~\ref{thm:c6}.
\end{proof}

\subsection*{When there is an $F_3$}

Now we can give the proof of the first item of
Theorem~\ref{thm:structure} which we restate it as follows.

\begin{theorem}\label{lem:f3}
Let $G$ be a $(P_6,C_4)$-free graph with no universal vertex and no
clique cutset.  Suppose that $G$ contains an $F_3$.  Then $G$ is a
blowup of $F_3$.
\end{theorem}
\begin{proof}
Consider the graph $F_3$ as shown in Figure~\ref{fig:f123} and let
$C=\{v_1, \ldots, v_6\}$.  By Theorem~\ref{thm:c6}(\ref{c6-V}), and
with the same notation, every vertex in $V(G)\setminus C$ belongs to
$A_i\cup B_i\cup D_i\cup S \cup L$ for some $i$.  Note that $x\in
A_2$, $y\in A_4$ and $z\in A_6$.  We first claim that:
\begin{equation}\label{nobd}
\mbox{$B_i\cup D_i=\emptyset$, for all $i$.}
\end{equation}
Proof: Suppose on the contrary, and up to symmetry, that there is a
vertex $u\in B_1\cup D_1$.  Suppose that $u\in B_1$.  By
Theorem~\ref{thm:c6}(\ref{c6-cliques}) we have $ux\in E(G)$, and by
Theorem~\ref{thm:c6}(\ref{c6-edge-emp}) we have $uy\notin E(G)$.  Then
either $\{u,x,z,v_6\}$ or $\{u,v_3,y,z\}$ induces a $C_4$.  Now
suppose that $u\in D_1$.  By Theorem~\ref{thm:c6}(\ref{c6-edge-emp})
we have $ux, uz\notin E(G)$.  Then $u$-$v_4$-$v_5$-$z$-$x$-$v_2$ is a
$P_6$.   So (\ref{nobd}) holds.

Next, we claim that:
\begin{equation}\label{nos}
\mbox{$L\cup S=\emptyset$}.
\end{equation}
Proof: Suppose that $L\neq\emptyset$.  By
Theorem~\ref{thm:c6}(\ref{c6-cliques}), $S$ is a clique.  Since $S$ is
not a clique cutset, and by (\ref{nobd}), some vertex $w$ in $L$ has a
neighbor $a\in A$, say $a\in A_1$.  But then
$w$-$a$-$v_6$-$v_5$-$v_4$-$v_3$ is a $P_6$, a contradiction.  Hence
$L=\emptyset$.  Now if $S\neq\emptyset$, then by
Theorem~\ref{thm:c6}(\ref{c6-cliques}) and (\ref{c6-edge-comp}), any
vertex in $S$ is universal, a contradiction.  So (\ref{nos}) holds.

\medskip

We note that every vertex $a\in A_2$ is either complete or
anticomplete to $\{y,z\}$, for otherwise $G[\{a, v_3, y, z, v_1\}]$
has an induced $C_4$.  So let $A'_2=\{v_2\}\cup \{u\in A_2\mid u$ is
anticomplete to $\{y,z\}\}$ and $X=A_2\setminus A'_2$.  Note that
$x\in X$.  Define sets $A'_4$, $Y$, $A'_6$, $Z$ similarly.

By Theorem~\ref{thm:c6}(\ref{c6-edge-comp}) and (\ref{c6-edge-emp}),
we know that $[A_1, A'_2\cup X \cup A'_6\cup Z]$ is complete and
$[A_1, A'_4\cup Y]=\emptyset$.  Likewise, $[A_3, A'_2\cup X \cup
A'_4\cup Y]$ is complete and $[A_3,A'_6\cup Z]=\emptyset$, and $[A_5,
A'_4\cup Y \cup A'_6\cup Z]$ is complete and $[A_5, A'_1\cup
X]=\emptyset$.  Moreover there is no edge $a_1a_3$ with $a_1\in A_1$
and $a_3\in A_3$, for otherwise $\{a_1, a_3, y, z\}$ induces a $C_4$.
So $[A_1,A_3]=\emptyset$, and similarly $[A_3,A_5]=\emptyset$ and
$[A_5, A_1]=\emptyset$.

There is no edge $a'_2a'_4$ with $a'_2\in A'_2$ and $a'_4\in A'_4$,
for otherwise $\{a'_2, a'_4, y,x\}$ induces a $C_4$.  So
$[A'_2,A'_4]=\emptyset$, and similarly $[A'_4,A'_6]=\emptyset$ and
$[A'_6, A'_2]=\emptyset$.

There is no edge $a'_2y'$ with $a'_2\in A'_2$ and $y'\in Y$, for
otherwise $\{a'_2, y', z, v_1\}$ induces a $C_4$.  Hence, and by
symmetry, $[A'_2,Y\cup Z]=\emptyset$, and similarly $[A'_4,Z\cup
X]=\emptyset$ and $[A'_6, X\cup Y]=\emptyset$.

Finally, any two vertices $x'\in X$ and $y'\in Y$ are adjacent, for
otherwise $\{x', v_3, y', z\}$ induces a $C_4$.  Hence $[X,Y]$ is
complete, and similarly $[X,Z]$ and $[Y,Z]$ are complete. Now we exhibit the mapping
$Q_v\rightarrow v$, $v\in V(F_3)$ of the definition of a blowup, as follows:
 $A_i'\rightarrow v_i$, for $i$ even, and $A_i\rightarrow v_i$, for $i$ odd, $X\rightarrow x$,  $Y\rightarrow y$, and
 $Z\rightarrow z$. Then the above properties mean that $G$ is a blowup of $F_3$. This completes the proof.
\end{proof}

\subsection*{When there is an $F_1$ and no $F_3$}

Here we give the proof of the second item of
Theorem~\ref{thm:structure}.

\begin{theorem}\label{thm:f1nof3}
Let $G$ be a $(P_6, C_4)$-free graph with no universal vertex and no
clique cutset.  Suppose that $G$ contains an $F_1$ and no $F_3$.  Then
$G$ is a band.
\end{theorem}
\begin{proof}
Consider the graph $F_1$ as shown in Figure~\ref{fig:f123} and let
$C=\{v_1, \ldots, v_5\}$.  We use the same notation as in
Theorem~\ref{thm:c5}.  So $x\in X_{12}$, $y\in X_{23}$ and $z\in
X_{34}$.  By Theorem~\ref{thm:c5}(\ref{c5-b}) and (\ref{c5-c}), we
know that $[X_{23}, X_{12}\cup X_{34}]$ is complete and $[X_{12},
X_{34}]=\emptyset$.  Note that $X_{12}$ is a clique, for otherwise
$v_1, y$ and two non-adjacent vertices from $X_{12}$ induce a $C_4$.
Similarly, $X_{23}$ and $X_{34}$ are cliques.  We claim that:
\begin{equation}\label{f1-nowxs}
\mbox{$W=\emptyset$, and $X_{51}\cup X_{54}=\emptyset$, and
$A=\emptyset$.}
\end{equation}
Proof: Suppose the contrary.  Up to symmetry, there is a vertex $u\in
W_1\cup W_2\cup W_5\cup X_{51}\cup A$.  Suppose $u\in W_1$.  By
Theorem~\ref{thm:c5}(\ref{c5-b}) we have $ux, uy\notin E(G)$.  Then
$u$-$v_1$-$x$-$y$-$v_3$-$v_4$ is a $P_6$.  Now suppose $u\in W_2$.  By
Theorem~\ref{thm:c5}(\ref{c5-b}) we have $uz\notin E(G)$.  Then either
$u$-$v_2$-$y$-$z$-$v_4$-$v_5$ or $u$-$y$-$z$-$v_4$-$v_5$-$v_1$ is a
$P_6$.  Now suppose $u\in W_5$.  By Theorem~\ref{thm:c5}(\ref{c5-b})
we have $uz\notin E(G)$.  Then $u$-$v_5$-$v_1$-$v_2$-$v_3$-$z$ is a
$P_6$.  Now suppose $u\in X_{51}$.  By
Theorem~\ref{thm:c5}(\ref{c5-b}) we have $uy, uz\notin E(G)$.  Then
$u$-$v_1$-$v_2$-$y$-$z$-$v_4$ is a $P_6$.  Thus we have established
that $W=\emptyset$ and $X_{51}\cup X_{54}=\emptyset$ so $X=X_{12}\cup
X_{23}\cup X_{34}$.  Finally, suppose that $u\in A$.  We note that for
any two vertices $x'\in X_{12}$ and $y'\in X_{23}$ the vertex $u$ is
either complete or anticomplete to $\{x',y'\}$, for otherwise
$G[u,v_1, x', y', v_3]$ contains an induced $C_4$.  The same holds for
any two vertices in $X_{23}$ and $X_{34}$.  It follows that $u$ is
either complete or anticomplete to $X$.  If $u$ is complete to $X$
then by Theorem~\ref{thm:c5}(\ref{c5-a}), $u$ is a universal vertex, a
contradiction.  If $u$ is anticomplete to $X$ then $\{v_1, \ldots,
v_5, x, y, z, u\}$ induces an $F_3$, a contradiction.  Thus
(\ref{f1-nowxs}) holds.

\medskip

By (\ref{f1-nowxs}) we have $V(G)=C\cup T_1\cup\cdots\cup T_5\cup
X_{12}\cup X_{23}\cup X_{34}$.  By Theorem~\ref{thm:c5}(\ref{c5-b}) we
know that $[T_5, T_2\cup T_3\cup X_{23}]=\emptyset$.  We claim that:
\begin{equation}\label{f1-t5}
\mbox{$[T_5, X_{12}\cup X_{34}]=\emptyset$, and $[T_5, T_1\cup T_4]$
is complete.}
\end{equation}
Proof: Pick any vertex $t_5\in T_5$.  Suppose up to symmetry that
$t_5$ has a neighbor $x'\in X_{12}$.  Then either $\{t_5,x',y,z\}$
induces a $C_4$ or $v_5$-$t_5$-$x'$-$v_2$-$v_3$-$z$ is a $P_6$, a
contradiction.  Now suppose up to symmetry that $t_5$ has a
non-neighbor $t_1\in T_1$.  Then $t_5$-$v_5$-$t_1$-$v_2$-$v_3$-$z$ is
a $P_6$ (since $t_1z\notin E(G)$ by Theorem~\ref{thm:c5}(\ref{c5-b})).
Thus (\ref{f1-t5}) holds.

\medskip

By Theorem~\ref{thm:c5}(\ref{c5-b}) we have $[T_1, T_3\cup T_4\cup
X_{34}]=\emptyset$ and $[T_4, T_1\cup T_2\cup X_{12}]=\emptyset$.  We
claim that:
\begin{equation}\label{f1-t1t4}
\mbox{$[T_1,X_{12}]$ and $[T_4,X_{34}]$ are complete.}
\end{equation}
Proof: If, up to symmetry, there are non-adjacent vertices $t_1\in
T_1$ and $x'\in X_{12}$, then either $\{t_1,v_1,x',y\}$ induces a $C_4$ or
$t_1$-$v_1$-$x'$-$y$-$z$-$v_4$ is a $P_6$.  Thus (\ref{f1-t1t4})
holds.

\medskip

By Theorem~\ref{thm:c5}(\ref{c5-b}) we have $[T_2, T_4]=\emptyset$ and
$[T_3,T_1]=\emptyset$.  We claim that:
\begin{equation}\label{f1-t2t3}
\longbox{$[T_2, X_{12}\cup X_{23}]$ and $[T_3, X_{23}\cup X_{34}]$ are
complete.  Moreover, every vertex in $T_2$ is complete either to $T_1$
or to $T_3$, and every vertex in $T_3$ is complete either to $T_2$ or
to $T_4$.}
\end{equation}
Proof: Up to symmetry pick any $t_2\in T_2$, $x'\in X_{12}$ and $y'\in
X_{23}$.  Then $t_2y'\in E(G)$, for otherwise either $\{t_2, v_2, y',
z\}$ induces a $C_4$ or $t_2$-$v_2$-$y'$-$z$-$v_4$-$v_5$ is a $P_6$.
Then $t_2x'\in E(G)$, for otherwise $\{t_2, y', x', v_1\}$ induces a
$C_4$.  This proves the first sentence of (\ref{f1-t2t3}).  Now
suppose that some $t_2\in T_2$ has a non-neighbor $t_1\in T_1$ and a
non-neighbor $t_3\in T_3$.  Then either $\{t_1, v_1, t_2, y\}$ induces
a $C_4$ or $t_1$-$v_1$-$t_2$-$y$-$t_3$-$v_4$ is a $P_6$.  Thus
(\ref{f1-t2t3}) holds.
\begin{equation}\label{f1-y}
\mbox{Every vertex in $X_{23}$ is anticomplete to $T_1$ or $T_4$.}
\end{equation}
Proof: If any $y'\in X_{23}$ has neighbors $t_1\in T_1$ and $t_4\in
T_4$, then $\{y, t_1, v_5, t_4\}$ induces a $C_4$.  Thus (\ref{f1-y})
holds.

By (\ref{f1-y}) there is a partition $Y_1, Y_4$ of $X_{23}$ such that
$[Y_1, T_4]=[Y_4, T_1]=\emptyset$.

\medskip

Now let $Q_i=\{v_i\}\cup T_i$ for each $i\in \{1,4,5\}$.  We observe
that the set $T_2\cup \{v_2\}\cup X_{12}\cup Y_1$ is a clique, because
each of $T_2\cup \{v_2\}$, $X_{12}$ and $Y_1$ and they are pairwise
complete as proved above.  Likewise $T_3\cup \{v_3\}\cup X_{34}\cup
Y_4$ is a clique.  Let $R_2=\{u\in T_2\cup \{v_2\}\cup X_{12}\cup
Y_1\mid u$ is complete to $Q_1\}$, and let $Q_2= (T_2\cup \{v_2\}\cup
X_{12}\cup Y_1)\setminus R_2$.  Likewise let $R_3=\{u\in T_3\cup
\{v_3\}\cup X_{34}\cup Y_4\mid u$ is complete to $Q_4\}$, and let
$Q_3= (T_3\cup \{v_3\}\cup X_{34}\cup Y_4)\setminus R_3$.  Note that
$\{v_2\}\cup X_{12}\subseteq R_2$ and $\{v_3\}\cup X_{34}\subseteq
R_3$ by (\ref{f1-t1t4}).  So $Q_2\subseteq T_2\cup Y_1$ and
$Q_3\subseteq T_3\cup Y_4$.  We observe that $[Q_2, Q_3]$ is complete
by (\ref{f1-t2t3}) and because $X_{23}$ is a clique.  Further, we
claim that:
\begin{equation}\label{f1-QR}
\mbox{$[Q_3, R_2]$ and $[Q_2, R_3]$ are complete.}
\end{equation}
Proof: Suppose that there are non-adjacent vertices $q\in Q_3$ and
$r\in R_2$.  Then $r\notin \{v_2\}\cup Y_1$, and so $r\in T_2\cup
X_{12}$, and $q$ has a non-neighbor $t\in T_4$.  If $r\in X_{12}$,
then $q$-$v_3$-$t$-$v_5$-$v_1$-$r$ is a $P_6$ (since $rt\notin E(G)$,
by Theorem~\ref{thm:c5}(\ref{c5-b})), a contradiction.  So $r\in T_2$.
Then since $qz\in E(G)$ (by (\ref{f1-t2t3})) and $\{q, z, r, v_2\}$
does not induce a $C_4$, $rz\notin E(G)$.  But then
$v_5$-$t$-$z$-$q$-$v_2$-$r$ is a $P_6$, a contradiction.  Thus
(\ref{f1-QR}) holds.

Moreover, by the definition of $Q_1,\ldots,Q_4,R_2$ and $R_3$, the
pairs $\{Q_1,Q_2\}$, $\{Q_2,Q_3\}$ and $\{R_2,R_3\}$ are graded.
Hence the sets $Q_1, \ldots, Q_5, R_2, R_3$ form a partition of $V(G)$
which shows that $G$ is a band.
\end{proof}

\subsection*{When there is a $C_6$ and no $F_1$}

Here we give the proof of the third item of
Theorem~\ref{thm:structure}, which we restate as follows.
\begin{theorem}\label{thm:c6nof1}
Let $G$ be a $(P_6, C_4)$-free graph that has no clique-cutset and no
universal vertex, and suppose that $G$ is $F_1$-free.  If $G$ contains
an induced $C_6$, then $G$ is a blowup of one of the graphs $H_1, H_2,
H_3, H_4$.
\end{theorem}
\begin{proof}
Let $C=\{v_1, v_2, \ldots, v_6\}$ be the vertex-set of a $C_6$ in $G$,
with edges $v_iv_{i+1}$ (mod $6$).  We use Theorem~\ref{thm:c6} with
the same notation.  If $C$ is not dominating, then by
Theorem~\ref{thm:BH} and since $G$ has no universal vertex, $G$ is a
blowup of the Petersen graph.  Therefore we may assume that $C$ is
dominating.  So $L=\emptyset$ and $V(G) = V(C) \cup A \cup B \cup D
\cup S$.  Moreover, since $G$ is $F_1$-free, we have $[A_i, A_{i+2}] =
\emptyset$ and $[A_i, B_{i+1}]$ is complete.  So, by
Theorem~\ref{thm:c6}, each of the sets $A_1\cup\{v_1\}, \ldots,
A_6\cup\{v_6\},$ $B_1, \ldots, B_6,$ $D_1, D_2, D_3, S$ is a clique
and that any two of them are either complete or anticomplete to each
other.  So $G$ is a blowup of some graph.  We now make this more
precise.  Since $G$ has no universal vertex, by
Theorem~\ref{thm:c6}(\ref{c6-cliques}) and (\ref{c6-edge-comp}), we
have $S = \emptyset$.  If $B = \emptyset$, then $G$ is a blowup of the
Petersen graph.  Now assume that $B \neq \emptyset$.  First, suppose
that two consecutive $B_j$'s are non-empty, say $B_i, B_{i+1} \neq
\emptyset$.  Then by Theorem~\ref{thm:c6}(\ref{c6-subBi}), $B_{i+3}
\cup B_{i+4} = \emptyset$, and by Theorem~\ref{thm:c6}(\ref{c6-BiDi})
$D=\emptyset$.  So again by Theorem~\ref{thm:c6}(\ref{c6-subBi}), $G$
is a blowup of $H_4$.  Next, suppose that no two consecutive $B_j$'s
are non-empty and let $B_i \neq \emptyset$.  Then $B_{i-1} = \emptyset
= B_{i+1}$ and by Theorem~\ref{thm:c6}(\ref{c6-BiDi}), $D_i =
\emptyset = D_{i+1}$.  Now, if $B_{i+3} \neq \emptyset$ or $B_{i+2}
\cup B_{i+4} = \emptyset$, then $G$ is a blowup of $H_2$, and if
$B_{i+3} = \emptyset$ and $B_{i+2} \cup B_{i+4} \neq \emptyset$, then
by Theorem~\ref{thm:c6}(\ref{c6-BiDi}), $D = \emptyset$, and so $G$ is
a blowup of $H_3$.
\end{proof}

\subsection*{When there is an $F_2$ and no $C_6$}

Here we give the proof of the fourth item of
Theorem~\ref{thm:structure}, which we restate it as follows.

\begin{theorem}\label{thm:f2}
Let $G$ be a $(P_6, C_4)$-free graph that has no clique-cutset and no
universal vertex, and suppose that $G$ is $C_6$-free.  If $G$ contains
an $F_2$, then $G$ is a blowup of either $H_5$ or $F_{k,\ell}$ for
some integers $k,\ell\ge 1$.
\end{theorem}
\begin{proof}
Consider the graph $F_2$ as shown in Figure~\ref{fig:f123} and let
$C=\{v_1, \ldots, v_5\}$.  We use the same notation as in
Theorem~\ref{thm:c5}.  Note that $t\in T_5$, and $x\in X_{12}$ and
$y\in X_{34}$, so Theorem~\ref{thm:c5}(\ref{c5-e}) implies that the
sets $X_{23}$, $X_{45}$, $X_{15}$ and $W_2$, $W_3$, $W_5$ are all
empty, and one of $W_1, W_4$ is empty.  So $V(G)=C\cup
T_1\cup\cdots\cup T_5\cup X_{12}\cup X_{34} \cup A\cup W_1\cup W_4$.
We establish a number of properties.  (Some of them were also proved
in \cite[Proof of Lemma~4]{GH}.)
\newcounter{ici}
\begin{enumerate}[(i)]\setlength\itemsep{0pt}
\item\label{f2-1}
{\it Each vertex in $T_5$ is either complete or anticomplete to
$X_{12}\cup X_{34}$.  In particular, $t$ is complete to $X_{12}\cup
X_{34}$.} \\
Proof: Suppose that some vertex $t_5\in T_5$ is not complete and not
anticomplete to $X_{12}\cup X_{34}$.  It follows that $t_5$ has a
neighbor $x'\in X_{12}$ and a non-neighbor $y'\in X_{34}$, or
vice-versa.  Then $v_5$-$t_5$-$x'$-$v_2$-$v_3$-$y'$ is a $P_6$.
\item\label{f2-2}
{\it $X_{12}$ and $X_{34}$ are cliques.} \\
Proof: If, up to symmetry, $X_{12}$ contains two non-adjacent vertices
$x',x''$, then by (\ref{f2-1}), $\{t,x',x'',v_2\}$ induces a $C_4$.
\item\label{f2-3}
{\it Each vertex in $T_2$ is either complete or anticomplete to
$X_{12}$, and each vertex in $T_3$ is either complete or anticomplete
to $X_{34}$.} \\
Proof: If, up to symmetry, some vertex $t_2\in T_2$ has a neighbor
$x'$ and a non-neighbor $x''$ in $X_{12}$, then, by (\ref{f2-2}),
$x''$-$x'$-$t_2$-$v_3$-$v_4$-$v_5$ is a $P_6$.
\item\label{f2-4}
{\it $[T_2, X_{34}]=\emptyset$, and $[T_3, X_{12}]=\emptyset$.} \\
Proof: Suppose, up to symmetry, that there are adjacent vertices
$t_2\in T_2$ and $y'\in X_{34}$.  If $t_2t\in E(G)$ then
$\{t_2,v_3,v_4,t\}$ induces a $C_4$.  If $t_2t\notin E(G)$, then by
(\ref{f2-1}), $\{t_2,y',t,v_1\}$ induces a $C_4$.
\item\label{f2-5}
{\it $[T_1, T_2\cup T_5 \cup X_{12}]$ and $[T_4, T_3\cup T_5\cup
X_{34}]$ are complete.} \\
Proof: Suppose, up to symmetry, that some vertex $t_1\in T_1$ has a
non-neighbor $u\in T_2\cup T_5 \cup X_{12}$.  Recall that $t_1y\notin
E(G)$ by Theorem~\ref{thm:c5}(\ref{c5-b}).  Also, since
$\{v_5,t,y,v_3,v_2,t_1\}$ does not induce a $C_6$, $t_1t\in E(G)$.
Suppose that $u\in X_{12}$.  Then $\{t_1,t,u,v_2\}$ induces a $C_4$, a
contradiction.  In particular $t_1x\in E(G)$.  Now suppose that $u\in
T_5$ and $u\neq t$.  If $ux\in E(G)$, then $\{u,x,t_1,v_5\}$ induces a
$C_4$.  If $ux\notin E(G)$, then by (\ref{f2-1}), $uy\notin E(G)$, and
$u$-$v_5$-$t_1$-$v_2$-$v_3$-$y$ is a $P_6$.  Finally, if $u\in T_2$,
then by (\ref{f2-4}), we have $uy\notin E(G)$, and
$u$-$v_2$-$t_1$-$v_5$-$v_4$-$y$ is a~$P_6$.
\item\label{axx}
{\it $[A, X_{12}\cup X_{34}]$ is complete.} \\
Proof: If, up to symmetry, there are non-adjacent vertices $a\in A$
and $x'\in X_{12}$, then by Theorem~\ref{thm:c5}(\ref{c5-a}) and
(\ref{f2-1}) the set $\{a,t,x',v_2\}$ induces~a~$C_4$.
\setcounter{ici}{\value{enumi}}
\end{enumerate}

Now let:
\begin{eqnarray*}
Q_i &=& \{v_i\} \cup T_i \mbox{ for } i\in\{1,4\}, \\
Q_2 &=& \{v_2\}\cup \{u\in T_2\mid u \mbox{ is complete to } X_{12}\}
\mbox{ and } R_2 = T_2\sm Q_2, \\
Q_3 &=& \{v_3\}\cup \{u\in T_3\mid u \mbox{ is complete to }
X_{34}\} \mbox{ and } R_3 = T_3\sm Q_3, \\
Q_5 &=& \{u\in T_5\mid u \mbox{ is complete to } X_{12}\cup X_{34}\}
\mbox{ and } R_5 = \{v_5\} \cup (T_5\sm Q_5),
\end{eqnarray*}
Recall that, by Theorem~\ref{thm:c5}(\ref{c5-b}), $[T_i,
T_{i+2}]=\emptyset$, for all $i$.  Then:
\begin{enumerate}[(i)]\setlength\itemsep{0pt}\setcounter{enumi}{\value{ici}}
\item
{\it $[Q_2,R_3]$ and $[Q_3,R_2]$ are complete.} \\
Proof: If there are non-adjacent vertices $u\in Q_2$ and $r\in R_3$,
then $r$-$v_3$-$u$-$x$-$t$-$v_5$ is a $P_6$.  The proof is similar for
$[Q_3,R_2]$.
\item
$[R_2, R_3]=\emptyset$. \\
Proof: If $r_2\in R_2$ and $r_3\in R_3$ are adjacent then
$x$-$v_1$-$r_2$-$r_3$-$v_4$-$y$ is a $P_6$.
\setcounter{ici}{\value{enumi}}
\end{enumerate}

Suppose that $W_1\cup W_4=\emptyset$.  By (\ref{axx}) and
Theorem~\ref{thm:c5}(\ref{c5-a}), $[A,V(G)\setminus A]$ is complete
and $A$ is a clique; since $G$ has no universal vertex, we deduce that
$A=\emptyset$.  Then $V(G)$ is partitioned into the ten cliques $Q_1$,
$Q_2$, $Q_3$, $Q_4$, $Q_5$, $R_5$, $X_{12}$, $R_2$, $R_3$, $X_{34}$,
and any two of them are either complete or anticomplete to each other,
and the adjacencies proved above show that $G$ is a blowup of $H_5$.

Therefore let us assume that $W_1\cup W_4\neq\emptyset$.  By
Theorem~\ref{thm:c5}(\ref{c5-d}) one of $W_1$ and $W_4$ is empty.  Up
to symmetry, let us assume that $W_1\neq\emptyset$ and
$W_4=\emptyset$.  Hence $V(G)=Q_1\cup\cdots\cup Q_5\cup R_2\cup
R_3\cup R_5\cup X_{12}\cup X_{34}\cup W_1\cup A$.  Recall that every
induced $C_5$ in $G$ is dominating, by Theorem~\ref{thm:BH}.  Then:
\begin{enumerate}[(i)]\setlength\itemsep{0pt}\setcounter{enumi}{\value{ici}}
\item\label{w1m}
{\it $[W_1,Q_1]$ is complete, and $[W_1, Q_3\cup R_3\cup Q_4\cup
X_{12}]=\emptyset$.} \\
This follows directly from Theorem~\ref{thm:c5}(\ref{c5-b})--(\ref{c5-e}).
\item\label{wq5}
{\it $[W_1, Q_5]$ is complete.} \\
Proof: If any $w\in W_1$ and $u\in Q_5$ are non-adjacent, then either
$\{w,v_1,u,y\}$ induces a $C_4$ or $\{u,x,v_2,v_3,y\}$ is a
non-dominating $C_5$ by (\ref{w1m}).
\item
$[W_1,Q_2\cup R_5]=\emptyset$. \\
Proof: Suppose that $w\in W_1$ and $u\in Q_2\cup R_5$ are adjacent.
If $u\in Q_2$, then, since $t\in Q_5$ and by (\ref{w1m}) and
(\ref{wq5}), $\{w,t,x,u\}$ induces a $C_4$.  If $u\in R_5$, then
$w$-$u$-$v_4$-$v_3$-$v_2$-$x$ is a $P_6$ by (\ref{w1m}).
\item\label{wr3}
$R_3=\emptyset$.  \\
Proof: Pick any $w\in W_1$.  If there is any vertex $r\in R_3$, then
$\{w,v_1,v_2,r,v_4,$ $y\}$ induces a $P_6$ or a $C_6$  by (\ref{w1m}).
\item\label{uh}
{\it Each component $Z$ of $W_1$ is homogeneous in $G\sm A$.} \\
Proof: Otherwise, there are adjacent vertices $z,z'\in Z$ and a vertex
$u\notin W_1\cup A$ adjacent to $z$ and not to $z'$.  By the preceding
points $u$ is in $R_2\cup X_{34}$.  If $u\in R_2$, then
$z'$-$z$-$u$-$v_3$-$v_4$-$v_5$ is a $P_6$.  If $u\in X_{34}$, then
$z'$-$z$-$u$-$v_3$-$v_2$-$x$ is a $P_6$ by (\ref{w1m}).
\item\label{un}
{\it Each component $Z$ of $W_1$ has either a neighbor in $R_2$ and no
neighbor in $X_{34}$, or a neighbor in $X_{34}$ and no neighbor in
$R_2$.} \\
Proof: If $Z$ has no neighbor in $R_2\cup X_{34}$, then by the
preceding points we have $N(Z)=Q_1\cup Q_5\cup A'$ for some
$A'\subseteq A$, and so $N(Z)$ is a clique by Theorem~\ref{thm:c5},
contradicting the hypothesis that $G$ has no clique cutset.  On the
other hand if $Z$ has a neighbor $r\in R_2$ and a neighbor $u\in
X_{34}$, then by (\ref{uh}) for any $z\in Z$ we see that
$\{z,r,v_3,u\}$ induces a $C_4$.
\item\label{zcl}
{\it Each component $Z$ of $W_1$ is a clique.} \\
Proof: Suppose that $Z$ contains non-adjacent vertices $z,z'$.  By
(\ref{uh}) and~(\ref{un}) $z$ and $z'$ have a common neighbor $u$ in
$R_2\cup X_{34}$.  Then $\{z,u,z',t\}$ or $\{z,u,z',v_1\}$ induces a
$C_4$.
\item\label{z1z2}
{\it If $Z, Z'$ are distinct components of $W_1$, then $N(Z)\cap
N(Z')\cap (R_2\cup X_{34})=\emptyset$.} \\
(Otherwise there is a $C_4$ as in the proof of (\ref{zcl}).)
\item\label{a0}
$A=\emptyset$.  \\
Proof: Suppose that there exists $a\in A$.  Since $G$ has no universal
vertex, there is a non-neighbor $z$ of $a$.  By
Theorem~\ref{thm:c5}(\ref{c5-a}) and by (\ref{axx}) we have $z\in
W_1$.  By (\ref{uh}) and~(\ref{un}) $z$ has a neighbor $u\in R_2\cup
X_{34}$.  But then $\{a,u,z,t\}$ or $\{a,u,z,v_1\}$ induces a $C_4$.
\end{enumerate}

By (\ref{wr3}) and (\ref{a0}) we have $V(G)=Q_1\cup\cdots\cup Q_5\cup
R_2\cup R_5\cup X_{12}\cup X_{34}\cup W_1$.  Now it is a routine
matter to check that $G$ is a blowup of $F_{k,\ell}$ for some
$k,\ell\ge 1$.  We clarify this point by exhibiting the mapping
$Q_v\rightarrow v$ of the definition of a blowup, as follows.  If $Z$
is any component of $W_1$, we say that it is an $R_2$-component
(resp.~$X_{34}$-component) if it has a neighbor in $R_2$ (resp.~in
$X_{34}$), and we call the set $N(Z)\cap R_2$ (resp.~$N(Z)\cap
X_{34}$) the \emph{support} of $Z$.  By (\ref{un}) and (\ref{z1z2})
the supports are non-empty and pairwise disjoint.  Let $Z_1, Z_2,
\ldots, Z_p$ be the $R_2$-components of $W_1$, and let $Z_1', Z_2',
\ldots, Z_q'$ be the $X_{34}$-components of $W_1$.  Let $k = p +1$ and
$\ell=q+1$.  Then:
\begin{itemize}\setlength\itemsep{0pt}
\item
$Z_i\rightarrow u_i$ and $N(Z_i)\cap R_2\rightarrow a_i$ for each $i
\in \{1,2\ldots, p\}$, and $X_{12}\rightarrow u_{p+1}$ and
$Q_2\rightarrow a_{p+1}$, and $R_2 \setminus \cup_{i=1}^p (N(Z_i)\cap
R_2) \rightarrow a_0$.
\item
$Z_j'\rightarrow w_j$ and $N(Z_j')\cap X_{34}\rightarrow b_j$ for each
$j \in \{1,2\ldots, q\}$, and $R_5\rightarrow w_{q+1}$ and
$Q_4\rightarrow b_{q+1}$, and $X_{34} \setminus \cup_{j=1}^q
(N(Z_j')\cap X_{34}) \rightarrow b_0$.
\item $Q_1\rightarrow x$, $Q_5\rightarrow y$, and $Q_3\rightarrow z$.
\end{itemize}
Since the components of $W_1$ and their supports are cliques, we see
that $G$ is a blowup of $F_{k,\ell}$.  This completes the proof of the
theorem.
\end{proof}

\subsection*{When there is a $C_5$, no $C_6$ and no $F_2$}

Here we give the proof of the last item of
Theorem~\ref{thm:structure}.
\begin{theorem}\label{thm:final}
Let $G$ be a $(P_6, C_4)$-free graph that has no clique-cutset and no
universal vertex, and suppose that $G$ is $C_6$-free and $F_2$-free.
If $G$ contains a $C_5$, then $G$ is either a belt or a boiler.
\end{theorem}
\begin{proof}
Let $C=\{v_1, \ldots, v_5\}$ be the vertex-set of a $C_5$ in $G$ with
edges $v_iv_{i+1}$ (mod $5$).  We use the same notation as in
Theorem~\ref{thm:c5}.    We
choose $C$ such that $|T|$ is minimized.  Remark that since $G$ is
$(P_6, C_4, C_6)$-free every hole in $G$ has length $5$ and is
dominating by Theorem~\ref{thm:BH}.  We establish a number of
properties.  (Some of them were also proved in \cite[Lemma~5]{GH}.)
\begin{enumerate}[(i)]\setlength\itemsep{0pt}
\item\label{ff-tic}
{\it If $X_{i-2,i-1}\cup X_{i+1,i+2}=\emptyset$, then $T_i$ is complete to
$T_{i-1}\cup T_{i+1}$.}  \\
Proof: Up to symmetry let $i=1$ and suppose that $X_{23}\cup
X_{45}=\emptyset$ and that some vertex $t_1\in T_1$ has a non-neighbor
$t_2\in T_2$.  Let $C'=\{t_1,v_2,v_3,v_4,v_5\}$.  So $C'$ induces a
$C_5$, and $t_2$ has only two neighbors on it, so the choice of $C$
(minimizing $|T|$) implies the existence of a vertex that has three
neighbors on $C'$ and two on $C$.  Such a vertex must be in
$X_{23}\cup X_{45}$, a contradiction.
\item\label{ff-zti1}
{\it Every component $Z$ of $W_i$ is anticomplete to one of $T_{i-1},
T_{i+1}$.} \\
Proof: Let $i=1$ and suppose that there are vertices $z,z'\in Z$ such
that $z$ has a neighbor $t_2\in T_2$ and $z'$ has a neighbor $t_5\in
T_5$.  If we can choose $z=z'$, then $C\cup\{z,t_2,t_5\}$ induces an
$F_2$.  Otherwise let $P$ be a shortest path between $z$ and $z'$ in
$G[Z]$.  Then $V(P)\cup\{t_2,v_2,t_5,v_5\}$ contains an induced $P_6$.
\item\label{ff-zti2}
{\it For every component $Z$ of $W_i$, every vertex of $T_{i-1}\cup
T_{i+1}$ is either complete or anticomplete to $Z$.} \\
Proof: Let $i=1$.  Suppose that $y,z$ are adjacent vertices in $Z$ and
that some vertex $t_2\in T_2$ is adjacent to $y$ and not to $z$.  By
(\ref{ff-zti1}) $[Z,T_5]=\emptyset$.  Then
$z$-$y$-$t_2$-$v_3$-$v_4$-$v_5$ is a $P_6$.
\item\label{ff-wix}
{\it Every vertex in $W_i$ has a neighbor in $X_{i-2,i+2}$.  In
particular if $W_i\neq\emptyset$, then $X_{i-2,i+2}\neq\emptyset$.}
\\
Proof: Let $i=1$ and suppose that some vertex of $W_1$ has no neighbor
in $X_{34}$.  Let $Z$ be the component of $W_1$ that contains this
vertex.  By (\ref{ff-zti1}) we may assume that $[Z,T_5]=\emptyset$.
Let $Z_0=\{z\in Z\mid z$ has no neighbor in $X_{34}\}$, so
$Z_0\neq\emptyset$.  Let $Y_0$ be a component of $G[Z_0]$, and let
$Y_1=N(Y_0)\cap (Z\setminus Z_0)$ and $Y_2=N(Y_0)\cap T_2$, and
$A_0=N(Y_0)\cap A$.  By Theorem~\ref{thm:c5} and since
$[Z,T_5]=\emptyset$ we have $N(Y_0)=\{v_1\}\cup T_1\cup Y_1\cup
Y_2\cup A_0$ and $Y_0$ is complete to $\{v_1\}\cup T_1$, and by
(\ref{ff-zti2}) $Y_2$ is complete to $Y_0$.  Suppose that some vertex
$y\in Y_1$ is not complete to $Y_0$.  Then there are adjacent vertices
$y_0,z_0\in Y_0$ and a vertex $x\in X_{34}$ such that
$z_0$-$y_0$-$y$-$x$-$v_4$-$v_5$ is a $P_6$.  Hence $Y_0$ is complete
to $N(Y_0)\setminus A_0$.  Since $G$ has no clique-cutset, there are
non-adjacent vertices $u,v\in N(Y_0)$.  By Theorem~\ref{thm:c5} and
(\ref{ff-zti2}) we know that $[Y_1\cup A_0, \{v_1\}\cup T_1\cup Y_2]$
is complete, so we have either (a) $u,v\in Y_1$, or (b) $u\in Y_1$ and
$v\in A_0$, or (c) $u\in T_1$ and $v\in Y_2$.  Pick any $y_0\in Y_0$.
In case (a), by the definition of $Z_0$ there are vertices $x,x'\in
X_{34}$ such that $xu, x'v\in E(G)$.  If we can choose $x=x'$, then
$\{x,u,y_0,v\}$ induces a $C_4$; and in the opposite case either $\{x,
x', u, v, y_0\}$ induces a non-dominating $C_5$ (if $xx'\in E(G)$),
because $v_5$ has no neighbor in it, or $\{y_0, u, v, x, x', v_4\}$
induces a $C_6$, a contradiction.  In case (b) we may choose $y_0$
adjacent to $v$.  By the definition of $Z_0$, $u$ has a neighbor $x\in
X_{34}$.  Then $vx\notin E(G)$, for otherwise $\{v,x,u,y_0\}$ induces
a $C_4$.  But then $\{v_1, v_3, v_4, v_5, x, y_0, u, v\}$ induces an
$F_2$.  In case (c), $\{y_0,u,v_2,v\}$ induces a $C_4$.
\item\label{ff-xiitt}
{\it $X_{i+2,i-2}$ is anticomplete to one of $T_{i-1}, T_{i+1}$.} \\
Proof: Let $i=1$ and suppose that there are vertices $x,y\in X_{34}$
such that $x$ has a neighbor $t_2\in T_2$ and $y$ has a neighbor
$t_5\in T_5$.  Then $xt_5\notin E(G)$, for otherwise
$\{v_1,t_2,x,t_5\}$ induces a $C_4$; and similarly $yt_2\notin E(G)$.
Moreover $xy\notin E(G)$, for otherwise
$v_2$-$t_2$-$x$-$y$-$t_5$-$v_5$ is a $P_6$.  But then
$\{v_1,v_2,v_3,v_4,x,y,t_2,t_5\}$ induces an $F_2$.
\item\label{ff-xiitx}
{\it If $[X_{i+2,i-2},T_{i+1}]\neq\emptyset$ then
$X_{i-1,i}=\emptyset$.  Likewise if
$[X_{i+2,i-2},T_{i-1}]\neq\emptyset$ then $X_{i,i+1}=\emptyset$.} \\
Proof: Let $i=1$, and suppose that some vertex $x\in X_{34}$ has a
neighbor $t\in T_2$ and that there is a vertex $y\in X_{51}$.  Then
$xy\notin E(G)$, for otherwise $\{x,v_4,v_5,y\}$ induces a $C_4$, and
$ty\in E(G)$, for otherwise $v_2$-$t$-$x$-$v_4$-$v_5$-$y$ is a $P_6$;
but then $C\cup\{t,x,y\}$ induces an $F_2$.
\item\label{ff-xiitx2}
{\it Every vertex in $X_{i+2,i-2}$ that has a neighbor in $T_{i+1}$ is
complete to $T_{i-2}$.} \\
Proof: Let $i=1$, and suppose that some vertex $x\in X_{34}$ has a
neighbor $t\in T_2$ and that $x$ is not adjacent to a vertex $y\in
T_4$.  Then by Theorem~\ref{thm:c5}(\ref{c5-b}), $ty\notin E(G)$.  But
then $C\cup\{t,x,y\}$ induces an $F_2$.
\item\label{ff-wxiit}
{\it If $W_i\neq \emptyset$, then $[X_{i+2,i-2}, T_{i+2}\cup T_{i-2}]$
is complete.} \\
Proof: Let $i=1$, and suppose that, up to symmetry, there are
non-adjacent vertices $x\in X_{34}$ and $t\in T_3$ and that there is a
vertex $w\in W_1$.  Then $\{w,v_1,v_2,t,v_4,x\}$ induces a $P_6$ or a
$C_6$.  \setcounter{ici}{\value{enumi}}
\end{enumerate}

Suppose that $X=\emptyset$.  Then (\ref{ff-wix}) implies that
$W=\emptyset$, so $V(G)= C\cup T\cup A$.  Moreover $A=\emptyset$, for
otherwise any vertex in $A$ is universal in $G$, by
Theorem~\ref{thm:c5}(\ref{c5-a}); and (\ref{ff-tic}) implies that
$[T_i, T_{i+1}]$ is complete for all $i$.  So $G$ is a blowup of
$C_5$, which is a special case of a belt.

\medskip

Now assume that $X\neq\emptyset$, say $X_{34}\neq\emptyset$.  By
Theorem~\ref{thm:c5}(\ref{c5-d}) and by symmetry, we may assume that
$X_{23}\cup X_{45}\cup X_{51}=\emptyset$, so $X=X_{12}\cup X_{34}$,
and consequently, by (\ref{ff-xiitt}) and (\ref{ff-xiitx}) and up to
symmetry, that $[X_{34}, T_5]=\emptyset$ and $[X_{12},
T_5]=\emptyset$.  By (\ref{ff-wix}) we have $W=W_1\cup W_4$, and by
Theorem~\ref{thm:c5}(\ref{c5-e}) one of $W_1,W_4$ is empty, so, still
up to symmetry, we may assume that $W_4=\emptyset$.  Let:
\begin{eqnarray*}
W_1^T &=& \{w\in W_1\mid w \mbox{ has a neighbor in } T_2\},  \\
W_1^N &=& \{w\in W_1\mid w \mbox{ has no neighbor in } T_2\},  \\
X_{34}^T &=& \{x\in X_{34}\mid x \mbox{ has a neighbor in } T_2\}, \\
X_{34}^N &=& \{x\in X_{34}\mid x \mbox{ has no
neighbor in  $T_2$ and has a neighbor in } W_1\}, \\
X_{34}^0 &=& \{x\in X_{34}\mid x \mbox{ has no
neighbor in } T_2\cup W_1\}, \\
X_{34}^W &=& \{x\in X_{34}\mid x \mbox{ has a neighbor in } W_1\}.
\end{eqnarray*}
Clearly $W_1=W_1^T\cup W_1^N$ and $X_{34}=X_{34}^T\cup X_{34}^N\cup
X_{34}^0$.  Moreover we have $X_{34}^N\subseteq X_{34}^W \subseteq
X_{34}^N \cup X_{34}^T$.  Recall that $[W_1,T_1]$ is complete and that
$[W_1, T_3\cup T_4\cup X_{12}]=\emptyset$ by
Theorem~\ref{thm:c5}(\ref{c5-b})--(\ref{c5-e}).  By (\ref{ff-tic}),
$[T_1, T_2\cup T_5]$ and $[T_4, T_3\cup T_5]$ are complete.  We
establish some additional facts.
\begin{enumerate}[(i)]\setlength\itemsep{0pt}\setcounter{enumi}{\value{ici}}
\item\label{ff-w1t5}
$[W_1, T_5]=\emptyset$. \\
Proof: Suppose that $w\in W_1$ and $t\in T_5$ are adjacent.  By
(\ref{ff-wix}) $w$ has a neighbor $x\in X_{34}$.  Since $[X_{34},
T_5]=\emptyset$, we see that $v_5$-$t$-$w$-$x$-$v_3$-$v_2$ is a $P_6$.
\item\label{ff-w1tn}
$[W_1^T, W_1^N]=\emptyset$.  \\
This follows directly from (\ref{ff-zti2}).
\item\label{ff-wxt}
{\it For every edge $wx$ with $w\in W_1$ and $x\in X_{34}$, every
vertex $u$ in $T_2$ is either complete or anticomplete to $\{w,x\}$.
Hence $[W_1^T, X_{34}^N]=\emptyset$ and $[W_1^N, X_{34}^T]=\emptyset$.
Also every vertex $u$ in $A$ is either complete or anticomplete to
$\{w,x\}$.} \\
Proof: In the opposite case there is a $C_4$ in
$G[\{v_1,w,x,v_3,u\}]$.
\item\label{ff-ax34tw1t}
{\it $[A, X_{34}^T \cup W_1^T]$ is complete.} \\
Proof: Consider any $a\in A$.  First pick any $x\in X_{34}^T$, so $x$
has a neighbor $t\in T_2$.  Then $at\in E(G)$ by
Theorem~\ref{thm:c5}(\ref{c5-a}), and $ax\in E(G)$, for otherwise
$\{a,t,x,v_4\}$ induces a $C_4$.  Now pick any $w\in W_1^T$.  So $w$
has a neighbor $t\in T_2$ and, by (\ref{ff-wix}), a neighbor $x\in
X_{34}$.  Then $xt\in E(G)$ by (\ref{ff-wxt}), so $x\in X_{34}^T$, and
$ax\in E(G)$ by the preceding point of this claim.  Then $aw\in E(G)$,
for otherwise $\{a,v_1,w,x\}$ induces a~$C_4$.
\item\label{ff-xw1}
{\it Any vertex $x\in X_{34}^W$ is complete to $(X_{34}\sm x)\cup
T_3\cup T_4$.} \\
Proof: Suppose up to symmetry that $x$ has a non-neighbor $y\in
(X_{34}\sm x)\cup T_3$.  Let $w\in W_1$ be any neighbor of $x$.  Then
either $\{w,x,y,v_3\}$ induces a $C_4$ (if $wy\in E(G)$) or
$v_5$-$v_1$-$w$-$x$-$v_3$-$y$ is a $P_6$.
 \item\label{ff-x12t3}
{\it Every vertex in $X_{12}$ has a neighbor in $T_3$.} \\
Proof: Suppose on the contrary that the set $Z=\{z\in X_{12}\mid z$
has no neighbor in $T_3\}$ is non-empty, and let $Y$ be the vertex-set
of a component of $G[Z]$.  Let $Y'=N(Y)\cap X_{12}$, $T'_1=N(Y)\cap
T_1$, $T'_2=N(Y)\cap T_2$, and $A'=N(Y)\cap A$.  By
Theorem~\ref{thm:c5} and the current assumption we have
$N(Y)=\{v_1,v_2\}\cup Y'\cup T'_1\cup T'_2\cup A'$.  Suppose that some
vertex $u\in Y'\cup T'_1\cup T'_2$ is not complete to $Y$; so there
are adjacent vertices $y,z\in Y$ with $uy\in E(G)$ and $uz\notin
E(G)$.  If $u\in Y'$, then $u\in X_{12}\setminus Z$, so $u$ has a
neighbor $t\in T_3$, and then $z$-$y$-$u$-$t$-$v_4$-$v_5$ is a $P_6$.
If $u\in T'_1$, then $z$-$y$-$u$-$v_5$-$v_4$-$v_3$ is a $P_6$.  The
proof is similar if $u\in T'_2$.  Hence $Y$ is complete to
$\{v_1,v_2\}\cup Y'\cup T'_1\cup T'_2$.  Since $G$ has no clique
cutset, the set $N(Y)$ contains two non-adjacent vertices $u,v$.  By
Theorem~\ref{thm:c5}, and since $[T_1,T_2]$ is complete, and up to
symmetry, we have $u\in Y'$ and $v\in Y'\cup T'_1\cup T'_2\cup A'$.
So $u$ has a neighbor $t\in T_3$.  Pick any $y\in Y$.  If $v\in Y'$,
then $v$ has a neighbor $s\in T_3$, and either $\{y,u,v,s\}$ induces a
$C_4$ (if we can choose $s=t$) or $\{y,u,v,s,t\}$ induces $C_5$ that
does not dominate $v_5$.  If $v\in T'_1$, then $\{y,u,t,v_4,v_5,v\}$
induces a $C_6$.  If $v\in T'_2$, then either $\{y,u,t,v\}$ induces a
$C_4$, or $\{y,u,t,v_3,v\}$ induces a $C_5$ that does not dominate
$v_5$.  If $v\in A'$, then we can choose $y$ adjacent to $v$, and then
$\{y,u,t,v\}$ induces a $C_4$.
\item
{\it $[X_{12}, T_1]$ is complete.} \\
Proof: This follows from (\ref{ff-x12t3}) and (\ref{ff-xiitx2}).
 \item\label{ff-ax12}
{\it $[X_{12}, A]$ is complete.} \\
Proof: Pick any $a\in A$ and $x\in X_{12}$. By (\ref{ff-x12t3}) $x$ has
a neighbor $t\in T_3$. We have $at\in E(G)$ by Theorem~\ref{thm:c5},
and $ax\in E(G)$, for otherwise $\{a,t,x,v_1\}$ induces a $C_4$.
\item\label{ff-x34z}
{\it For any component $Z$ of $G[X_{34}^0]$ the set $[Z, N(Z)\setminus
A]$ is complete and $N(Z)\setminus A$ is a clique.} \\
Proof: Let $Z$ be (the vertex-set of) a component of $G[X_{34}^0]$.
Then $N(Z)\setminus A\subseteq \{v_3, v_4\}\cup T_3\cup T_4\cup
X_{34}^N\cup X_{34}^T$.  First suppose that $[Z, N(Z)\setminus A]$ is
not complete.  So there are adjacent vertices $y,z\in Z$ and a vertex
$u\in N(Z)\setminus A$ with $uy\in E(G)$ and $uz\notin E(G)$.  Clearly
$u\notin \{v_3,v_4\}$.  If $u\in X_{34}^N\cup X_{34}^T$, then, by
(\ref{ff-xw1}) $u$ has no neighbor in $W_1$, so $u$ has a neighbor
$t\in T_2$, and then $z$-$y$-$u$-$t$-$v_1$-$v_5$ is a $P_6$.  If $u\in
T_3$, then $z$-$y$-$u$-$v_2$-$v_1$-$v_5$ is a $P_6$.  If $u\in T_4$
then $z$-$y$-$u$-$v_5$-$v_1$-$v_2$ is a $P_6$, a contradiction.  Now
suppose that $N(Z)\setminus A$ is not a clique, so it contains two
non-adjacent vertices $u,v$.  Pick any $z\in Z$.  By
Theorem~\ref{thm:c5} and since $[T_4,T_3]$ is complete we have either
(a) $u,v\in X_{34}^N\cup X_{34}^T$ or (b) $u\in X_{34}^N\cup X_{34}^T$
and $v\in T_3\cup T_4$.  In case (a), by (\ref{ff-xw1}) $u$ and $v$
have no neighbor in $W_1$, so they have neighbors respectively $t$ and
$t'$ in $T_2$; then $\{z, u,v, t,t'\}$ induces either a $C_4$ or a
non-dominating $C_5$ (because $v_5$ has no neighbor in it), a
contradiction.  In case~(b), item (\ref{ff-xw1}) implies that $u$ has
no neighbor in $W_1$, so $u$ has a neighbor $t\in T_2$.  If $v\in
T_3$, then $\{z, u,t,v_2,v\}$ induces a non-dominating $C_5$ (because
of $v_5$).  If $v\in T_4$, then $v_2$-$t$-$u$-$z$-$v$-$v_5$ is a
$P_6$.
\item\label{ff-x340}
{\it For each component $Z$ of $G[X_{34}^0]$ there are vertices $a\in
A$, $z\in Z$, $w\in W_1$ and $x\in X_{34}^N$ such that $az, wx\in
E(G)$ and $aw,ax\notin E(G)$.} \\
Proof: We have $N(Z)\subseteq X_{34}^N\cup X_{34}^T\cup \{v_3,
v_4\}\cup T_3\cup T_4\cup A$.  Since $G$ has no clique cutset there
are two non-adjacent vertices $u,v\in N(Z)$.  By (\ref{ff-x34z}) and
Theorem~\ref{thm:c5}, and since $T_3\cup T_4$ is a clique, we have
$u\in A$ and consequently $v\in X_{34}^N\cup X_{34}^T$, and by
(\ref{ff-ax34tw1t}) $v\in X_{34}^N$.  So $v$ has a neighbor $w\in
W_1$, and $uw \notin E(G)$ by (\ref{ff-wxt}).
\setcounter{ici}{\value{enumi}}
\end{enumerate}

Suppose that $X_{34}^N=\emptyset$.  Then $X_{34}^0=\emptyset$ by
(\ref{ff-x340}), and $W_1^N=\emptyset$ by~(\ref{ff-wix})
and~(\ref{ff-wxt}).  So $X_{34} = X_{34}^T$ and $W_1 = W_1^T$.  Now
$[A,V(G)\setminus A]$ is complete by Theorem~\ref{thm:c5},
(\ref{ff-ax34tw1t}) and (\ref{ff-ax12}), and since $G$ has no
universal vertex it follows that $A=\emptyset$.  So $V(G)=C\cup
T_1\cup \cdots \cup T_5\cup W_1^T \cup X_{12}\cup X_{34}^T$.  Let:
\begin{eqnarray*}
Q_i &=& \{v_i\}\cup T_i \mbox{ for each } i\in\{1,4,5\}. \\
Q_2 &=& \{v\mid v \mbox{ is universal in } G[\{v_2\}\cup T_2 \cup
X_{12}\cup W_1]\}. \\
R_2 &=& (\{v_2\}\cup T_2 \cup X_{12}\cup W_1)\sm Q_2.  \\
Q_3 &=& \{v\mid v \mbox{ is universal in } G[\{v_3\}\cup T_3 \cup
X_{34}]\}.  \\
R_3 &=& (\{v_3\}\cup T_3 \cup X_{34}) \sm Q_3.
\end{eqnarray*}
Hence $V(G)=Q_1\cup\cdots\cup Q_5\cup R_2\cup R_3$.  We claim that
$Q_2\neq\emptyset$.  Indeed, if $W_1=\emptyset$ then $v_2\in Q_2$.  So
suppose that $W_1\neq\emptyset$.  By (\ref{ff-wix}) and (\ref{ff-xw1})
the set $Y_{34}=\{x\in X_{34}\mid x$ has a neighbor in $W_1\}$ is
non-empty and is a clique.  Since $X_{34}^N\cup X_{34}^0=\emptyset$,
every vertex of $Y_{34}$ has a neighbor in $T_2$, and it follows that
some vertex $t$ in $T_2$ is complete to $Y_{34}$ (otherwise there are
vertices $y', y''\in Y_{34}$ and $t',t''\in T_2$ that induce a $C_4$).
Let us verify that $t \in Q_2$.  We know that $t$ is complete to
$T_2\sm t$ by Theorem~\ref{thm:c5}.  Any $w\in W_1$ has a neighbor
$x\in X_{34}$ by (\ref{ff-wix}), and so $tw\in E(G)$ for otherwise
$\{t,x,w,v_1\}$ induces a $C_4$.  Now consider any $y\in X_{12}$.
Pick any $w\in W_1$ and $x\in X_{34}\cap N(w)$.  Then $ty \in E(G)$,
for otherwise $y$-$v_2$-$t$-$x$-$v_4$-$v_5$ is a $P_6$.  So $t\in
Q_2$, and the claim that $Q_2\neq\emptyset$ is established.
Now the properties of the nine sets $Q_1, \ldots, Q_5, R_2, R_3$ satisfy
all the axioms of the belt.  We make this more precise as follows:
\begin{itemize}\setlength\itemsep{0pt}
\item By Theorem~\ref{thm:c5} and by (\ref{ff-tic}), we know that
$Q_1, Q_4$ and $Q_5$ are non-empty cliques, $[Q_1\cup Q_4, Q_5]$ is
complete and $[Q_1, Q_4]=\emptyset$.

\item
Clearly $Q_2$ and $Q_3$ are cliques, with $v_3\in Q_3$, and
$Q_2\neq\emptyset$ as seen above.

\item
By (\ref{ff-tic}), (\ref{ff-xiitx2}) and Theorem~\ref{thm:c5},
$[Q_1,Q_2\cup R_2]$ and $[Q_4,Q_3\cup R_3]$ are complete.

\item
By the definition of $Q_2$ and $Q_3$, Theorem~\ref{thm:c5}, and since
$[X_{12}\cup X_{34},T_5]=\emptyset$, we have $[Q_2,Q_4\cup
Q_5]=\emptyset$ and $[Q_3,Q_1\cup Q_5]=\emptyset$.

\item
By the definition of $Q_2$, $Q_3$, $R_2$ and $R_3$, we have: for each
$j\in {2,3}$, $[Q_j, R_j]$ is complete, every vertex in $R_j$ has a
non-neighbor in $R_j$, every vertex in $Q_2\cup R_2$ has a neighbor in
$Q_3\cup R_3$ (by (\ref{ff-wix}) and (\ref{ff-x12t3})), and every
vertex in $Q_3\cup R_3$ has a neighbor in $Q_2\cup R_2$ (by the
definition of $X_{34}^T$)).
\end{itemize}

Thus $G$ is a belt.

\medskip

Therefore we may assume that $X_{34}^N\neq\emptyset$.  So, $W_1\neq
\emptyset$.  Then:
\begin{enumerate}[(i)]\setlength\itemsep{0pt}\setcounter{enumi}{\value{ici}}
\item\label{ff-x34n}
{\it $X_{34}^N\cup X_{34}^T\cup T_3\cup T_4$ is a clique.} \\
Proof: By (\ref{ff-wxiit}) and by Theorem~\ref{thm:c5}, it is enough
to show that $X_{34}^N\cup X_{34}^T$ is a clique.  Suppose that there
are non-adjacent vertices $x,x'\in X_{34}^N\cup X_{34}^T$.  Pick any
$y\in X_{34}^N$.  By (\ref{ff-xw1}), $y\notin \{x,x'\}$ and $yx,yx'\in
E(G)$, and $x, x'\in X_{34}^T$.  So $x$ has a neighbor $t\in T_2$, and
$x'$ has a neighbor $t'\in T_2$.  Then $\{y,x,x',t,t'\}$ induces a
cycle of length either $4$ (if $t=t'$) or $5$ and not dominating
(because $v_5$ has no neighbor in it), a contradiction.
\item\label{ff-x34ch}
{\it $G[X_{34}^0]$ is chordal.} \\
Proof: If $G[X_{34}]$ contains a hole $C$, then $C$ either has
length~$4$ or at least~$6$ or is a non-dominating $C_5$ (because of
$v_5$).
\item\label{ff-sau}
{\it For every simplicial vertex $s$ of $G[X_{34}^0]$, there are
vertices $a\in A$ and $u\in X_{34}^0\cup X_{34}^N$ with $sa, su\in
E(G)$ and $au\notin E(G)$.  } \\
Proof: Let $Z$ be the vertex-set of the component of $G[X_{34}^0]$
that contains~$s$.  So $N_Z(s)$ is a clique.  We have $N(s) \subseteq
N_Z(s) \cup (N(Z)\setminus A)\cup A$.  By (\ref{ff-x34z}) the set
$N_Z(s)\cup (N(Z)\setminus A)$ is a clique.  Since $G$ has no clique
cutset there are two non-adjacent vertices $u,v$ in $N(s)$, and so
$u\in N_Z(s)\cup (N(Z)\setminus A)$ and $v\in A$.  Since
$N(Z)\setminus A\subseteq X_{34}^N\cup X_{34}^T\cup \{v_3, v_4\}\cup
T_3\cup T_4$ and $A$ is complete to $X_{34}^T\cup \{v_3, v_4\}\cup
T_3\cup T_4$ by Theorem~\ref{thm:c5} and~(\ref{ff-ax34tw1t}), we have
$u\in X_{34}^0\cup X_{34}^N$.
\setcounter{ici}{\value{enumi}}
\end{enumerate}

Let $A_0=\{a\in A\mid a$ has a neighbor in $X_{34}^0$ and a
non-neighbor in $X_{34}^N\}$.  By~(\ref{ff-x340}) we have
$A_0\neq\emptyset$.  Since $A_0$ and $X_{34}^N$ are cliques (by
Theorem~\ref{thm:c5} and (\ref{ff-x34n})) and by the third item of
Lemma~\ref{lem:ab}, there is a vertex $x_0$ in $X_{34}^N$ that is
anticomplete to $A_0$.  Let $w_0$ be a neighbor of $x_0$ in $W_1$.
Then $w_0$ is anticomplete to $A_0$ by (\ref{ff-wxt}).

\begin{enumerate}[(i)]\setlength\itemsep{0pt}\setcounter{enumi}{\value{ici}}
\item\label{ff-x0X340comp}
{\it $x_0$ is complete to $X_{34}^0$.} \\
Proof: If there is a vertex $x\in X_{34}^0$ that is non-adjacent to
$x_0$, then $v_2$-$v_1$-$w_0$-$x_0$-$v_4$-$x$ is a $P_6$.
\item\label{ff-x34a0ch}
{\it $G[X_{34}^0\cup A_0]$ is chordal.} \\
Proof: If $G[X_{34}^0\cup A_0]$ contains a hole $C$, then $C$ either
has length~$4$ or at least~$6$ or is a non-dominating $C_5$ (because
of $w_0$).
\item\label{ff-zaz}
{\it Every vertex in $X_{34}^0$ has a neighbor in $A_0$.} \\
Proof: Let $Z$ be the vertex-set of any component of $G[X_{34}^0]$,
and let $Z_A=\{z\in Z\mid z$ has a neighbor in $A_0\}$, and suppose
that $Z\neq Z_A$.  By (\ref{ff-x34a0ch}) and Lemma~\ref{lem:chxa}
applied to $G[Z\cup A_0]$, $Z$ and $A_0$, some simplicial vertex $s'$
of $G[Z]$ has no neighbor in $A_0$.  Let $S=\{s''\in Z\mid$
$N_Z[s'']=N_Z[s']\}$; so the vertices in $S$ are simplicial in $G[Z]$
and pairwise clones, and $S$ is a clique.  Let $s$ be a vertex in $S$
with the smallest number of neighbors in $A$.  If $s$ has any neighbor
$a\in A_0$, then, since $\{S, A_0\}$ is a graded pair of cliques, by
Lemma~\ref{lem:ab} all vertices in $S$ are adjacent to $a$, a
contradiction.  So $s$ has no neighbor in $A_0$.  By (\ref{ff-sau})
there are vertices $b\in A$ and $u\in Z\cup X_{34}^N$ with $sb, su\in
E(G)$ and $bu\notin E(G)$.  We know that $b\notin A_0$, so $b$ is
complete to $X_{34}^N$, and so $u\in Z$.  Moreover $u\notin S$, for
otherwise the choice of $s$ is contradicted (since $b\in A$, and the
pair $\{A,S\}$ is graded).  Hence $u$ is not a simplicial vertex of
$G[Z]$, and so it has a neighbor $v\in Z\setminus N[s]$.  Consider any
$a\in A_0$.  We know that $as\notin E(G)$; then also $au\notin E(G)$,
for otherwise $\{a,b,s,u\}$ induces a $C_4$; and $av\notin E(G)$, for
otherwise $s$-$u$-$v$-$a$-$v_1$-$w_0$ is a $P_6$.  Hence $\{s,u,v\}$
is anticomplete to $A_0$.  Let $Y$ be the component of $G[Z\setminus
Z_A]$ that contains $s,u,v$.  Let $Z_Y=\{z\in Z_A\mid z$ has a
neighbor in $Y\}$, and let $A_Y=\{b'\in A\mid b'$ has a neighbor in
$Y\}$.  Note that $[A_Y, X_{34}^N]$ is complete.  Since
$Z_A\neq\emptyset$ and $G[Z]$ is connected, $Z_Y\neq \emptyset$.  Then
$[Y,Z_Y]$ is complete, for otherwise there are adjacent vertices
$y,y'\in Y$, a vertex $z\in Z_Y$, and a vertex $a\in A_0$ such that
$y'$-$y$-$z$-$a$-$v_1$-$w_0$ is a $P_6$.  Then $Z_Y$ is a clique, for
otherwise $\{s,v,z,z'\}$ induces a $C_4$ for any two non-adjacent
vertices $z,z'$ in $Z_Y$.  Moreover, for any $b'\in A_Y$ and $z\in
Z_Y$, we have $b'z\in E(G)$, for otherwise $\{b',y,z,a\}$ induces a
$C_4$ for any $y\in Y\cap N(b)$ and $a\in A_0\cap N(z)$.  We have
$N(Y)\subseteq Z_Y\cup A_Y\cup (N(Z)\setminus A)$, and by
Theorem~\ref{thm:c5}, items (\ref{ff-ax34tw1t}) and (\ref{ff-x34z})
and the fact that $[A_Y, X_{34}^N]$ is complete, this set is a clique,
a contradiction.
\item
{\it $[X_{34}^0, X_{34}^T]$ is complete.} \\
Proof: Suppose that some $z\in X_{34}^0$ and $x\in X_{34}^T$ are
non-adjacent.  By (\ref{ff-zaz}) $x$ has a neighbor $a\in A_0$.  Then,
by (\ref{ff-ax34tw1t}), $\{a, x, x_0, z\}$ induces a $C_4$.
\item
{\it For any two components $Z,Z'$ of $G[W_1]$, the sets $N(Z)\cap
X_{34}$ and $N(Z')\cap X_{34}$ are disjoint.} \\
Proof: Otherwise $\{v_1,x,z,z'\}$ induces a $C_4$ for some $z\in Z$,
$z'\in Z'$ and $x\in N(Z)\cap N(Z')\cap X_{34}$.
\end{enumerate}

Let:
\begin{eqnarray*}
Q &=& \{v_1\}\cup T_1, \\
B &=& \{v_3, v_4\}\cup T_3 \cup T_4 \cup X_{34}^T \cup X_{34}^N, \\
M &=& \{v_2, v_5\}\cup T_2 \cup T_5 \cup X_{12}\cup W_1,  \\
L &=& X_{34}^0.
\end{eqnarray*}
We know that $A$ and $Q$ are cliques, and $B$ is a clique by
(\ref{ff-x34n}).  Every vertex in $L$ has a neighbor in $A$ by
(\ref{ff-zaz}), and every vertex in $M$ has a neighbor in $B$ by
(\ref{ff-wix}) and (\ref{ff-x12t3}).  The subgraph $G[L]$ is $(P_4,
2P_3)$-free by Lemma~\ref{lem:chxab}, using $A_0$ in the role of $Y$,
$x_0$ in the role of $c$, and $v_1$ and $w_0$, respectively, in the
role of $c'$ and $c''$ .  The subgraph $G[M]$ has at least three
components because $\{v_2\}\cup X_{12}$, $\{v_5\}\cup T_5$ and $W_1$
are pairwise anticomplete to each other and non-empty, and $G[M]$ is
$(P_4, 2P_3)$-free by Lemma~\ref{lem:chxab}, using $B$ in the role of
$Y$, $v_1$ in the role of $c$ and the fact that $G[M]$ is not
connected.  Hence the sets $Q, A, B, L, M$ form a partition of $V(G)$ that
shows that $G$ is a boiler.
\end{proof}

\section{Additional properties of belts and boilers}\label{sec:bb}

Belts  and boilers have some additional and useful properties that we give below.

\subsection{Belts}
\begin{theorem}\label{thm:belts}
Let $G$ be a belt, with the same notation as in Section~\ref{sec:intro}.  Then:
\begin{enumerate}[(a)]\setlength\itemsep{0pt}
\item\label{belta}
For each $j\in\{2,3\}$, any two non-adjacent vertices in $R_j$ have no
common neighbor in $Q_{5-j}$.
\item\label{beltb}
$[R_2,R_3]=\emptyset$.
\item\label{beltc}
For each $j\in\{2,3\}$, every vertex of $Q_j$ that has a neighbor in
$R_{5-j}$ is complete to $Q_{5-j}$.
\item\label{beltd}
The graphs $G[R_2]$ and $G[R_3]$ are $(P_4, 2P_3)$-free.
\end{enumerate}
\end{theorem}
\begin{proof}
(\ref{belta}) If two non-adjacent vertices $r,r'\in R_2$ have a common
neighbor $v$ in $Q_3$, then $\{v_1,r,r',v\}$ induces a $C_4$.

\smallskip
	
(\ref{beltb}) Suppose that any $r_2\in R_2$ and $r_3\in R_3$ are
adjacent.  By the definition of a belt, for each $j\in\{2,3\}$ the
vertex $r_j$ has a non-neighbor $r'_j\in R_j$.  Then $r_2r'_3\notin
E(G)$, for otherwise $\{r_2,r'_3,v_4,r_3\}$ induces a $C_4$, and
similarly $r_3r'_2\notin E(G)$.  Then $\{r'_2,v_1,r_2,r_3,v_4,r'_3\}$
induces a $P_6$ or $C_6$.

\smallskip

(\ref{beltc}) Consider any $u\in Q_3$ which has a neighbor $r_2\in
R_2$, and suppose that $u$ has a non-neighbor $v\in Q_2$.  By the
definition of a belt $r_2$ has a non-neighbor $r'_2\in R_2$.  Then
$ur'_2\notin E(G)$, for otherwise $\{u,r'_2,v_1,r_2\}$ induces a
$C_4$.  But then $r'_2$-$v$-$r_2$-$u$-$v_4$-$v_5$ is a $P_6$.  The
proof is similar when $j=2$.

\smallskip

(\ref{beltd}) Pick a vertex $q_i\in Q_i$ for each $i\in\{1,4,5\}$.
Lemma~\ref{lem:chxab}, using vertices $q_1$, $q_4$ and $q_5$ in the
role of $c$, $c'$ and $c''$, implies that $G[R_2]$ is $(P_4,
2P_3)$-free.  The proof is similar for $G[R_3]$.
\end{proof}

Note that Theorem~\ref{thm:belts}(\ref{beltd}) means that $(R_2, Q_3)$
and $(R_3,Q_2)$ are $\cal C$-pairs.

\subsection{Boilers}\label{subsec:boilers}
Let $G$ be a boiler, with the same notation as in the definition.
Since every vertex in $A$ has a non-neighbor in $B$,
Lemma~\ref{lem:ab} implies that some vertex $b^*$ in $B$ is
anticomplete to $A$.  Let $m^*$ be any neighbor of $b^*$ in $M$.  Then
$m^*$ too is anticomplete to $A$ (for otherwise $\{m^*, a, b, b^*\}$
induces a $C_4$ for some $a\in A$ and $b\in B_1\cup B_2$).  Pick a
vertex $z\in Q$.

If $L$ is a clique, then $(A\cup M, B\cup L)$ is a $\cal C$-pair, so
the structure of $G$ is completely determined by
Theorem~\ref{thm:cpair} and the fact that $Q$ is complete to $A\cup M$
and anticomplete to $B\cup L$.

Therefore let us assume that $L$ is not a clique.  Let $U$ be the set
of universal vertices of~$L$.  (Possibly $U=\emptyset$.)  Let $A_L =
\{a\in A\mid a$ has a neighbor in $L\}$ and $A'_L=\{a\in A\mid a$ has
a neighbor in $L\setminus U\}$.

\begin{theorem}\label{lem:boil1}
Let $G$ be a boiler, with the same notation as above, and assume that
$L$ is not a clique.  Then, up to a permutation of the set
$\{3,\ldots,k\}$, there is an integer $j\in\{3,\ldots, k\}$ such that the
following hold:
\begin{enumerate}[(i)]\setlength\itemsep{0pt}
\item\label{boil-1}
For each $a\in A_L\setminus A'_L$, there is an integer
$i\in\{j,\ldots,k\}$ such that $a$ is complete to $M_1\cup
B_1\cup\cdots\cup M_{i-1}\cup B_{i-1}$ and anticomplete to $M_{i}\cup
B_{i}\cup\cdots\cup M_k\cup B_k$;
\item\label{boil-2}
$A'_L$ is complete to $(M\cup B)\setminus (M_{k}\cup B_{k})$ and
anticomplete to $M_k\cup B_k$;
\item\label{boil-3}
$A\setminus A_L$ is complete to $M_1\cup B_1\cup\cdots\cup M_{j-1}\cup
B_{j-1}$ and anticomplete to $M_{j}\cup B_{j}\cup\cdots\cup M_k\cup
B_k$.
\end{enumerate}
\end{theorem}
\begin{proof}
Since $A$ and $B$ are disjoint cliques and $G$ is $C_4$-free,
$[A,B_1\cup B_2]$ is complete, and $b^*$ is anticomplete to $A$,
Lemma~\ref{lem:ab} implies that there is a permutation of $\{3,..,k\}$
such that for every vertex $a\in A$ there is an integer $i\in
\{3,\ldots,k\}$ such that $a$ is complete to $M_1\cup
B_1\cup\cdots\cup M_{i-1}\cup B_{i-1}$ and anticomplete to $M_{i}\cup
B_{i}\cup\cdots\cup M_k\cup B_k$.  We may assume that $b^*\in B_k$ and
$m^*\in M_k$.

Let $J = \{i\in\{3, \ldots, k\} \mid$ some vertex in $A$ is anticomplete
to $M_i\cup B_i\}$.  By the preceding paragraph there is an integer
$j$ such that $J=\{j,\ldots,k\}$.  In particular this implies the
validity of item (\ref{boil-1}) of the lemma.

\smallskip

Now consider any vertex $a\in A'_L$.  So $a$ has a neighbor $x\in
L\setminus U$, so $x$ has a non-neighbor $x'\in L$, and by the
definition of a boiler we have $ax'\notin E(G)$.  Suppose that $a$ is
not complete to $M_i\cup B_i$ for some $i<k$, so $a$ is anticomplete
to $M_i\cup B_i$, and pick any $m\in M_i$.  Then
$m$-$z$-$a$-$x$-$b^*$-$x'$ is a $P_6$.  So $a$ is complete to $(M\cup
B)\setminus (M_{k}\cup B_{k})$, which proves (\ref{boil-2}).

\smallskip
Finally, consider any vertex $d\in A\setminus A_L$.  So $d$ is
anticomplete to $L$.  Pick any $i\in J$ and $b\in B_i$.  So there is a
vertex $a\in A_L$ that is anticomplete to $B_i\cup M_i$.  By the
definition of $A_L$ the vertex $a$ has a neighbor $x\in L$.  Then $db$
is not an edge, for otherwise $\{d,b,x,a\}$ induces a $C_4$.  It
follows that $d$ is anticomplete to $B_i\cup M_i$ which proves
(\ref{boil-3}).
\end{proof}

\section{Bounding the chromatic number}

In this section, we give a proof for Theorem~\ref{thm:54bound} and
Theorem~\ref{thm:reeds}.

We say that a stable set of a graph $G$ is \emph{good} if it meets
every clique of size $\omega(G)$ in $G$; and that it is \emph{very
good} if it meets every (inclusionwise) maximal clique of $G$. Moreover, we say that a clique $K$ in $G$ is a \emph{$t$-clique} of $G$ if $|K|=t$.

\smallskip
We will use the following theorem as a tool in proving
Theorem~\ref{thm:54bound}.
\begin{theorem}\label{thm:tools}
Let $G$ be a graph such that every proper induced subgraph $G'$ of $G$
satisfies $\chi(G')\le \lceil \frac{5}{4}\omega(G')\rceil$.  Suppose
that one of the following occurs:

\begin{enumerate}[(i)]\itemsep=1pt
\item\label{degq}
$G$ has a vertex of degree at most
$\lceil\frac{5}{4}\omega(G)\rceil-1$.
\item\label{goods}
$G$ has a (very) good stable set;
\item\label{stable}
$G$ has a stable set $S$ such that $G\sm S$ is perfect.
\item\label{fives}
For some integer $t\ge 5$ the graph $G$ has $t$ stable sets
$S_1,\ldots,S_t$ such that $\omega(G\sm (S_1\cup\cdots\cup S_t))\le
\omega(G)-(t-1)$.
\end{enumerate}
Then $\chi(G)\le \lceil \frac{5}{4}\omega(G)\rceil$.
\end{theorem}
\begin{proof}
(\ref{degq}) Suppose that $G$ has a vertex $u$ with $d(u)\le
\lceil\frac{5}{4}\omega(G)\rceil-1$.  By the hypothesis we have
$\chi(G\sm u)\le \lceil\frac{5}{4}\omega(G\sm u)\rceil$.  So we can
take any $\chi(G\sm u)$-coloring of $G\sm u$ and extend it to a
$\lceil\frac{5}{4}\omega(G)\rceil$-coloring of $G$, using for $u$ a
(possibly new) color that does not appear in its neighborhood.

(\ref{goods}) Suppose that $G$ has a (very) good stable set $S$.  Then
$\omega(G\sm S)=\omega(G)-1$.  By the hypothesis we have $\chi(G\sm S)
\le \lceil \frac{5}{4}\omega(G\sm S)\rceil = \lceil \frac{5}{4}
(\omega(G)-1) \rceil \le \lceil \frac{5}{4}\omega(G)\rceil -1$.  We
can take any $\chi(G\sm S)$-coloring of $G\sm S$ and add $S$ as a new
color class, and we obtain a coloring of $G$.  Hence $\chi(G) \le
\lceil\frac{5}{4}\omega(G)\rceil$.

(\ref{stable}) Suppose that $G$ has a stable set $S$ such that $G\sm
S$ is perfect.  Then $\chi(G\sm S)=\omega(G\sm S)\le \omega(G)$.  We
can take any $\chi(G\sm S)$-coloring of $G\sm S$ and add $S$ as a new
color class.  Hence $\chi(G)\le \omega(G)+1 \le
\lceil\frac{5}{4}\omega(G)\rceil$.

(\ref{fives}) Note that $\frac{t}{t-1}\le \frac{5}{4}$ because $t\ge
5$.  We take any $\chi(G\sm (S_1\cup\cdots\cup S_t))$-coloring of
$G\sm (S_1\cup\cdots\cup S_t)$ and use $S_1, \ldots, S_t$ as $t$ new
colors and we get a coloring of $G$.  Then $\chi(G)\le \chi(G\sm
(S_1\cup\cdots\cup S_t))+t\le \lceil\frac{5}{4}(\omega(G)-(t-1))\rceil
+t \le \lceil\frac{5}{4}\omega(G)\rceil$ because $\frac{t}{t-1}\le
\frac{5}{4}$.
\end{proof}

\subsection{Chromatic bound for blowups}

We first note that  by a result of Lov\'asz \cite{Lovasz}, any blowup of a
perfect graph is a perfect graph.

For any integer $t\ge 2$ we say that $G$ is a \emph{$t$-blowup} of $H$
if $|Q_u|=t$ for all $u\in V(H)$.  Remark that, for an integer $k$, a
$k$-coloring of the $t$-blowup of $H$ is equivalent to a collection of
$k$ stable sets of $H$ such that every vertex of $H$ belongs to at
least $t$ of them.

\bigskip

\noindent{\bf Blowups of Petersen graph}

Let $H_1$ be the Petersen graph as shown in Figure~\ref{fig:h12345}.

\begin{lemma}\label{lem:2blowup}
Let $G$ be the $2$-blowup of the Petersen graph $H_1$. Then $\chi(G)=5$.
\end{lemma}
\begin{proof}
The five sets $\{a,b,w_3,w_6\}$, $\{b,c,w_1,w_4\}$, $\{a,c,w_2,w_5\}$,
$\{z,w_1,w_3,w_5\}$ and $\{z, w_2, w_4, w_6\}$ are five stable sets,
and every vertex of $H_1$ belongs to two of them.  As observed above
this is equivalent to a $5$-coloring of $G$.  This is optimal because
$G$ has $20$ vertices and every stable set in $G$ has size at most
$4$.
\end{proof}

\begin{theorem}\label{thm:bu-peter}
If $G$ is any blowup of the Petersen graph $H_1$, then $\chi(G)\le
\lceil \frac{5}{4} \omega(G)\rceil$.
\end{theorem}
\begin{proof}
Let $q=\omega(G)$.  We prove the theorem by induction on $|V(G)|$.  We
may assume that $G$ is connected (otherwise we consider each component
separately) and that $G$ is not a clique.  Moreover, the theorem holds
easily if $G$ is any induced subgraph of $H_1$.  Now suppose that $G$
is not an induced subgraph of $H_1$.  So there is $x\in V(H_1)$ such
that $|Q_x|\ge 2$.  Since $G$ is connected and not a clique there
exists $y\in N_{H_1}(x)$ such that $Q_y\neq\es$, and so $q\ge 3$.  By
Theorem~\ref{thm:tools}~(\ref{goods}) we may assume that $G$ has no
good stable set.

Note that every maximal clique of $G$ consists of $Q_u\cup Q_v$ for
some edge $uv\in E(H_1)$ with $Q_u\neq\es$ and $Q_v\neq\es$, and we
denote it as $Q_{uv}$.  We say that such a maximal clique is
\emph{balanced} if $|Q_u|\ge 2$ and $|Q_v|\ge 2$.

Suppose that every $q$-clique of $G$ is balanced.  So $q\ge 4$.  Let
$X$ be a subset of $V(G)$ obtained by taking $\min\{2, |Q_v|\}$
vertices from $Q_v$ for each $v\in V(H_1)$.  We claim that:
\begin{equation}\label{omm4}
\omega(G\sm X)=q-4.
\end{equation}
Proof: Consider any maximal clique $K$ in $G$.  As observed above we
have $K=Q_u\cup Q_v$ for some edge $uv\in E(G)$ with $Q_u\neq\es$ and
$Q_v\neq\es$.  Suppose that $|K|=q$.  The hypothesis that every
$q$-clique is balanced implies that $X$ contains exactly four vertices
from $K$, so $|K\sm X|=|K|-4=q-4$.  Now suppose that $|K|\le q-1$.
The definition of $X$ implies that either $|K|\ge 3$ and $X$ contains
at least two vertices from $Q_u$ and one from $Q_v$, or vice-versa, or
$|K|=2$ and $X$ contains one vertex from each of $Q_u,Q_v$, and in any
case we have $|K\sm X|\le q-4$.  Thus (\ref{omm4}) holds.

\medskip

By (\ref{omm4}) and the induction hypothesis we have $\chi(G\sm X) \le
\lceil \frac{5}{4} \omega(G\sm X)\rceil = \lceil \frac{5}{4}
(q-4)\rceil = \lceil \frac{5}{4} q\rceil -5$.  By
Lemma~\ref{lem:2blowup} we know that $G[X]$ is $5$-colorable.  We can
take any $\chi(G\sm X)$-coloring of $G\sm X$ and use five new colors
for the vertices of $X$, and we obtain a coloring of $G$.  It follows
that $\chi(G) \le \lceil \frac{5}{4} q\rceil$ as desired.

Therefore we may assume that some $q$-clique of $G$ is not balanced,
say, up to symmetry, the clique $Q_{z a}$, with $|Q_{z}|\ge q-1$ and
$|Q_{a}|\le 1$.  So we also have $|Q_{b}|\le 1$ and $|Q_{c}|\le 1$.

Suppose that both $Q_{a w_1}$ and $Q_{a w_4}$ are $q$-cliques.  So
$|Q_{w_1}|\ge q-1$ and $|Q_{w_4}|\ge q-1$.  This implies $|Q_{w_j}|\le
1$ for each $j\in\{2,3,5,6\}$.  It follows that each of the cliques
$Q_{bw_2}$, $Q_{bw_5}$, $Q_{cw_3}$, $Q_{cw_6}$, $Q_{w_2w_3}$,
$Q_{w_5w_6}$ has size at most~$2$, so they are not $q$-cliques.  Then
$\{z,w_1,w_4\}$ is a good stable set.

Therefore we may assume that one of $Q_{a w_1}$ and $Q_{a w_4}$ is not
a $q$-clique.  Likewise, one of $Q_{b w_2}$ and $Q_{b w_5}$ is not a
$q$-clique, and one of $Q_{cw_3}$ and $Q_{c w_6}$ is not a $q$-clique.
This implies, up to symmetry, that we have either: (a) each of $Q_{a
w_1}$, $Q_{bw_5}$, $Q_{cw_3}$ is not a $q$-clique, or (b) each of
$Q_{aw_1}$, $Q_{bw_2}$, $Q_{cw_3}$ is not a $q$-clique.  In case (a),
we see that $\{z,w_2,w_4,w_6\}$ is a good stable set of $G$.  Hence
assume that we are in case~(b) and not in case~(a), and so $Q_{bw_5}$
is a $q$-clique, and so $|Q_{w_5}|\ge q-1$.  Hence $|Q_{w_4}|\le 1$
and $|Q_{w_6}|\le 1$.  It follows that $Q_{aw_4}$ and $Q_{cw_6}$ are
cliques of size at most~$2$, so they are not $q$-cliques.  Now
$Q_{aw_4}$, $Q_{bw_2}$, and $Q_{cw_6}$ are not $q$-cliques, so we are
in a situation similar to case~(a).  This completes the proof.
\end{proof}

We immediately have the following.

\begin{cor}\label{cor:bu-c5}
If $G$ is any blowup of
$C_5$, then $\chi(G)\le
\lceil \frac{5}{4} \omega(G)\rceil$. \hfill{$\Box$}
\end{cor}

\bigskip

\noindent{\bf Blowups of $F_3$}

Consider the graph $F_3$  as shown in Figure~\ref{fig:f123}.

\begin{lemma}\label{lem:2bu-f3}
Let $G$ be the $2$-blowup of $F_3$.  Then $\chi(G)=7$.
\end{lemma}
\begin{proof}
For each $v\in V(F_3)$ we call $v$ and $v'$ the two vertices of $Q_v$
in $G$.  The seven sets $\{x,v_4,v_6\}$, $\{y,v_2,v'_6\}$,
$\{z,v'_2,v'_4\}$, $\{x',v_5\}$, $\{y',v_1\}$, $\{z',v_3\}$ and
$\{v'_1,v'_3,v'_5\}$ form a $7$-coloring of $G$.  Hence $\chi(G) \leq
7$.  On the other hand we see that $\chi(G[Q_{v_1}\cup Q_{v_2}\cup
Q_{v_3}\cup Q_y\cup Q_z])\ge 5$ since that subgraph has $10$ vertices
and no stable set of size~$3$, and consequently $\chi(G[Q_x\cup
Q_1\cup Q_2\cup Q_3\cup Q_y\cup Q_z])\ge 7$.  Hence $\chi(G) \geq 7$.
\end{proof}

We say that $G$ is a \emph{special blowup} of $F_3$ if (up to
symmetry) we have $|Q_u|\le 1$ for each $u\in\{x,v_4,v_5,v_6\}$ and
$|Q_v|=t$ for each $v\in\{v_1,v_2,v_3,y,z\}$, for some integer $t\ge
2$.
\begin{lemma}\label{lem:sbu-f3}
Let $G$ be a special blowup of $F_3$.  Then
$\chi(G)\le\lceil\frac{5}{4}\omega(G)\rceil$.
\end{lemma}
\begin{proof}
We prove the theorem by induction on $|V(G)|$.  If $Q_x\cup
Q_{v_4}\cup Q_{v_5}\cup Q_{v_6}=\es$, then $G$ is a blowup of $C_5$,
so the lemma holds by Corollary~\ref{cor:bu-c5}.  Hence assume that
$Q_x\cup Q_{v_4}\cup Q_{v_5}\cup Q_{v_6}\neq\es$.  It follows that
$\omega(G)=2t+1$.  Let $X$ be a subset of $V(G)$ obtained by taking
two vertices from $Q_v$ for each $v\in \{v_1,v_2,v_3,y,z\}$ and the set $Q_x\cup Q_{v_4}\cup Q_{v_5}\cup Q_{v_6}$.
Then $\omega(G\sm X)=2t-4=\omega(G)-5$.  In $F_3$
the six sets $\{v_1,v_3,v_5\}$, $\{v_2,y\}$, $\{v_2,z\}$, $\{v_1,y\}$,
$\{v_3,z\}$ and $\{x,v_4, v_6\}$ are such that every vertex from
$\{v_1,v_2,v_3,y,z\}$ belongs to two of them and every vertex from
$\{x,v_4,v_5,v_6\}$ belongs to one of them; hence they are equivalent
to a $6$-coloring of $G[X]$.  We can take any $\chi(G\sm X)$-coloring
of $G\sm X$ and use six new colors for $X$, and we obtain a coloring
of $G$.  Hence $\chi(G)\le\chi(G\sm X)+6 \le
\lceil\frac{5}{4}(\omega(G)-5)\rceil+6 = \lceil\frac{5}{4}
\omega(G)-\frac{25}{4}\rceil+6\ \le \lceil\frac{5}{4}\omega(G)\rceil$.
\end{proof}

\begin{theorem}\label{thm:bu-f3}
If $G$ is any blowup of $F_3$, then $\chi(G)\le \lceil \frac{5}{4}
\omega(G)\rceil$.
\end{theorem}
\begin{proof}
Let $q=\omega(G)$.  We prove the theorem by induction on $|V(G)|$.
Obviously the theorem holds if $G$ is any induced subgraph of $F_3$.  Now
suppose that $G$ is not an induced subgraph of $F_3$.  By
Theorem~\ref{thm:tools}~(\ref{goods}) we may assume that $G$ has no
good stable set.

Note that every maximal clique of $G$ consists of $Q_u\cup Q_v\cup
Q_w$ for some triangle $\{u,v,w\}$ in $F_3$, and we denote it as
$Q_{uvw}$.  We say that such a maximal clique is \emph{balanced} if
$|Q_u|\ge 2$, $|Q_v|\ge 2$, and $|Q_w|\ge 2$.

Suppose that every $q$-clique of $G$ is balanced.  Let $X$ be a subset
of $V(G)$ obtained by taking $\min\{2, |Q_v|\}$ vertices from $Q_v$
for each $v\in V(F_3)$.  The hypothesis that every $q$-clique is
balanced implies that $X$ contains exactly six vertices from each
$q$-clique of $G$, so $\omega(G\sm X)=\omega(G)-6$.  By the induction
hypothesis we have $\chi(G\sm X) \le \lceil \frac{5}{4} \omega(G\sm
X)\rceil = \lceil \frac{5}{4} (q-6)\rceil = \lceil \frac{5}{4} q-
\frac{30}{4}\rceil \le \lceil \frac{5}{4} q\rceil -7$.  By
Lemma~\ref{lem:2bu-f3} we know that $G[X]$ is $7$-colorable.  We can
take any $\chi(G\sm X)$-coloring of $G\sm X$ and use seven new colors
for the vertices of $X$, and we obtain a coloring of $G$.  It follows
that $\chi(G) \le \lceil \frac{5}{4}q\rceil$ as desired.  Therefore we
may assume that some $q$-clique of $G$ is not balanced.

\medskip

For each $v\in V(F_3)$, let $R_v$ consist of one vertex from $Q_v$ if
$Q_v\neq\es$, otherwise let $R_v=\es$.  We claim that we may assume
that:
\begin{equation}\label{qxv}
\mbox{Each of $Q_x$, $Q_y$ and $Q_z$ is non-empty.}
\end{equation}
Proof: Suppose up to symmetry that $Q_x=\es$.  If also $Q_{v_2}=\es$,
then $G$ is a blowup of $F_3\setminus\{x,v_2\}$, which is a chordal
graph, so $\chi(G)=\omega(G)$ and the theorem holds.  Therefore
 $Q_{v_2}\neq\es$. Likewise, $Q_{v_1}\neq\es$ and $Q_{v_3}\neq\es$.
Since $R_{v_1}\cup R_{v_3}\cup R_{v_5}$ is not a good stable set, we have $Q_{v_5}=\es$.
 Moreover, if $Q_{v_4}\cup Q_{v_6}=\es$, then $G$ is a blowup of $C_5$, and the theorem holds by Corollary~\ref{cor:bu-c5}.
  So up to symmetry we may assume that  $Q_{v_4}\neq\es$. Now if $Q_z=\es$, then $G$ is a blowup of $F_3\setminus\{x,z,v_5\}$, which is a chordal
graph, so $\chi(G)=\omega(G)$ and the theorem holds.
  So suppose that $Q_z\neq\es$. Then $R_{v_2}\cup R_{v_4} \cup R_z$ is a good stable set of $G$. Hence we may assume that (\ref{qxv}) holds.

\medskip

We claim that we may assume that:
\begin{equation}\label{qc}
\longbox{Each of $Q_{xyz}$, $Q_{xyv_3}$, $Q_{yzv_5}$, $Q_{zxv_1}$,
$Q_{xv_1v_2}$, $Q_{yv_3v_4}$ is a $q$-clique, and either $Q_{zv_5v_6}$
or  $Q_{xv_2v_3}$ is a $q$-clique.}
\end{equation}
Proof: If two of $R_{v_1}, R_{v_3}, R_{v_5}$ are empty, say
$R_{v_1}\cup R_{v_3}=\es$, then $G$ is a blowup of
$F_3\sm\{v_1,v_3\}$, which is a chordal graph, so $\chi(G)=\omega(G)$.
So at least two of $R_{v_1}, R_{v_3}, R_{v_5}$ are non-empty.  Since
$R_{v_1}\cup R_{v_3}\cup R_{v_5}$ is not a good stable set, there is a
$q$-clique in $G\sm (R_{v_1}\cup R_{v_3}\cup R_{v_5})$, and this
clique can only be $Q_{xyz}$.  Now consider the stable set $R_{x46} =
R_x\cup R_{v_4}\cup R_{v_6}$, which is not empty by (\ref{qxv}).
Since it is not a good stable set, there is a $q$-clique in $G\sm
R_{x46}$, and so $Q_{yzv_5}$ is a $q$-clique.  Likewise, $Q_{xyv_3}$
and $Q_{zxv_1}$are $q$-cliques.  Now consider the stable set $R_x\cup
R_{v_5}$.  Since it is not a good stable set, we deduce that one of
$Q_{yv_3v_4}$ and $Q_{zv_6v_1}$ is a $q$-clique.  Likewise, one of
$Q_{zv_5v_6}$ and $Q_{xv_2v_3}$ is a $q$-clique, and one of
$Q_{xv_1v_2}$ and $Q_{yv_4v_5}$ is a $q$-clique.  Up to symmetry this
yields the possibilities described in (\ref{qc}).  Thus we may assume
that (\ref{qc}) holds.

\medskip

Next we claim that we may assume that:
\begin{equation}\label{qno}
\mbox{ $Q_{zv_5v_6}$ is not a $q$-clique.}
\end{equation}
Proof: Suppose not.

First we show that we may assume that $|Q_{v_1}|\geq 2$. Suppose that $|Q_{v_1}|= \varepsilon \leq 1$. Let $a=|Q_{v_2}|$ and $b=|Q_x|$.  Since $Q_{xv_1v_2}$ is a $q$-clique,
we have $a+b+\varepsilon=q$.  Then, using the $q$-cliques given by
(\ref{qc}), we deduce successively that $|Q_z|=a$,
$|Q_y|=\varepsilon$, $|Q_{v_5}|=b$, $|Q_{v_6}|=\varepsilon$,
$|Q_{v_3}|=a$, and $|Q_{v_4}|=b$.  We have $|Q_{xv_2v_3}|= b+2a\le
q=a+b+\varepsilon$, so $a\le\varepsilon$.  Also we have
$|Q_{yv_4v_5}|= 2b+\varepsilon\le q=a+b+\varepsilon$, so $b\le a$.
Hence $b\le a\le\varepsilon\le 1$, which means that $G$ is isomorphic
to an induced subgraph of $F_3$, so the theorem holds.  So we may assume that $|Q_{v_1}|\geq 2$. Likewise,  we may assume that $|Q_{v_3}|\geq 2$, and $|Q_{v_5}|\geq 2$.

Next  we may assume that $|Q_{x}|\geq 2$ (otherwise since $Q_{xyz}$ and $Q_{yzv_5}$ are $q$-cliques (by (\ref{qc})), we have $|Q_{v_5}| \leq 1$, a contradiction). Likewise,  we have $|Q_{y}|\geq 2$ and  $|Q_{z}|\geq 2$.

Further, we may assume that  $|Q_{v_6}|\geq 2$ (otherwise since by (\ref{qc}) and by our assumption, $Q_{yzv_5}$ and $Q_{zv_5v_6}$ are $q$-cliques, we have $|Q_{y}| \leq 1$, a contradiction).  Likewise,  we have $|Q_{v_2}|\geq 2$ and  $|Q_{v_4}|\geq 2$.

Hence  the above analysis shows that every $q$-clique in $G$ is balanced, and the theorem holds as above. Thus we may assume
that (\ref{qno}) holds.

\medskip

Now by (\ref{qc}) and (\ref{qno}), we may assume that $Q_{xv_2v_3}$ is a $q$-clique. Let $a=|Q_{v_5}|$, $b=|Q_z|$ and $t=|Q_y|$. Then by (\ref{qc}),
$a+b+t=q$, and by using the $q$-cliques given by (\ref{qc}), we deduce successively that $|Q_x|=a$, $|Q_{v_1}|=t$ and $|Q_{v_2}|=b$.
Then again by (\ref{qc}) and by our assumption, since $Q_{xv_2v_3}$ and $Q_{xyv_3}$ are $q$-cliques, we see that $|Q_{v_3}|=b=t$.  So, $q=a+2t$. Since
$Q_{yv_3v_4}$ is a $q$-clique (by (\ref{qc})), we have $|Q_{v_4}|=a$. Thus $|Q_{yv_4v_5}|=2a+t\leq q=a+2t$, so $a\leq t$.
First suppose that $t\leq 1$. Then $a\leq 1$ and hence $q\leq 3$. This implies that, we may assume that $|Q_{v_6}|\leq 1$ (otherwise since $Q_{zv_5v_6}$ is not a $q$-clique (by (\ref{qno})),  $a+2t>a+t+|Q_{v_6}|$, and hence $t\geq 2$ which is a contradiction.). Thus  $G$ is an induced subgraph of $F_3$ and the theorem holds. So suppose that $t\geq 2$. Since some $q$-clique of $G$ is not balanced, there is a vertex $w\in
\{x,v_4,v_5\}$ such that $|Q_w|\le 1$. In any case, we have $a\leq 1$, and hence $q\leq 2t+1$. Now $|Q_{v_6v_1z}| = |Q_{v_6}|+2t \leq q\leq 2t+1$, so $|Q_{v_6}|\leq 1$. Hence the above analysis shows that $G$ is a special blowup of $F_3$, so the theorem holds as
a consequence of Lemma~\ref{lem:sbu-f3}.
\end{proof}

\bigskip

\noindent{\bf Blowups of $H_2,H_3, H_4$ and $H_5$}

Let $H_2,\ldots, H_5$ be the graphs as shown in Figure~\ref{fig:h12345}.

\begin{theorem}\label{bu-h2}
Let $G$ be any blowup of $H_2$.  Then $\chi(G)\le \lceil
\frac{5}{4}\omega(G)\rceil$.
\end{theorem}
\begin{proof}
By the definition of a blowup, $V(G)$ is partitioned into cliques
$Q_v$, $v\in V(H_2)$.  If $Q_v\neq\emptyset$ we call $v$ one vertex of
$Q_v$, and if $|Q_v|\ge 2$ we call $v'$ a second vertex of $Q_v$.   We denote, e.g., the clique $Q_a\cup Q_{v_1}\cup Q_{v_2}$ by
$Q_{av_1v_2}$, etc. Let
$q=\omega(G)$. We prove the theorem by induction on $|V(G)|$.  By
Theorem~\ref{thm:tools} we may assume that every vertex $x\in V(G)$
satisfies $d(x)\ge \lceil\frac{5}{4}q\rceil$ and that $G$ has no good
stable set.

Suppose that $Q_{v_1}\cup Q_{v_2}=\emptyset$.  If $Q_b\neq\emptyset$,
then $\{b\}$ is a good stable set.  If $Q_b=\emptyset$, then $G$ is a
blowup of $C_5$, and the result follows from
Corollary~\ref{cor:bu-c5}.  Hence we may assume that $Q_{v_1}\cup
Q_{v_2}\neq\emptyset$.  Then both $Q_{v_1}$ and $Q_{v_2}$ are
non-empty, for otherwise some vertex in $Q_{v_1}\cup Q_{v_2}$ is
simplicial (and so has degree less than $q$).  Since $N[v_1]$ is
partitioned into the two cliques $Q_{v_6}$ and $Q_{av_1v_2}$, and $d(v_1)\ge \lceil\frac{5}{4}q\rceil$, we deduce that
$|Q_{v_6}|\ge \lceil\frac{q}{4}\rceil+1\ge 2$; and similarly (since
$N[v_1]$ is also partitioned into cliques $Q_{av_1v_6}$ and $Q_{v_2}$) we have $|Q_{v_2}|\ge
\lceil\frac{q}{4}\rceil+1\ge 2$.  Likewise $|Q_{v_3}|\ge 2$ and
$|Q_{v_1}|\ge 2$.  By the same argument we may assume that both
$Q_{v_4}$ and $Q_{v_5}$ are non-empty, and consequently $|Q_{v_4}|\ge
2$ and $|Q_{v_5}|\ge 2$.

If $Q_c = \emptyset$, then $G$ is a blow-up of $F_3$, and the
theorem follows from Theorem~\ref{thm:bu-f3}. So we may assume that $|Q_c|\geq 1$.
Then the set of maximal cliques  of $G$ is $\{$$Q_{av_1v_6}$, $Q_{av_1v_2}$, $Q_{av_2v_3}$,  $Q_{bv_3v_4}$, $Q_{bv_4v_5}$, $Q_{bv_5v_6}$, $Q_{abcv_3}$, $Q_{abcv_6}$$\}$.

Suppose that $|Q_c|\ge 2$.  Consider the five stable sets
$\{v_1,v_3,v_5\}$, $\{v_2,v_4,v_6\}$, $\{c,v'_1,v'_5\}$,
$\{c',v'_2,v'_4\}$, and $\{v'_3,v'_6\}$.  Then
every maximal clique of $G$ contains four vertices from their union;
so the result follows from Theorem~\ref{thm:tools}~(\ref{fives}) (with
$t=5$).    Therefore let us assume
that $|Q_c|=1$.

Suppose that both $Q_a$ and $Q_b$ are non-empty.  Consider the six
stable sets $\{v_1,v_3,v_5\}$, $\{v_2,v_4,v_6\}$, $\{v'_3,v'_6\}$,
$\{a,v'_5\}$, $\{b,v'_2\}$ and $\{c,v'_1,v'_4\}$.  Then  every maximal clique of $G$ contains five vertices from their
union; so the result follows from
Theorem~\ref{thm:tools}~(\ref{fives}) (with $t=6$).

Therefore we may assume up to symmetry that $Q_a=\emptyset$.  Note
that $Q_{bcv_3}$ is not a $q$-clique of $G$, because
$Q_{bv_3v_4}$ is a clique and $|Q_{v_4}|>|Q_c|$.
Likewise, $Q_{bcv_6}$ is not a $q$-clique of $G$.
Consider the five stable sets $\{v_1,v_3,v_5\}$, $\{v_2,v_4, v_6\}$,
$\{v'_2,v'_4\}$, $\{v'_3,v'_6\}$ and $\{c, v'_1,v'_5\}$.  Then  every maximal clique of $G$ contains four vertices
from their union, except for $Q_{bcv_3}$ and
$Q_{bcv_6}$, which contain only three vertices from
their union, but we know that these two are not $q$-cliques.  It
follows that $\omega(G\setminus X)\le q-4$, so the result follows from
Theorem~\ref{thm:tools}~(\ref{fives}).
\end{proof}

\begin{theorem}\label{bu-h3}
Let $G$ be any blowup of $H_3$.  Then $\chi(G)\le \lceil
\frac{5}{4}\omega(G)\rceil$.
\end{theorem}
\begin{proof}
By the definition of a blowup, $V(G)$ is partitioned into nine cliques
$Q_i$, $v_i\in V(H_3)$.  If $Q_i\neq\emptyset$ we call $v_i$ one
vertex of $Q_i$.  Note that every maximal clique of $G$ consists of $Q_u\cup Q_v\cup Q_w$ for some triangle $\{u,v,w\}$ in $H_3$.
If each of $Q_1, Q_4, Q_7$ is non-empty, then
$\{v_1,v_4,v_7\}$ is a good stable set of $G$, and the result follows
from Theorem~\ref{thm:tools}~(\ref{goods}).  Hence we may assume that
one of $Q_1, Q_4, Q_7$ is empty.  Likewise we may assume that one of
$Q_2, Q_5, Q_8$ is empty, and that one of $Q_3, Q_6, Q_9$ is empty.
Up to symmetry and relabelling, this yields the following two cases.

\noindent (i) $Q_i\cup Q_{i+1}=\emptyset$ for some $i$.  Then $G$ is a
chordal graph, so $\chi(G)=\omega(G)$.

\noindent (ii) $Q_i\cup Q_{i+2}\cup Q_{i+4}=\emptyset$ for some $i$.
Then $G$ is a blowup of $C_5$, and the result follows from
Corollary~\ref{cor:bu-c5}.
\end{proof}

\begin{theorem}\label{thm:bu-h5}
Let $G$ be a blowup of $H_5$.  Then $\chi(G)\le
\lceil\frac{5}{4}\omega(G)\rceil$.
\end{theorem}
\begin{proof}
By the definition of a blowup, $V(G)$ is partitioned into ten cliques
$Q_v$, $v\in V(H_5)$.  Note that if $Q_{t_{i-1}}\cup
Q_{t_{i+1}}=\emptyset$ for some $i$, then the vertices of $Q_{t_i}$
can be moved to $Q_{v_i}$, so we may assume in that case that
$Q_{t_i}=\emptyset$ too.  Let $q = \omega(G)$.  We prove the theorem
by induction on $|V(G)|$.

If $Q_{v_i}\cup Q_{t_i}=\emptyset$ for some $i$, then  $G$ is a chordal graph (as it is a blowup of a chordal graph), so $\chi(G)=\omega(G)$.  Hence let us
assume that $Q_{v_i}\cup Q_{t_i}\neq\emptyset$ for all $i$.  For each
$i$ let $x_i=t_i$ if $Q_{t_i}\neq\emptyset$, else let $x_i=v_i$.  In
any case if $d(x_i)< \lceil\frac{5}{4}q\rceil$ then we can conclude
using Theorem~\ref{thm:tools}~(\ref{degq}) and induction.  Hence
assume that $d(x_i)\ge \lceil\frac{5}{4}q\rceil$ for all $i$.  If
$x_i=t_i$, then $N[x_i]$ is partitioned into the two sets
$Q_{v_{i-1}}$ and $Q_{v_i}\cup Q_{t_i}\cup Q_{v_{i+1}}$, and the
latter set is a clique (of size at most $q$), so the inequality
$d(x_i)\ge \lceil\frac{5}{4}q\rceil$ implies $|Q_{v_{i-1}}|\ge
\lceil\frac{q}{4}\rceil+1\ge 2$.  Similarly $|Q_{v_{i+1}}|\ge
\lceil\frac{q}{4}\rceil+1\ge 2$.  On the other hand suppose that
$x_i=v_i$ (i.e., $Q_{t_i}=\emptyset$).  If $Q_{t_{i-2}}\neq\emptyset$
then the same argument implies $|Q_{v_{i-1}}|\ge 2$; while if
$Q_{t_{i-2}}=\emptyset$, then, as observed above, we have
$Q_{t_{i-1}}=\emptyset$, so the same argument (about $v_i)$, implies
$|Q_{v_{i-1}}|\ge 2$ again.  Hence in all cases we have $|Q_{v_j}|\ge
2$ for all $j$.

For each $i$ let $u_i, v_i$ be two vertices in $Q_{v_i}$.  Consider
the five stable sets $\{u_i, v_{i+2}\}$ $(i=1,\ldots,5)$, and let $X$ be
their union.  Any maximal clique $K$ of $G$ is included in
$Q_{v_i}\cup Q_{v_{i+1}}$ for some $i$, and so $K$ contains $u_i, v_i,
u_{i+1}, v_{i+1}$.  So $\omega(G\setminus X)=q-4$ and we can conclude
using Theorem~\ref{thm:tools}~(\ref{fives}) (with $t=5$) and the
induction hypothesis.
\end{proof}

\begin{theorem}\label{bu-h4}
Let $G$ be any blowup of $H_4$.  Then $\chi(G)\le \lceil
\frac{5}{4}\omega(G)\rceil$.
\end{theorem}
\begin{proof}
By the definition of a blowup, $V(G)$ is partitioned into nine cliques
$Q_v$, $v\in V(H_4)$.  If $Q_v\neq\emptyset$ we call $v$ one vertex of
$Q$.  If $Q_{v_5}\cup Q_{v_6}=\emptyset$, then $G$ is a chordal graph,
so $\chi(G)=\omega(G)$.  Hence let us assume up to symmetry that
$Q_{v_5}\neq\emptyset$.  If $Q_{v_1}=\emptyset$, then $G$ is a blowup
of $H_5$, so the result follows from Theorem~\ref{thm:bu-h5}.  Hence
let us assume that $Q_{v_1}\neq\emptyset$.  If $Q_{v_3}=\emptyset$,
then $G$ is a blowup of $H_5$ again.  Hence let us assume that
$Q_{v_3}\neq\emptyset$.  Now it is easy to see that $\{v_1,v_3,v_5\}$
is a good stable set, so the result follows from
Theorem~\ref{thm:tools}~(\ref{goods}).
\end{proof}

\noindent{\bf Blowups of $F_{k,\ell}$}

\begin{theorem}\label{thm:bu-fkl}
For integers $k,\ell\ge 0$, let $G$ be any blowup of
$F_{k,\ell}$.  Then $\chi(G)\le \lceil \frac{5}{4}\omega(G)\rceil$.
\end{theorem}
\begin{proof}
We use the same notation as in the definition of $F_{k,\ell}$.  By the
definition of a blowup $V(G)$ is partitioned into cliques $Q_v,$ $v\in
V(F_{k,\ell})$, such that $[Q_u,Q_v]$ is complete if $uv\in
E(F_{k,\ell})$ and otherwise $[Q_u,Q_v]=\es$.  Let $Q_A=
\bigcup_{i=0}^k Q_{a_i}$ and $Q_B=\bigcup_{j=0}^\ell Q_{b_j}$.  Let
$D= \bigcup_{v\in U\cup W}Q_v$.  As a convention it is convenient, for
any $u\in V(F_{k,\ell})$ such that $Q_u\neq\es$, to use the name $u$
for one vertex of $Q_u$; moreover if $|Q_u|\ge 2$ we call $u'$ another
vertex from $Q_u$, and if $|Q_u|\ge 3$ we call $u''$ a third vertex
from $Q_u$.  We denote, e.g., the clique $Q_x\cup Q_y\cup Q_{u_i}$ by
$Q_{xyu_i}$, etc.  Let $q=\omega(G)$.  We prove the lemma by induction
on $|V(G)|+k+\ell$.  We may assume that $G$ does not satisfy any of
the hypotheses (\ref{degq})--(\ref{stable}) of
Theorem~\ref{thm:tools}, for otherwise we can find a
$\lceil\frac{5}{4}q\rceil$-coloring of $G$ using induction.

We remark that if $k>0$ and $Q_{u_i}=\es$ for some
$i\in\{1,\ldots,k\}$, then the vertices of $Q_{a_i}$ can be moved to
$Q_{a_0}$, and so $G$ is a blowup of $F_{k-1,\ell}$ and the result
holds by induction.  Moreover, if $k>0$ and either $|Q_{a_i}|\le
\lceil\frac{q}{4}\rceil$ for some~$i$, or $|Q_y|\le
\lceil\frac{q}{4}\rceil$, then, since $N[u_i]= Q_{a_i}\cup Q_{u_i}\cup
Q_x\cup Q_y$ and $Q_{a_iu_ix}$ and $Q_{u_ixy}$ are cliques that
contain $u_i$, we have $d(u_i)\le q-1+ \lceil\frac{q}{4}\rceil
<\lceil\frac{5}{4}q\rceil$, so the result holds by
Theorem~\ref{thm:tools}~(\ref{degq}).  In summary, we may assume that:
\begin{equation}\label{sizeq}
\longbox{If $k>0$ then $Q_{u_i}\neq\es$ and $|Q_{a_i}|>
\lceil\frac{q}{4}\rceil$ for all $i$, and $|Q_y|>
\lceil\frac{q}{4}\rceil$.  Also if $\ell>0$ then $Q_{w_j}\neq\es$ and
$|Q_{b_j}|> \lceil\frac{q}{4}\rceil$ for all $j$ and $|Q_x|>
\lceil\frac{q}{4}\rceil$.}
\end{equation}
It follows from (\ref{sizeq}) that $k\le 3$, for otherwise $|Q_A|>q$;
and similarly $\ell\le 3$.  Moreover, if $\ell>0$ then $k\le 2$, for
otherwise $|Q_A\cup Q_x|>q$; and similarly if $k>0$ then $\ell\le 2$.
We assume up to symmetry that $k\le \ell$.  Consequently we have
either $k=0$ and $\ell\le 3$, or $k=1$ and $\ell\in \{1,2\}$, or
$k=\ell=2$.  In any case $k\le 2$.  If $k\le 1$ and $\ell\le 1$, then
$G$ is a blowup of (an induced subgraph of) $H_5$, so the result
follows from Theorem~\ref{thm:bu-h5}.  So we may assume that $\ell\ge
2$.  Consequently we have either $k=0$ and $\ell\in\{2,3\}$, or $k=1$
and $\ell=2$, or $k=\ell=2$.

Suppose that $Q_A=\es$.  Then $Q_z=\es$, for otherwise $d(z)\le q-1$,
and $Q_y=\es$, for otherwise $\{y\}$ is a good stable set.  Then we
can view $G$ as a blowup of $F_{0,\ell-1}$ (putting $Q_{b_\ell}$ and
$Q_{w_\ell}$ in the role of $Q_z$ and $Q_{a_0}$ respectively) and use
induction.  Therefore we may assume that $Q_A\neq\es$.  If $k\ge 1$,
then $|Q_{a_1}|\ge 2$ by (\ref{sizeq}), and if $k=0$ then
$|Q_{a_0}|\ge 2$, for otherwise either $d(z)\le q$ (if $Q_{z}\neq\es$)
or $d(a_0)\le q$ (if $Q_{z}=\es$).  Hence in any case we have
$|Q_A|\ge 2$.  Let $a,a'$ be two vertices from $Q_A$, chosen as
follows: if $k=0$, let $a,a'\in Q_{a_0}$.  If $k=1$, let $a,a'\in
Q_{a_1}$.  If $k= 2$, let $a\in Q_{a_1}$ and $a'\in Q_{a_2}$.

\medskip

Let $p=\max\{|Q_v|, v\in U\cup W\}$.  So $p\ge 1$.  We claim that:
\begin{equation}\label{p2}
\mbox{We may assume that $p\ge 2$.}
\end{equation}
Proof: Suppose that $p=1$; so $Q_v=\{v\}$ for all $v\in U\cup W$.  If
$|Q_z|\le 1$, then $U\cup W\cup Q_z$ is a stable set, and $G\sm (U\cup
W\cup Q_z)$ is perfect (a blowup of $P_4$), so the result follows from
Theorem~\ref{thm:tools}~(\ref{stable}).  Hence $|Q_z|\ge 2$.  Define
five stable sets as follows: Let $T_1=\{a,b_1\}$, $T_2=\{b_2,x\}$,
$T_3=\{z,x'\}$, $T_4=\{a',y\}$, and $T_5=\{z',y'\}$, where $y,y'\in
Q_y$, with the convention that $y'$ vanishes if $|Q_y|=1$, and in that
case if $|Q_x|\ge 3$ then $T_5=\{z',x''\}$ for some $x''\in
Q_x\sm\{x,x'\}$, and $y$ too vanishes if $Q_y=\es$.  Let $T^*
=T_1\cup\cdots\cup T_5$.  We claim that every maximal clique $K$ of
$G$ satisfies $|K\sm T^*|\le q-4$.  The following cases (i)--(vii)
occur: \\
(i) $K=Q_z\cup Q_A$.  Then $K$ contains four vertices ($z,z',a,a'$)
from $T^*$, so $|K\sm T^*|\le q-4$.  Likewise, if $K=Q_z\cup Q_B$,
then $K$ contains $z,z',b_1,b_2$.  \\
(ii) $K=Q_x\cup Q_A$.  Then $K$ contains $a,a',x, x'$ from $T^*$.  \\
(iii) $K=Q_y\cup Q_B$.  Then $Q_y\neq\es$ because $Q_B$ is not a
maximal clique (since $Q_z\neq\es$).  If $|Q_y|\ge 2$, then $K$
contains four vertices $b_1,b_2,y,y'$ from $T^*$.  If $|Q_y|=1$ then
(since $|Q_z|\ge 2$) $|K|<|Q_z\cup Q_A|\le q$, and $K$ contains three
vertices $b_1,b_2,y$ from $T^*$, so $|K\sm T^*|\le q-4$.  \\
(iv) $k\ge 1$ and $K=Q_{xyu_i}$ for some $i\in\{1, \ldots, k\}$.  Then
$Q_y\neq\es$ because $Q_{xu_i}$ is not a maximal clique (since
$Q_{a_i}\neq\es$).  If $|Q_y|\ge 2$, then $K$ contains four vertices
($x,y,x',y'$) from $T^*$.  If $|Q_y|=1$, then (since $|Q_{a_i}|\ge 2$)
$|K|<|Q_{xu_ia_i}|\le q$ and $K$ contains three vertices $x,x',y$ from
$T^*$.  \\
(v) $k\ge 1$ and $K= Q_{xa_iu_i}$ for some $i\in\{1,\ldots,k\}$, say
$i=1$.  If $k=1$ then $K$ contains $x,x', a,a'$.  Suppose $k= 2$.
Since $q\ge|Q_{xa_1a_2}|$, and $|Q_{a_2}|\ge 2$, we have $|K|\le q-1$.
Then $K$ contains three vertices $x,x',a$ from $T^*$, so $|K\sm
T^*|\le q-4$.  \\
(vi) $K=Q_{xyw_j}$ for some $j\in\{1,\ldots,\ell\}$.  If $|Q_y|\ge 2$
then $K$ contains four vertices ($x,y,x',y'$) from $T^*$.  If
$|Q_y|\le 1$, then $K$ contains at least two vertices from $T^*$, so
if $|K|\le q-2$ we are done.  If $|K|\ge q-1$, then $|Q_x|+2\ge |K|\ge
q-1\ge |Q_z\cup Q_B|-1\ge 2(\ell+1)-1\ge 5$, so $|Q_x|\ge 3$, so the
vertex $x''$ exists and $K$ contains three vertices $x,x',x''$ from
$T^*$.  \\
(vii) $K=Q_{yb_jw_j}$ for some $j\in\{1,\ldots,\ell\}$.  Since $q\ge
|Q_z\cup Q_B|$, we have $|Q_{b_j}|\le q-2\ell$.  If $\ell= 3$, then
either $|K|\le q-4$, or $|K|=q-3$ and $Q_y\neq\es$ and $K$ contains
$y$ from $T^*$.  Hence suppose that $\ell=2$.  So $b_j\in K$.  Then
either $|K|\le q-3$, or $|K|=q-2$ and $K$ also contains $y$ from
$T^*$.  So $|K\sm T^*|\le q-4$.  \\
In either case Theorem~\ref{thm:tools}~(\ref{fives}) implies the
desired result.  Thus (\ref{p2}) holds.

\medskip

Suppose that $k\le 1$.  We know that $\ell\in\{2,3\}$.  By
(\ref{sizeq}) we have $|Q_{b_j}|\ge \lceil\frac{q}{4}\rceil +1$ for
all $j\in\{1,\ldots,\ell\}$.  Recall that $Q_A\neq\es$.  Let $a^*=a_0$ if
$Q_{a_0}\neq\es$ and $a^*=a_1$ otherwise.  In either case the set
$N(a^*)$ can be partitioned into two cliques such that $Q_z$ is one of
them.  By Theorem~\ref{thm:tools}~(\ref{degq}) we may assume that
$d(a)\ge \lceil\frac{5}{4}q\rceil$, so $|Q_{z}|\ge
\lceil\frac{q}{4}\rceil +1$.  Consequently $q\ge |Q_z\cup Q_B|\ge
(\ell+1) (\lceil\frac{q}{4}\rceil +1)$.  The inequality $q\ge (\ell+1)
(\lceil\frac{q}{4}\rceil +1)$ is violated if $\ell\ge 3$, so $\ell=2$.
Moreover, the inequality with $\ell=2$ implies $q\ge 12$.  Hence
(\ref{sizeq}) yields that $|Q_x|\ge 3$, and $|Q_{b_j}|\ge 3$ for each
$j\in\{1,2\}$, and $|Q_A|\ge 3$, and similarly $|Q_z|\ge 3$.

Suppose that $k=0$.  We may assume that $Q_{xyw_j}$ is a $q$-clique
for each $j\in\{1,2\}$, for otherwise the set $\{a_0, b_j, w_{3-j}\}$
is a good stable set.  Hence $|Q_{w_1}|=|Q_{w_2}|=p$.  Note that the set of maximal cliques of $G$ is $\{Q_{za_0}, Q_{xa_0}, Q_{xyw_1},$ $ Q_{xyw_2}, Q_{yw_1b_1}, Q_{yw_2b_2}, Q_{zb_0b_1b_2}\}$ plus
$Q_{yb_0b_1b_2}$ if $Q_y\neq\es$.  Let $S_1=\{b_1,
w_2, a_0\}$, $S_2=\{b_2, w_1, a'_0\}$, $S_3=\{z, w'_1, w'_2\}$, and
$S_4=\{b'_1, x\}$.  If $Q_y\neq\es$, let $S_5=\{a''_0, y\}$.  If
$Q_y=\es$, then one of $Q_{w_1b_1}, Q_{w_2b_2}$ is a $q$-clique, for
otherwise $\{x,z\}$ is a good stable set; so for some $j\in\{1,2\}$ we
have $|Q_{w_jb_j}|=q\ge |Q_{b_1b_2}|$, whence $p=|Q_{w_j}|\ge
|Q_{b_{3-j}}|\ge 3$; so we let $S_5=\{a''_0, w''_1, w''_2\}$.  In
either case, $S_1,\ldots, S_5$ are stable sets and it is easy to see
that every maximal clique of $G$ contains at least four vertices from their
union, so the result follows from
Theorem~\ref{thm:tools}~(\ref{fives}).

Now suppose that $k=1$, and so $\ell=2$.  By (\ref{sizeq}), we have $|Q_y|\geq 2$. Note that the set of maximal cliques of $G$ is $\{Q_{za_0a_1}, Q_{xa_0a_1}, Q_{xa_1u_1}, Q_{xyu_1}, Q_{xyw_1}, Q_{xyw_2},$ $ Q_{yw_1b_1}, Q_{yw_2b_2}, Q_{yb_0b_1b_2}, Q_{zb_0b_1b_2}\}$. Let $S_1=\{b_1, w_2, u_1\}$
plus $a_0$ if $Q_{a_0}\neq\es$.  Let $S_2=\{b_2, w_1, a_1\}$,
$S_3=\{x,z\}$, $S_4=\{y',z'\}$, and $S_5=\{a'_1,y\}$.  In either case,
$S_1,\ldots, S_5$ are stable sets and  that every
maximal clique of $G$ contains at least four vertices from their union, so the result follows from
Theorem~\ref{thm:tools}~(\ref{fives}).

Finally suppose that $k=2$ and $\ell=2$.  Let
$S_1=\{a_1,b_1,u_2,w_2\}$, $S_2=\{a_2,b_2, u_1, w_1\}$,
$S_3=\{x,b'_1\}$, $S_4=\{y,a'_1\}$, and let $S_5$ consist of one
vertex from each component of $Q_z\cup (D\sm \{u_1,u_2,w_1,w_2\})$.
Let $S^*=S_1\cup\cdots\cup S_5$.  We claim that every maximal clique
$K$ of $G$ satisfies $|K\sm S^*|\le q-4$.  Indeed if $K=Q_x\cup Q_A$
then $K$ contains $x,a_1,a'_1,a_2$ from $S^*$.  If $K=Q_z\cup Q_A$
then $Q_z\neq\es$ and $K$ contains $z,a_1,a'_1,a_2$.  If
$K=Q_{xa_1u_1}$ then $K$ contains $x,a_1,a'_1,u_1$.  If
$K=Q_{xa_2u_2}$ then $K$ contains $x,a_2,u_2$ from $S^*$, so if
$|K|\le q-1$ we are done; and if $|K|=q$ then $|Q_{xa_2u_2}|=q\ge
|Q_{xa_1a_2}|$ so $|Q_{u_2}|\ge 2$, so $Q_{u_2}$ contains a vertex
$u'_2$ from $S_5$.  If $K=Q_{xyu_1}$ then $K$ contains $x,y,u_1$ from
$S^*$, so if $|K|\le q-1$ we are done; and if $|K|=q$ then since $p\ge
2$ we have $|Q_{u_1}|\ge 2$, so $Q_{u_1}$ contains a vertex $u'_1$
from $S^*$.  The other cases are symmetric.  Hence the result follows
from Theorem~\ref{thm:tools}~(\ref{fives}).  This completes the proof.
\end{proof}

\subsection{Chromatic bound for bands, belts and boilers}

\begin{theorem}\label{thm:colband}
Let $G$ be a band.  Then $\chi(G)\le
\lceil\frac{5}{4}\omega(G)\rceil$.
\end{theorem}
\begin{proof}
We use the same notation as in the definition of a band  (see also Figure~\ref{fig:bbb}:(b)), and we prove
the theorem by induction on $|V(G)|$.   First suppose that $[R_2,R_3]$
is not complete.  By Lemma~\ref{lem:ab} there exist non-adjacent
vertices $u\in R_2$ and $v\in R_3$ such that every maximal clique in
$G[R_2\cup R_3]$ contains  $u$ or $v$.  If $Q_5\neq\emptyset$, pick
any $w\in Q_5$ and let $S=\{u,v,w\}$; else let $S=\{u,v\}$.  Then   $S$ is a very good stable set of $G$, so the
result follows from Theorem~\ref{thm:tools}~(\ref{goods}).  Therefore
we may assume that $[R_2,R_3]$ is complete.  Now suppose that
$[Q_1,Q_2]$ is not complete.  By Lemma~\ref{lem:ab} there exist
non-adjacent vertices $u\in Q_1$ and $v\in Q_2$ such that every
maximal clique in $G[Q_1\cup Q_2]$ contains  $u$ or $v$.  If
$Q_4\neq\emptyset$, pick any $w\in Q_4$ and let $S=\{u,v,w\}$; else
let $S=\{u,v\}$.  Then   $S$ is a very good
stable set of $G$, so the
result follows from Theorem~\ref{thm:tools}~(\ref{goods}).  Therefore we may assume that $[Q_1,Q_2]$ is
complete, and similarly that $[Q_3,Q_4]$ is complete.  Now $G$ is a
blowup of $C_5$, so the result follows from Corollary~\ref{cor:bu-c5}.
\end{proof}

We say that a graph $G$ is an \emph{extended} $\cal C$-pair if $V(G)$
can be partitioned into three sets $Q,X,A$ such that $(X,A)$ is a
$\cal C$-pair, $Q$ is a clique,  $[Q,X]$ is complete and $[Q,A]=\emptyset$.
\begin{lemma}\label{thm:colxcp}
Let $G$ be an extended $\cal C$-pair. Then $\chi(G)\le
\lceil\frac{5}{4}\omega(G)\rceil$.
\end{lemma}
\begin{proof}
We prove the lemma by induction on $|V(G)|$. Let $V(G)$ be
partitioned into $Q,X,A$ as in the definition above. Let $q=\omega(G)$. If some vertex
$a\in A$ has no neighbor in $X$, then $a$ is simplicial, so  $d(a)<q$, and we can
conclude using Theorem~\ref{thm:tools}(\ref{degq}) and by the
induction hypothesis. Therefore we may assume that every vertex in
$A$ has a neighbor in $X$.

Suppose that $G[X]$ has four pairwise non-adjacent simplicial vertices
$s_1, s_2,$ $s_3,$ $s_4$.  If $d(s_i)\le \lceil\frac{5}{4}q\rceil-1$,
then we can conclude using Theorem~\ref{thm:tools}(\ref{degq}).  So
assume that $d(s_i) \ge \lceil\frac{5}{4}q\rceil$.  We have $N(s_i) =
Q\cup N_X(s_i)\cup N_A(s_i)$, and $Q\cup N_X(s_i)$ is a clique, so we
must have $|N_A(s_i)|\ge \lceil\frac{q}{4}\rceil+1$.  By the
definition of a $\cal C$-pair the sets $N_A(s_1), \ldots, N_A(s_4)$
are pairwise disjoint.  It follows that $|A|\ge 4
(\lceil\frac{q}{4}\rceil+1) >q$, a contradiction.  Hence $G[X]$ has at
most three pairwise non-adjacent simplicial vertices.  If $X$ is a
clique then $G$ is a chordal graph, so $\chi(G)=\omega(G)$ and the
theorem holds trivially.  Therefore we may assume that $G[X]$ has
exactly $k$ pairwise non-adjacent simplicial vertices with
$k\in\{2,3\}$.  Since $G[X]\in {\cal C}$ and by Lemma~\ref{lem:3simp},
we have the following two cases (a) and (b).

(a) $k=2$, so $X$ is partitioned into three cliques $X_1, X_2$ and $U$
such that $X_1, X_2$ are non-empty, $[U, X_1\cup X_2]$ is complete and
$[X_1,X_2]=\emptyset$.  Suppose that $U\neq\emptyset$.  Then
Theorem~\ref{thm:cpair} and the fact that every vertex in $A$ has a
neighbor in $X$ implies that some vertex $u$ in $U$ is universal in
$G$, so $\{u\}$ is a very good stable set and we conclude using
Theorem~\ref{thm:tools}(\ref{goods}).  Hence $U=\emptyset$.  Then $G$
is a band, and we conclude with Theorem~\ref{thm:colband}.

\begin{figure}[h]
\centering
 \includegraphics{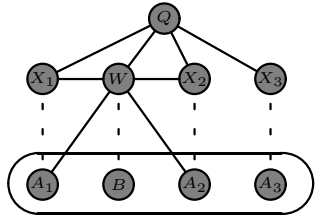}
\caption{Schematic representation of the graph in Case~(b) of Lemma~\ref{thm:colxcp} where $U=\es$.  Here, each shaded circle represents a clique, and the circles inside the oval form a clique,   a solid line between two circles indicates  that the two sets are complete to each other, the absence of line between any two circles indicates that the sets are  anticomplete to each other, and a dashed line between  two circles   indicates that the adjacency between the two sets are arbitrary.}\label{fig:caseB}
\end{figure}

(b) $k=3$, so  $X$ is partitioned
into five cliques $X_1, X_2, X_3, W$ and $U$ such that $X_1,X_2,X_3$
are non-empty and pairwise anticomplete, $[W, X_1\cup X_2]$ is
complete, $[W, X_3]=\emptyset$, and $[U, X\setminus U]$ is complete.
As in case~(a) we may assume that $U=\emptyset$.  By
Theorem~\ref{thm:cpair} and the fact that every vertex in $A$ has a
neighbor in $X$, the set $A$ is partitioned into four sets
$A_1, A_2, A_3, B$ such that $N_A(X_i)=A_i$ for each $i\in\{1,2,3\}$,
$N_A(W)=A_1\cup A_2\cup B$, and $[W, A_1\cup A_2]$ is complete, and
there is no other edge between $X$ and $A$. Moreover, if one of $[X_j, A_j]$ ($j\in\{1,2,3\}$) is not complete,
then $[X_t, A_t]$ is complete for each $t\in \{1,2,3\}\sm \{j\}$.  See Figure~\ref{fig:caseB}.

Suppose that $B\neq \es$. Since every vertex of $A$ has a neighbor in $X$, every vertex of $B$ has a neighbor in $W$. So by Lemma~\ref{lem:ab}, there exists a vertex $w\in W$ such that $[w,B]$ is complete. Hence $w$ is universal in $G[V(G)\sm (X_3\cup A_3)]$. We may assume that $[X_3,A_3]$ is not  complete (otherwise $\{w,x_3\}$, for any $x_3\in X_3$, is a very good stable set of $G$, and we can conclude by using Theorem~\ref{thm:tools}.). Then by Lemma~\ref{lem:ab}, there exist non-adjacent vertices $x_3\in X_3$ and $a_3\in A_3$ such that every maximal clique in $G[X_3\cup A_3]$ contains $x_3$ or $a_3$. Then $\{w,x_3,a_3\}$ is a very good stable set of $G$, and we can conclude by using Theorem~\ref{thm:tools}. So we may assume that $B=\es$.

Suppose that $[X_1,A_1]$ is not complete.  Then, as remarked earlier, $[X_2,A_2]$ and $[X_3,A_3]$ are complete. Also by Lemma~\ref{lem:ab}, there exist non-adjacent vertices $x_1\in X_1$ and $a_1\in A_1$ such that every maximal clique in $G[X_1\cup A_1]$ contains $x_1$ or $a_1$. Pick any $x_2\in X_2$ and $x_3\in X_3$. Then $\{a_1,x_1,x_2,x_3\}$ is a very good stable set of $G$, and we can conclude by using Theorem~\ref{thm:tools}. Therefore assume that $[X_1, A_1]$ is
complete, and, similarly, that $[X_2, A_2]$ is complete.

Suppose that $[X_3,A_3]$ is not complete. Then by Lemma~\ref{lem:ab}, there
are non-adjacent vertices $x_3\in X_3$ and $a_3\in A_3$ such that every maximal
clique in $G[X_3\cup A_3]$ contains $x_3$ or $a_3$.  If
$W\neq\emptyset$, then  any $w\in W$ is universal in $G[V(G)\sm (X_3\cup A_3)]$. But now $\{w,x_3,a_3\}$ is a very good stable set of $G$, and  we can conclude by using Theorem~\ref{thm:tools}. So $W=\es$. Now pick any $x_1\in X_1$ and $x_2\in X_2$. Then $\{x_1,x_2,x_3,a_3\}$ is a very good stable set of $G$, and we can conclude by Theorem~\ref{thm:tools}.  Therefore assume that
$[X_3,A_3]$ is complete.

Now $G$ is a blowup of $F_{2,0}$ (with $A_1\cup A_2$ is the role
of $Q_A$, and $X_3$ in the role of $Q_B$, and $A_3$ in the role of
$Q_z$, and $Q$ in the role of $Q_y$, and $W$ in the role of $Q_x$, and
$X_1,X_2$ in the role of $Q_{u_1}, Q_{u_2}$), so we can conclude using
Theorem~\ref{thm:bu-fkl}. This completes the proof.
\end{proof}

\begin{theorem}\label{thm:colbelt}
Let $G$ be a belt.  Then $\chi(G)\le
\lceil\frac{5}{4}\omega(G)\rceil$.
\end{theorem}
\begin{proof}
We use the same notation as in the definition of a belt, and we will
also use the properties listed in Theorem~\ref{thm:belts}.  We prove
the theorem by induction on $\omega(G)$.  If $\omega(G)=2$ then $G$ is
a $C_5$ and the theorem holds obviously.  Now assume that
$\omega(G)\ge 3$.  Let $q=\omega(G)$.

Suppose that both $R_2, R_3$ are non-empty.  Recall from
Theorem~\ref{thm:belts} that $G[R_2]$ is $(P_4, C_4, 2P_3)$-free,
hence chordal.  Moreover, the axiom that $G[R_2]$ has no universal
vertex implies that $R_2$ is not a clique, so it has two non-adjacent
simplicial vertices $r_1, r_2$.  For each $h\in\{1,2\}$ let $X_h$ be
the closed neighborhood of $r_h$ in $R_2$; so $X_h$ is a clique.  Let
$Y_h=N(r_h)\cap Q_3$.  If $d(r_h)< \lceil\frac{5}{4}q\rceil$ then we
can conclude using Theorem~\ref{thm:tools}~(\ref{degq}) and induction.
Hence assume that $d(r_h)\ge \lceil\frac{5}{4}q\rceil$ for each
$h\in\{1,2\}$.  By the definition of a belt, we have $N[r_h] = Q_1\cup
Q_2\cup X_h\cup Y_h$, and $Q_1\cup Q_2\cup X_h$ is a clique, so we
must have $|Y_h|\ge \lceil\frac{q}{4}\rceil+1$.  By
Theorem~\ref{thm:belts}(\ref{belta}), the sets $Y_1, Y_2$ are pairwise
disjoint.  By the same argument $G[R_3]$ has two non-adjacent
simplicial vertices and consequently there are two disjoint subsets
$Z_1,Z_2$ of $Q_2$ with size at least $\lceil\frac{q}{4}\rceil+1$.  By
Theorem~\ref{thm:belts}(\ref{beltc}) the set $Y_1\cup Y_2\cup Z_1\cup
Z_2$ is a clique, and its size is strictly larger than $q$, a
contradiction.

Therefore we may assume that $R_3=\emptyset$.  Let $X=Q_2\cup R_2\cup
Q_5$ and $A=Q_3\cup Q_4$.  Then the partition of $V(G)$ into $Q_1, X$
and $A$ shows that $G$ is an extended $\cal C$-pair, so the result
follows from Lemma~\ref{thm:colxcp}.
\end{proof}

\begin{theorem}\label{thm:chiboil}
Let $G$ be a boiler.  Then $\chi(G)\le
\lceil\frac{5}{4}\omega(G)\rceil$.
\end{theorem}
\begin{proof}
We use the same definition as in the definition of a boiler.  Let
$q=\omega(G)$.  By Theorem~\ref{thm:tools} we may assume that every
vertex in $G$ has degree at least $\lceil\frac{5}{4}q\rceil$.  If $L$
is a clique, then the partition of $V(G)$ into $Q$, $M\cup A$ and
$L\cup B$ shows that $G$ is an extended $\cal C$-pair, so the result
follows from Lemma~\ref{thm:colxcp}.  Therefore assume that $L$ is
not a clique.  By the same argument as in the proof of
Theorem~\ref{thm:colbelt} implies that there are two disjoint
subsets $A_1, A_2$ of $A$ of size at least
$\lceil\frac{q}{4}\rceil+1$.  By the same argument applied to
$G[M_1\cup M_2]$, there are two disjoint sets $Y_1\subseteq B_1$ and
$Y_2\subseteq B_2$ of size at least $\lceil\frac{q}{4}\rceil+1$.  Then
$A_1\cup A_2\cup B_1\cup B_2$ is a clique, with size  strictly
larger than $q$, a contradiction.
\end{proof}

\subsection{Chromatic bounds for $(P_6, C_4)$-free graphs}

\noindent{\it Proof of Theorem~\ref{thm:54bound}.}~Let $G$ be any
$(P_6, C_4)$-free graph.  We prove the theorem by induction on
$|V(G)|$.

If $G$ has a universal vertex $u$, then $\omega(G)=\omega(G\setminus
u)+1$, and by the induction hypothesis we have
$\chi(G)=\chi(G\setminus u)+1 \le \lceil \frac{5}{4}(\omega(G\setminus
u)\rceil+1$, which implies $\chi(G) \le \lceil
\frac{5}{4}(\omega(G)\rceil$.

If $G$ has a clique cutset $K$, let $A,B$ be a partition of
$V(G)\setminus K$ such that both $A,B$ are non-empty and
$[A,B]=\emptyset$.  Clearly $\chi(G)=\max\{\chi(G[K\cup A]),
\chi(G[K\cup B])\}$, so the desired result follows from the induction
hypothesis on $G[K\cup A]$ and $G[K\cup B]$.

Finally, if $G$ has no universal vertex and no clique cutset, then
the result follows from Theorem~\ref{thm:structure} and
Theorems~\ref{thm:bu-peter}---\ref{thm:chiboil}.  $\Box$

\bigskip
Next we  prove Theorem~\ref{thm:reeds} by using the following theorem.

 \begin{theorem}[\cite{KM-JGT}]\label{thm:xbound:reeds}  If a graph $G$ satisfies $\chi(G) \le \lceil
\frac{5}{4}\omega(G)\rceil$, then it satisfies $\chi(G) \le
\lceil\frac{\Delta(G) + \omega(G) +1}{2}\rceil$.
\end{theorem}

\noindent{\it Proof of Theorem~\ref{thm:reeds}.}~ This follows from Theorems~\ref{thm:54bound} and \ref{thm:xbound:reeds}.  $\Box$

\paragraph{Acknowledgements.}
The first author thanks an  anonymous referee and \linebreak Mathew~C.~Francis  for their comments and suggestions.   The second author passed away while this paper was under review, and the first author would like to dedicate this revised version to his memory.



\begin{thebibliography}{99}
\bibitem{ACHRS-evenhole}
L. Addario-Berry, M. Chudnovsky, F. Havet, B. Reed and P. Seymour,
Bisimplicial vertices in even-hole-free graphs. {\it Journal of Combinatorial
Theory, Series  B} 98 (2008) 1119--1164.

\bibitem{BDHP}
A.~Brandst\"adt, K.~K.~Dabrowski, S.~Huang and D.~Paulusma, Bounding the
clique-width of $H$-free chordal graphs.  {\it Journal of Graph Theory}, 86
(2017) 42--77.
	
\bibitem{BDLM}
A.~Brandst\"adt, F.~F.~Dragan, H.-O.~Le and R.~Mosca, New graph classes of
bounded clique-width.  {\it Theory of Computing  Systems} 38 (2005) 623--645.
	
\bibitem{BH}
A.~Brandst\"adt and C.T.~Ho\`ang,  On clique separators, nearly chordal
graphs, and the maximum weight stable set problem.  {\it Theoretical
Computer  Science} 389 (2007) 295--306.

\bibitem{Brooks}
R. L. Brooks, On colouring the nodes of a network. {\it
Proceedings of Cambridge Philosophical Society} 37 (1941) 194--197.

\bibitem{CCH}
K.~Cameron, S.~Chaplick and C.T.~Ho\`ang,  On the structure of (pan, even
hole)-free graphs.  {\it Journal of Graph Theory} 87 (2018) 108--129.

\bibitem{CK}S.~A. Choudum and T. Karthick, Maximal cliques in $\{P_2\cup P_3, C_4\}$-free graphs. {\it Discrete Mathematics} 310 (2010) 3398--3403.

\bibitem{CKS}
S. A. Choudum, T. Karthick and M. A. Shalu,   Perfect coloring and
linearly $\chi$-bounded $P_6$-free graphs.  {\it Journal of Graph
Theory} {54} (2007) 293--306.

\bibitem{CER}
B. Courcelle, J. Engelfriet and G. Rozenberg,  Handle-rewriting
hypergraph grammars.  {\it Journal of Computer and System Sciences}
46 (1993) 218--270.

\bibitem{CAO-quasiline}
M. Chudnovsky and A.~Ovetsky,  Coloring quasi-line graphs.  {\it Journal
of Graph Theory} 54 (2007) 41--50.

\bibitem{CLMTV}
M. Chudnovsky, I. Lo, F. Maffray, N. Trotignon and K.Vu\v skovi\'c,
Coloring square-free Berge graphs.  {\it Journal of Combinatorial Theory, Series B} (2018), available online.

\bibitem{CPST}
M. Chudnovsky, I. Penev, A. Scott and N. Trotignon,  Substitution and
$\chi$-boundedness.  {\it Journal of Combinatorial Theory, Series  B} 103 (2013)
567--586.

\bibitem{CP-claw}
M. Chudnovsky and P. Seymour,  Claw-free Graphs VI. Coloring claw-free
graphs.  {\it Journal of Combinatorial Theory, Series B} 100 (2010)
560--572.

\bibitem{GH}
S.~Gaspers and S.~Huang,  Linearly $\chi$-bounding $(P_6,C_4)$-free
graphs.  Extended Abstract In: Bodlaender H., Woeginger G. (eds)
Graph-Theoretic Concepts in Computer Science.  WG 2017.  {\it Lecture
Notes in Computer Science} 10520 (2017) 263--274.  Available on
https://arxiv.org/abs/1709.09750.

\bibitem{GHP}
S.~Gaspers, S.~Huang, D.~Paulusma,  Colouring square-free graphs
without long induced paths.  35th Symposium on Theoretical Aspects of
Computer Science (STACS 2018) : February 28--March 3, 2018, Caen,
France.  http://dro.dur.ac.uk/24230/

\bibitem{Gol78}
M.C.~Golumbic, Trivially perfect graphs.  {\it Discrete Mathematics}
24 (1978) 105--107.

\bibitem{Gol80}
M.C.~Golumbic,  {\it Algorithmic Graph Theory and Perfect Graphs.} 1st
edition, Academic Press, New York, 1980.  2nd edition, Annals of
Discrete Mathematics 57, Elsevier, 2004.

\bibitem{GHM}
S. Gravier, C. T. Ho\`ang and F. Maffray,  Coloring the hypergraph of
maximal cliques of a graph with no long path.  {\it Discrete
Mathematics} 272 (2003) 285--290.

\bibitem{Gyarfas}
A. Gy\'{a}rf\'{a}s,  Problems from the world surrounding perfect
graphs.  {\it Zastosowania Matematyki Applicationes Mathematicae} {19}
(1987) 413--441.

\bibitem{KM-JGT}
T. Karthick and F. Maffray,  Coloring (gem, co-gem)-free graphs.  {\it
Journal of Graph Theory,} 89 (2018) 288--303.

\bibitem{KM-arxiv}
T. Karthick and F. Maffray,  Square-free graphs with no six-vertex
induced~path. Available on arXiv:1805.05007 [cs.DM] (2018).


\bibitem{KP}
H. A. Kierstead and S.G. Penrice,  Radius two trees specify $\chi$-bounded
classes.  {\it Journal of Graph Theory} 18 (1994) 119--129.

\bibitem{KPT-P5}
H. A. Kierstead, S.G. Penrice and W.T. Trotter, On-line and first-fit
coloring of graphs that do not induce $P_5$.  {\it SIAM Journal of
Discrete Mathematics} 8 (1995) 485--498.

\bibitem{KZ}
H.A. Kierstead and Y. Zhu,  Radius three trees in graphs with large
chromatic number.  {\it SIAM Journal of Discrete Mathematics} 17
(2004) 571--581.

\bibitem{King}
A.D. King, Claw-free graphs and two conjectures on omega, Delta, and
chi, Ph.D. Thesis, McGill University, 2009.

\bibitem{KobRot}
D. Kobler and U. Rotics,  Edge dominating set and colorings on graphs
with fixed clique-width.  {\it Discrete Applied Mathematics} 126
(2003) 197--221.

\bibitem{Lovasz}
L.~Lov\'asz,  A characterization of perfect graphs.  {\it Journal of
Combinatorial Theory, Series B} 13 (1972) 95--98.

\bibitem{MReed}
F. Maffray and B. A. Reed,  A description of claw-free perfect graphs.
{\it Journal of Combinatorial Theory, Series  B} 75 (1999) 134--156.

\bibitem{MR}
J. A. Makowsky and U. Rotics,  On the clique-width of graphs with few
$P_4$'s.  {\it International Journal of Foundations of Computer
Science} 10 (1999) 329--348.

\bibitem{Mos}
R.~Mosca, Stable sets in certain $P_6$-free graphs. {\it Discrete
Applied Mathematics} 92 (1999) 177--191.

\bibitem{Rabern}
L. Rabern,  A note on Reed's conjecture.  {\it SIAM Journal on
Discrete Mathematics} 22 (2008) 820--827.

\bibitem{Rao}
M.~Rao, MSOL partitioning problems on graphs of bounded treewidth and
clique-width.  {\it Theoretical  Computer  Science} 377 (2007) 260--267.

\bibitem{Reed}
B. Reed,  $\omega, \Delta$ and $\chi$.  {\it Journal of Graph Theory}
{27} (1998) 177--212.

\bibitem{SS}
A. Scott and P. Seymour, Induced subgraphs of graphs with large chromatic
number.  I. Odd holes.  {\it Journal of Combinatorial Theory, Series  B} 121
(2016) 68–-84.

\bibitem{Tar}
R. E. Tarjan,  Decomposition by clique separators.  {\it Discrete
Mathematics} 55 (1985) 221--232.

\end{thebibliography}
\end{document}